\documentclass[12pt]{article}
\usepackage[letterpaper,top=2cm,bottom=2cm,left=3cm,right=3cm,marginparwidth=1.75cm]{geometry}
\usepackage{natbib}
\usepackage[toc,page]{appendix}
\usepackage{hhline}
\usepackage[utf8]{inputenc}
\usepackage[T1]{fontenc}
\usepackage{amsmath}
\usepackage{rotating}
\usepackage{multirow}
\usepackage{algorithm}
\usepackage{algpseudocode}
\usepackage{fbox}
\usepackage{physics}
\usepackage{mathtools}
\usepackage{graphicx}
\usepackage{float}
\usepackage{amssymb}
\usepackage{listings} 
\usepackage{amsfonts}
\usepackage{bbm}
\usepackage{listings}
\usepackage{longtable}
\usepackage[colorlinks=true, allcolors=blue]{hyperref}
\usepackage{multirow}
\usepackage{amsthm}
\usepackage{caption}
\usepackage{subcaption}
\usepackage{xcolor}
\newtheorem{theorem}{Theorem}[section]
\newtheorem{assumption}{Assumption}
\newtheorem{proposition}[theorem]{Proposition}

\newtheorem{lemma}[theorem]{Lemma}
\theoremstyle{definition}
\newtheorem{definition}{Definition}[section]

\definecolor{codegreen}{rgb}{0,0.6,0}
\definecolor{codegray}{rgb}{0.5,0.5,0.5}
\definecolor{codepurple}{rgb}{0.58,0,0.82}
\definecolor{backcolour}{rgb}{0.95,0.95,0.92}

\lstdefinestyle{mystyle}{
    backgroundcolor=\color{backcolour},   
    commentstyle=\color{codegreen},
    keywordstyle=\color{magenta},
    numberstyle=\tiny\color{codegray},
    stringstyle=\color{codepurple},
    basicstyle=\ttfamily\footnotesize,
    breakatwhitespace=false,         
    breaklines=true,                 
    captionpos=b,                    
    keepspaces=true,                 
    numbers=left,                    
    numbersep=5pt,                  
    showspaces=false,                
    showstringspaces=false,
    showtabs=false,                  
    tabsize=2
}

\date{}

\begin{document}


{
  \title{\bf Iterative Methods for Full-Scale Gaussian Process Approximations for Large Spatial Data}
  \author{Tim Gyger\footnotemark[2]
\thanks{Institute of Financial Services Zug, Lucerne University of Applied Sciences and Arts (tim.gyger@hslu.ch)}
  \and
Reinhard Furrer\thanks{Department of Mathematical Modeling and Machine Learning, University of Zurich}
  \and
  Fabio Sigrist\thanks{Seminar for Statistics, ETH Zurich} \footnotemark[1]
  }
  \maketitle
}

\bigskip
\begin{abstract}
Gaussian processes are flexible probabilistic regression models which are widely used in statistics and machine learning. However, a drawback is their limited scalability to large data sets. To alleviate this, full-scale approximations (FSAs) combine predictive process methods and covariance tapering, thus approximating both global and local structures. We show how iterative methods can be used to reduce computational costs in calculating likelihoods, gradients, and predictive distributions with FSAs. In particular, we introduce a novel preconditioner and show theoretically and empirically that it accelerates the conjugate gradient method's convergence speed and mitigates its sensitivity with respect to the FSA parameters and the eigenvalue structure of the original covariance matrix, and we demonstrate empirically that it outperforms a state-of-the-art pivoted Cholesky preconditioner. Furthermore, we introduce an accurate and fast way to calculate predictive variances using stochastic simulation and iterative methods. In addition, we show how our newly proposed fully independent training conditional (FITC) preconditioner can also be used in iterative methods for Vecchia approximations. In our experiments, it outperforms existing state-of-the-art preconditioners for Vecchia approximations. All methods are implemented in a free C++ software library with high-level Python and R packages (\url{https://github.com/fabsig/GPBoost}).

\end{abstract}

\section{Introduction}\label{intro}
Gaussian process (GP) models are frequently used as probabilistic nonparametric models. Calculating likelihoods and their gradients is typically done using a Cholesky factorization requiring $\mathcal{O}(n^3)$ operations and $\mathcal{O}(n^2)$ memory allocation, where $n$ is the number of observations. Computations for large data sets thus quickly become infeasible. A variety of strategies have been introduced to alleviate this computational burden; see, e.g., \citet{heaton2019case} for a review. Low-rank approximations such as modified predictive process models also known as fully independent training conditional (FITC) approximations \citep{quinonero2005unifying, banerjee2008gaussian, finley2009improving} are good at modeling global large-scale structures of Gaussian processes. In contrast, there are local approaches that rely on sparse linear algebra methods, where sparsity is enforced on either the precision or covariance matrices. Concerning the former, Vecchia \citep{vecchia1988estimation, datta2016hierarchical, katzfuss2021general} and SPDE-based \citep{lindgren2011explicit} approximations are popular in spatial statistics. Another method to achieve sparsity is covariance tapering, which sets the covariance to zero between two sufficiently distant locations \citep{furrer2006covariance}. Tapering is typically effective in capturing short-scale structures of spatial processes. The full-scale approximation (FSA) of \citet{sang2012full} combines a low-rank approximation with tapering applied to a residual process thus capturing both long- and short-range correlations. However, the computational complexity for factorizing a sparse $n\times n$ matrix is not linear in $n$ but difficult to generalize since it depends on the corresponding graph and the fill-reducing ordering \citep{davis2006direct}. We denote this computational complexity by $\mathcal{O}\big(g(n)\big)$, where $g(\cdot)$ is a function of the sample size $n$. For instance, we have approximately $\mathcal{O}(n^{3/2})$ for two-dimensional spatial data and at least $\mathcal{O}(n^{2})$ in higher dimensions \citep{lipton1979generalized}. For evaluating a log-likelihood and its gradient, the FSA with tapering involves the factorization of a sparse matrix, resulting in $\mathcal{O}\big(n\cdot (m^2 + n_\gamma^2) + g(n)\big)$ computational and $\mathcal{O}\big(n\cdot (m + n_\gamma)\big)$ memory complexity when using a Cholesky decomposition, where $n_\gamma$ is the average number of nonzero entries per row in the sparse matrix, and $m$ is the number of inducing points or knots. Further, the computational complexity for calculating predictive variances is $\mathcal{O}\big(n\cdot n_p \cdot n_\gamma\big)$, where $n_p$ denotes the number of prediction locations; i.e., both parameter estimation and prediction can become prohibitively slow for large $n$ and $n_p$ due to the above $n^{3/2}$ and  $n\cdot n_p$ terms. Iterative numerical techniques \citep{trefethen2022numerical}, such as the conjugate gradient (CG) method and the Lanczos tridiagonalization algorithm, which require only matrix-vector multiplications that can be trivially parallelized, are an alternative to direct solver methods such as the Cholesky decomposition. Recently, a growing body of research \citep{aune2014parameter,gardner2018gpytorch,pleiss2018constant,geoga2020scalable} has used these numerical approaches in combination with stochastic approximation methods such as the Hutchinson's estimator \citep{hutchinson1989stochastic} and the stochastic Lanczos quadrature (SLQ) \citep{ubaru2017fast}. Other popular GP approximations for spatial data include multiresolution approximations \citep{katzfuss2017multi} and approximations based on hierarchical matrices \citep{ambikasaran2015fast, abdulah2018parallel, litvinenko2020hlibcov, geoga2020scalable}. 

In this work, we propose iterative methods for likelihood evaluation, gradient calculation, and the computation of predictive distributions for full-scale approximations that combine inducing points and covariance tapering. We solve systems of linear equations using the CG method and calculate log-determinants and their derivatives using SLQ and Hutchinson's estimator. Furthermore, we introduce a novel way to compute predictive variances through a combination of the CG method and stochastic diagonal estimation \citep{bekas2007estimator}. We also introduce a preconditioner, denoted as FITC preconditioner, that allows for reducing run times and for doing variance reduction. In addition to the full-scale approximation with tapering, this preconditioner is also applicable to other GP approximations and for exact calculations with GPs. In this article, we show how it can be used for Vecchia approximations. We analyze our methods both theoretically and empirically on simulated and real-world data. 

The remainder of the paper is organized as follows. Section \ref{sect2} introduces the full-scale approximation with tapering applied to a residual process. For notational simplicity, we refer to this approximation as FSA in this article. In Section \ref{sect3}, we introduce iterative methods for the FSA and state convergence results for the CG method and its preconditioned variant. In Section \ref{sect4}, we conduct simulation studies to compare different inducing point methods and different combinations of approximation parameters $m$ and $\gamma$. Furthermore, we analyze the computational efficiency and statistical properties of the proposed methodology. In Section \ref{sectVecchia}, we show how our proposed FITC preconditioner can be used in iterative methods for Vecchia approximations, and we compare it to other state-of-the-art preconditioners. In Section \ref{sect5}, we apply our methodology to a large spatial data set comprising daytime land surface temperatures. 

\section{Gaussian process models for spatial data sets}\label{sect2}

We assume the following GP model:
\begin{align}
    \boldsymbol{Y}=F(\boldsymbol{X})+\boldsymbol{b}+\boldsymbol{\epsilon}, \quad \boldsymbol{b} \sim \mathcal{N}(\boldsymbol{0}, \mathbf{\Sigma}), \quad \boldsymbol{\epsilon} \sim \mathcal{N}\left(\boldsymbol{0}, \sigma^2 \boldsymbol{I}_n\right),  \label{Model}  
\end{align}
where $\boldsymbol{Y}=\left(Y(\boldsymbol{s}_1), \ldots, Y(\boldsymbol{s}_n)\right)^{\mathrm{T}} \in \mathbb{R}^n$ represents the response variable observed at locations $\mathcal{S} = \left\{\boldsymbol{s}_1, \ldots, \boldsymbol{s}_n\right\}$, $\boldsymbol{s}_i\in \mathbb{R}^d$, and $F(\boldsymbol{X}) \in \mathbb{R}^n$ denotes the fixed effects. In the following, we assume $F(\boldsymbol{X})=\boldsymbol{X}\boldsymbol{\beta},$ $\boldsymbol{\beta} \in \mathbb{R}^p$, $\boldsymbol{X} \in \mathbb{R}^{n \times p}$, but $F(\cdot)$ could also be modeled using machine learning methods such as tree-boosting \citep{sigrist2022gaussian, sigrist2022latent}. Further, $\boldsymbol{\epsilon}=\left({\epsilon}(\boldsymbol{s}_1), \ldots, {\epsilon}(\boldsymbol{s}_n)\right)^{\mathrm{T}} \in \mathbb{R}^n$ is an independent error term often referred to as the nugget effect. The zero-mean random effects $\boldsymbol{b} = \left(b(\boldsymbol{s}_1), \ldots, b(\boldsymbol{s}_n)\right)^{\mathrm{T}} \in \mathbb{R}^n$ are a finite-dimensional version of a GP, and the covariance matrix is denoted as $\boldsymbol{\Sigma} \in \mathbb{R}^{n \times n}$, where its elements are determined by a parametric covariance function $c\left(\cdot, \cdot\right)$:
\begin{align*}
    \Sigma_{ij} =\operatorname{Cov}\big(b(\boldsymbol{s}_i), b(\boldsymbol{s}_j)\big)=c\left(\boldsymbol{s}_i, \boldsymbol{s}_j\right), \quad \boldsymbol{s}_i, \boldsymbol{s}_j \in \mathcal{S}.
\end{align*}

The regression coefficients $\boldsymbol{\beta}$ and the parameters characterizing the covariance matrix denoted by $\boldsymbol{\theta} \in \mathbb{R}^{q}$ are usually determined using maximum likelihood estimation. For this, we minimize the negative log-likelihood function 
 \begin{align}
 \begin{split}
\mathcal{L}(\mathbf{\boldsymbol{\beta}},\boldsymbol{\theta};\boldsymbol{y},\boldsymbol{X})&= \frac{n}{2} \log (2 \pi)+\frac{1}{2}\log \det\big(\Tilde{\mathbf{\Sigma}}\big) +\frac{1}{2}(\boldsymbol{y}-\boldsymbol{X}\mathbf{\boldsymbol{\beta}})^{\mathrm{T}}\Tilde{\mathbf{\Sigma}}^{-1}(\boldsymbol{y}-\boldsymbol{X}\mathbf{\boldsymbol{\beta}}),
\end{split}\label{NEGLL} 
\end{align}
 where $\Tilde{\mathbf{\Sigma}} = {\mathbf{\Sigma}} + \sigma^2 \boldsymbol{I}_n$ is the covariance matrix of $\boldsymbol{Y}$, and $\boldsymbol{y}$ is the observed data. If a first- or second-order method for convex optimization, such as the Broyden–Fletcher–\\Goldfarb–Shanno (BFGS) algorithm, is used, the calculation of the gradient of the negative log-likelihood function with respect to the covariance parameters is required. This gradient is given by
\begin{align}
\frac{\partial}{\partial \boldsymbol{\theta}}\mathcal{L}(\mathbf{\boldsymbol{\beta}},\boldsymbol{\theta};\boldsymbol{y},\boldsymbol{X}) &= \frac{1}{2}\Tr \Big(\Tilde{\mathbf{\Sigma}}^{-1}\frac{\partial \Tilde{\mathbf{\Sigma}}}{\partial \boldsymbol{\theta}}\Big)  -\frac{1}{2}(\boldsymbol{y}-\boldsymbol{X}\mathbf{\boldsymbol{\beta}})^{\mathrm{T}}\Tilde{\mathbf{\Sigma}}^{-1}  \frac{\partial \Tilde{\mathbf{\Sigma}}}{\partial \boldsymbol{\theta}} \Tilde{\mathbf{\Sigma}}^{-1}(\boldsymbol{y}-\boldsymbol{X}\mathbf{\boldsymbol{\beta}}).\label{DNEGLL} 
\end{align}

The predictive distribution at $n_p$ new locations $\mathcal{S}^p = \{\boldsymbol{s}^p_1,...,\boldsymbol{s}^p_{n_p}\}$ is $\mathcal{N}(\boldsymbol{\mu}^p,\mathbf{\Sigma}^p)$ with predictive mean $\boldsymbol{\mu}^p=\boldsymbol{X}^p{\boldsymbol{\beta}}+{\mathbf{\Sigma}}^{\mathrm{T}}_{n{n_p}}\Tilde{\mathbf{\Sigma}}^{-1}(\boldsymbol{y}-{\boldsymbol{X}}{\boldsymbol{\beta}})$ and predictive covariance $\mathbf{\Sigma}^p={\mathbf{\Sigma}}_{{n_p}} + \sigma^2 \boldsymbol{I}_{n_p}-{\mathbf{\Sigma}}_{n{n_p}}^{\mathrm{T}}\Tilde{\mathbf{\Sigma}}^{-1} {\mathbf{\Sigma}}_{n{n_p}}\in\mathbb{R}^{{n_p}\times {n_p}}$, where $\boldsymbol{X}^p_i=\big({X}(\boldsymbol{s}^p_i)_{ 1}, \ldots, {X}(\boldsymbol{s}^p_i)_{p}\big) \in \mathbb{R}^{1\times p}$ is the $i$-th row of $\boldsymbol{X}^p$ containing predictor variables for prediction $i$, $i=1, \ldots, {n_p}$, ${\mathbf{\Sigma}}_{n{n_p}}=\left[c\left(\boldsymbol{s}_i, \boldsymbol{s}^p_j\right)\right]_{i=1:n, j=1:{n_p}} \in\mathbb{R}^{n\times {n_p}}$ is a cross-covariance matrix, and ${\mathbf{\Sigma}}_{{n_p}}=\left[c\left(\boldsymbol{s}^p_i, \boldsymbol{s}^p_j\right)\right]_{i=1:{n_p}, j=1:{n_p}} \in\mathbb{R}^{{n_p}\times {n_p}}$. 


\subsection{Full-Scale approximation}\label{sectFSA} 

The idea of the FSA \citep{sang2012full} is to decompose the GP into two parts 
\begin{align*}
\boldsymbol{b}=\boldsymbol{b}_\mathrm{l}+\boldsymbol{b}_{\mathrm{s}},
\end{align*}
where $\boldsymbol{b}_\mathrm{l}$ is a reduced rank predictive process modeling large-scale dependence and $\boldsymbol{b}_{\mathrm{s}} = \boldsymbol{b} - \boldsymbol{b}_\mathrm{l}$ is a residual process capturing the small-scale spatial dependence that is unexplained by the reduced rank process. For a fixed set $\mathcal{S}^* = \{\boldsymbol{s}_1^*,...,\boldsymbol{s}_m^*\}$ of $m$ distinct inducing points, or knots, the predictive process is given by
\begin{align*}
\boldsymbol{b}_\mathrm{l}(\boldsymbol{s}_i)=c\left(\boldsymbol{s}_i, \mathcal{S}^*\right)^{\mathrm{T}}\mathbf{\Sigma}_m^{-1} \boldsymbol{b}^*,
\end{align*}
where $\boldsymbol{b}^*=(b(\boldsymbol{s}_1^*),...,b(\boldsymbol{s}_m^*))^{\mathrm{T}}$. Its finite rank covariance function is
\begin{align*}
c_\mathrm{l}(\boldsymbol{s}_i,\boldsymbol{s}_j)=\operatorname{Cov}\big(b_\mathrm{l}(\boldsymbol{s}_i), b_\mathrm{l}(\boldsymbol{s}_j)\big)=c(\boldsymbol{s}_i, \mathcal{S}^*)^{\mathrm{T}}\mathbf{\Sigma}_m^{-1}c\left(\boldsymbol{s}_j, \mathcal{S}^*\right),
\end{align*}
where $\boldsymbol{c}(\boldsymbol{s}_i, \mathcal{S}^*) = \big(c(\boldsymbol{s}_i, \boldsymbol{s}^*_1),...,c(\boldsymbol{s}_i, \boldsymbol{s}^*_m)\big)^{\mathrm{T}}$, and $\mathbf{\Sigma}_m = \big[c(\boldsymbol{s}_i^*,\boldsymbol{s}_j^*)\big]_{i=1:m, j=1:m}\in\mathbb{R}^{m\times m}$ is a covariance matrix with respect to the inducing points. The corresponding covariance matrix is given by
\begin{align*}
    \mathbf{\Sigma}_{\mathrm{l}}= \mathbf{\Sigma}_{mn}^{\mathrm{T}}\mathbf{\Sigma}_{m}^{-1}\mathbf{\Sigma}_{mn},
\end{align*}
where $\mathbf{\Sigma}_{mn} = \big[c(\boldsymbol{s}_i^*, \boldsymbol{s}_j)\big]_{i=1:m, j=1:n}\in\mathbb{R}^{m\times n}$ is a cross-covariance matrix between the inducing and data points. The covariance function of the residual process $\boldsymbol{b}_{\mathrm{s}}(\boldsymbol{s})$ is
\begin{align*}
    c(\boldsymbol{s}_i,\boldsymbol{s}_j)-\boldsymbol{c}(\boldsymbol{s}_i, \mathcal{S}^*)^{\mathrm{T}}\mathbf{\Sigma}^{-1}_m \boldsymbol{c}(\boldsymbol{s}_j, \mathcal{S}^*).
\end{align*}
In the FSA, this is approximated using tapering
\begin{align*}
c_{\mathrm{s}}(\boldsymbol{s}_i,\boldsymbol{s}_j)=\big(c(\boldsymbol{s}_i,\boldsymbol{s}_j)-\boldsymbol{c}(\boldsymbol{s}_i, \mathcal{S}^*)^{\mathrm{T}}\mathbf{\Sigma}^{-1}_m \boldsymbol{c}(\boldsymbol{s}_j, \mathcal{S}^*)\big) \cdot c_{\text {taper }}(||\boldsymbol{s}_i-\boldsymbol{s}_j||_2 ; \gamma),
\end{align*}
where the tapering function $c_{\text {taper }}(\cdot; \gamma)$ is an isotropic correlation function which is zero for data located in two positions whose distance exceeds the prescribed taper range parameter $\gamma>0$, i.e., $c_{\text {taper }}(||\boldsymbol{s}_i-\boldsymbol{s}_j||_2; \gamma) = 0$ if $||\boldsymbol{s}_i-\boldsymbol{s}_j||_2 \geq \gamma$, for the Euclidean vector norm $||\cdot||_2$. The covariance matrix of the residual process is given by
\begin{align*}
\mathbf{\Sigma}_{\mathrm{s}} = (\mathbf{\Sigma} - \mathbf{\Sigma}_{mn}^{\mathrm{T}}\mathbf{\Sigma}_{m}^{-1}\mathbf{\Sigma}_{mn})\circ \mathbf{T}(\gamma),
\end{align*} 
where $\mathbf{T}(\gamma) = \big[c_{\text {taper }}(||\boldsymbol{s}_i-\boldsymbol{s}_j||_2; \gamma)\big]_{i=1:n, j=1:n}$ is the taper matrix, and the operator $\circ$ refers to the elementwise matrix product also called the Hadamard product. The parameter $\gamma$ determines the sparsity of the matrix $\mathbf{\Sigma}_{\mathrm{s}}$, and it controls a trade-off between accuracy and computational cost. In summary, the covariance function of the FSA is given by
\begin{align*}
  c_{\dagger}(\boldsymbol{s}_i, \boldsymbol{s}_j)=c_{\mathrm{l}}(\boldsymbol{s}_i,\boldsymbol{s}_j)+c_{\mathrm{s}}(\boldsymbol{s}_i,\boldsymbol{s}_j) \approx c(\boldsymbol{s}_i, \boldsymbol{s}_j),
\end{align*} 
resulting in the covariance matrix approximation
\begin{align*}
\Tilde{\mathbf{\Sigma}}_{\dagger}= \mathbf{\Sigma}_{\mathrm{l}}+\mathbf{\Sigma}_{\mathrm{s}}+\sigma^2 \boldsymbol{I}_n \approx \Tilde{\mathbf{\Sigma}}.
\end{align*} 

\subsection{Parameter estimation}

For the full-scale approximation, the negative log-likelihood and its derivative are given by
\begin{align}
\begin{split}
    \mathcal{L}_\dagger(\mathbf{\boldsymbol{\beta}},\boldsymbol{\theta};\boldsymbol{y},\boldsymbol{X}) &= \frac{n}{2} \log (2 \pi) + \frac{1}{2}\log \det\big(\Tilde{\mathbf{\Sigma}}_{\dagger}\big) + \frac{1}{2}(\boldsymbol{y}-\boldsymbol{X}\mathbf{\boldsymbol{\beta}})^{\mathrm{T}}\Tilde{\mathbf{\Sigma}}^{-1}_{\dagger}(\boldsymbol{y}-\boldsymbol{X}\mathbf{\boldsymbol{\beta}}),\\
\frac{\partial}{\partial \boldsymbol{\theta}}\mathcal{L}_\dagger(\mathbf{\boldsymbol{\beta}},\boldsymbol{\theta};\boldsymbol{y},\boldsymbol{X})&=  \frac{1}{2}\Tr \Big(\Tilde{\mathbf{\Sigma}}_{\dagger}^{-1}\frac{\partial \Tilde{\mathbf{\Sigma}}_{\dagger}}{\partial \boldsymbol{\theta}}\Big) -\frac{1}{2}(\boldsymbol{y}-\boldsymbol{X}\mathbf{\boldsymbol{\beta}})^{\mathrm{T}}\Tilde{\mathbf{\Sigma}}_\dagger^{-1}  \frac{\partial \Tilde{\mathbf{\Sigma}}_{\dagger}}{\partial \boldsymbol{\theta}} \Tilde{\mathbf{\Sigma}}_{\dagger}^{-1}(\boldsymbol{y}-\boldsymbol{X}\mathbf{\boldsymbol{\beta}}),
\end{split}\label{NEGLLDNEGLLFSA}
\end{align}
where $\frac{\partial \Tilde{\mathbf{\Sigma}}_{\dagger}}{\partial \boldsymbol{\theta}} = \frac{\partial \Tilde{\mathbf{\Sigma}}_{\mathrm{s}}}{\partial \boldsymbol{\theta}} + \frac{\partial {\mathbf{\Sigma}}_{\mathrm{l}}}{\partial \boldsymbol{\theta}}$ with 
\begin{align*}
\frac{\partial \Tilde{\mathbf{\Sigma}}_{\mathrm{s}}}{\partial \boldsymbol{\theta}} &= \Big(\frac{\partial {\mathbf{\Sigma}}}{\partial \boldsymbol{\theta}}-\frac{\partial {\mathbf{\Sigma}}_{\mathrm{l}}}{\partial \boldsymbol{\theta}}\Big)\circ\boldsymbol{T}(\gamma) + \frac{\partial {\sigma^2}}{\partial \boldsymbol{\theta}}\boldsymbol{I}_{n},\\
    \frac{\partial {\mathbf{\Sigma}}_{\mathrm{l}}}{\partial \boldsymbol{\theta}} &= \frac{\partial {\mathbf{\Sigma}}_{mn}^\mathrm{T}}{\partial \boldsymbol{\theta}}\mathbf{\Sigma}_{m}^{-1}\mathbf{\Sigma}_{mn} + \mathbf{\Sigma}_{mn}^\mathrm{T}\mathbf{\Sigma}_{m}^{-1}\frac{\partial {\mathbf{\Sigma}}_{mn}}{\partial \boldsymbol{\theta}} - \mathbf{\Sigma}_{mn}^\mathrm{T}\mathbf{\Sigma}_{m}^{-1}  \frac{\partial {\mathbf{\Sigma}}_{m}}{\partial \boldsymbol{\theta}}\mathbf{\Sigma}_{m}^{-1}\mathbf{\Sigma}_{mn},
\end{align*}
and $\Tilde{\mathbf{\Sigma}}_{\mathrm{s}}$ denotes the sparse matrix $\mathbf{\Sigma}_{\mathrm{s}}+\sigma^2 \boldsymbol{I}_n$.
The evaluation of the negative log-likelihood and its derivatives requires the solution of linear equations systems and the determinant of the
$n\times n$ matrix $\Tilde{\mathbf{\Sigma}}_{\dagger}$.
Applying the Sherman-Woodbury-Morrison formula, we obtain
\begin{align}
\begin{split}
    \Tilde{\mathbf{\Sigma}}_{\dagger}^{-1} &= (\Tilde{\mathbf{\Sigma}}_{\mathrm{s}}+\mathbf{\Sigma}_{mn}^{\mathrm{T}}\mathbf{\Sigma}_{m}^{-1}\mathbf{\Sigma}_{mn})^{-1}=\Tilde{\mathbf{\Sigma}}_{\mathrm{s}}^{-1}-\Tilde{\mathbf{\Sigma}}_{\mathrm{s}}^{-1} \mathbf{\Sigma}_{mn}^{\mathrm{T}}(\mathbf{\Sigma}_{m}+\mathbf{\Sigma}_{mn} \Tilde{\mathbf{\Sigma}}_{\mathrm{s}}^{-1} \mathbf{\Sigma}_{mn}^{\mathrm{T}})^{-1} \mathbf{\Sigma}_{mn} \Tilde{\mathbf{\Sigma}}_{\mathrm{s}}^{-1}.
\end{split}\label{SWMFSA}
\end{align}
Moreover, by Sylvester's determinant theorem, the determinant of $\Tilde{\mathbf{\Sigma}}_{\dagger}$ can be written as
\begin{align}
\begin{split}
    \det\big(\Tilde{\mathbf{\Sigma}}_{\dagger}\big) &= \det\big(\mathbf{\Sigma}_{m}+\mathbf{\Sigma}_{mn}\Tilde{\mathbf{\Sigma}}_{\mathrm{s}}^{-1}\mathbf{\Sigma}_{mn}^{\mathrm{T}}\big)\cdot\det\big(\mathbf{\Sigma}_{m}\big)^{-1}\cdot\det\big(\Tilde{\mathbf{\Sigma}}_{\mathrm{s}}\big).
    \end{split}\label{SWMDet}
\end{align}

The computational complexity associated with the computation of the negative log-likelihood and its derivatives is of order $\mathcal{O}\left(n\cdot (m^2 + n_\gamma^2) + g(n)\right)$, where we recall that $g(n) \approx n^{3/2}$ for spatial data, see Section \ref{intro}, and the required storage is of order $\mathcal{O}\big(n\cdot(m+n_\gamma)\big)$, where $m$ represents the number of inducing points, and $n_\gamma$ is the average number of nonzero entries per row in $\Tilde{\boldsymbol{\Sigma}}_{\mathrm{s}}$. The term $\mathcal{O}(n\cdot m^2)$ in the computational complexity arises from computing the matrix $\mathbf{\Sigma}_{m}+\mathbf{\Sigma}_{mn} \Tilde{\mathbf{\Sigma}}_{\mathrm{s}}^{-1} \mathbf{\Sigma}_{mn}^{\mathrm{T}}$ and its determinant, while the term $\mathcal{O}(n\cdot n_\gamma^2)$ results from calculating $\Tr \Big(\Tilde{\mathbf{\Sigma}}_\mathrm{s}^{-1}\frac{\partial \Tilde{\mathbf{\Sigma}}_\mathrm{s}}{\partial \boldsymbol{\theta}}\Big)$.

\subsection{Predictive distribution}

The predictive distribution for $n_p$ new locations given the FSA is 
$\mathcal{N}\big(\boldsymbol{\mu}_\dagger^p,\mathbf{\Sigma}^p_\dagger\big),$
where
\begin{align*}
\boldsymbol{\mu}^p_\dagger&= \boldsymbol{X}^p\mathbf{\boldsymbol{\beta}} + ({\mathbf{\Sigma}}_{n{n_p}}^\dagger)^{\mathrm{T}}\Tilde{\mathbf{\Sigma}}_\dagger^{-1}(\boldsymbol{y}-\boldsymbol{X}\mathbf{\boldsymbol{\beta}}), \\
\mathbf{\Sigma}^p_\dagger&={\mathbf{\Sigma}}_{{n_p}}^\dagger + \sigma^2 \boldsymbol{I}_{n_p}-({\mathbf{\Sigma}}_{n{n_p}}^\dagger)^{\mathrm{T}}\Tilde{\mathbf{\Sigma}}^{-1}_\dagger{\mathbf{\Sigma}}_{n{n_p}}^\dagger,
\end{align*}
and ${\mathbf{\Sigma}}_{n{n_p}}^\dagger = {\mathbf{\Sigma}}_{n{n_p}}^\mathrm{l} + {\mathbf{\Sigma}}_{n{n_p}}^\mathrm{s}$ is the FSA cross-covariance matrix between the prediction locations $\mathcal{S}^p$ and the training coordinates ${\mathcal{S}}$. These large-scale and short-scale terms are given by 
\begin{align*}
    {\mathbf{\Sigma}}_{n{n_p}}^\mathrm{l} &= {\mathbf{\Sigma}}_{mn}^\mathrm{T}{\mathbf{\Sigma}}_{m}^{-1}{\mathbf{\Sigma}}_{m{n_p}}\quad\text{ and }\quad{\mathbf{\Sigma}}_{n{n_p}}^\mathrm{s} = ({\mathbf{\Sigma}}_{n{n_p}} - {\mathbf{\Sigma}}_{n{n_p}}^\mathrm{l})\circ \boldsymbol{T}_{n{n_p}}(\gamma),
\end{align*}
where $\boldsymbol{T}_{n{n_p}}(\gamma) = \big[c_{\text {taper }}(||\boldsymbol{s}_i-\boldsymbol{s}^p_j||_2 ; \gamma)\big]_{i=1:n, j=1:{n_p}}\in\mathbb{R}^{n\times {n_p}}$ and \\${\mathbf{\Sigma}}_{m{n_p}}=\big[c\left(\boldsymbol{s}^*_i, \boldsymbol{s}^p_j\right)\big]_{i=1:m, j=1:{n_p}}
\in\mathbb{R}^{m\times {n_p}}$. Further, 
$${\mathbf{\Sigma}}_{{n_p}}^\dagger= {\mathbf{\Sigma}}_{m{n_p}}^\mathrm{T}{\mathbf{\Sigma}}_{m}^{-1}{\mathbf{\Sigma}}_{m{n_p}} + ({\mathbf{\Sigma}}_{{n_p}} - {\mathbf{\Sigma}}_{m{n_p}}^\mathrm{T}{\mathbf{\Sigma}}_{m}^{-1}{\mathbf{\Sigma}}_{m{n_p}})\circ \boldsymbol{T}_{{n_p}}(\gamma) \in\mathbb{R}^{n_p\times n_p}$$ 
are unconditional covariances and $\boldsymbol{T}_{n_p}(\gamma) = \big[c_{\text {taper }}\left(||\boldsymbol{s}^p_i-\boldsymbol{s}^p_j||_2 ; \gamma\right)\big]_{i=1:{n_p}, j=1:{n_p}}\in\mathbb{R}^{{n_p}\times {n_p}}$.

Assuming that the linear solve $\Tilde{\mathbf{\Sigma}}_\dagger^{-1}(\boldsymbol{y}-\boldsymbol{X}\mathbf{\boldsymbol{\beta}})$ is precomputed during 
the parameter estimation, the computational cost for the predictive mean $\boldsymbol{\mu}^p_\dagger$ is $\mathcal{O}\big({n_p}\cdot({n}_\gamma^p+m) + n\cdot m\big)$, where ${n}_\gamma^p$ represents the average number of nonzero entries per row in $({\mathbf{\Sigma}}_{n{n_p}}^\mathrm{s})^{\mathrm{T}}$. Additionally, assuming that the necessary Cholesky factor is precomputed, the computational complexity of the predictive covariance $\mathbf{\Sigma}^p_\dagger$ is $\mathcal{O}\big(m\cdot n_\gamma\cdot n + m\cdot {n}_\gamma^p\cdot n_p + m^2\cdot n + m^2\cdot n_p + n\cdot n_p \cdot n_\gamma + n\cdot n_p^2\big)$; see Appendix \ref{AppVar} for a derivation. Further, the predictive variances
\begin{align}
\text{diag}(\mathbf{\Sigma}^p_\dagger) = \sigma_1^2 \boldsymbol{I}_{n_p} + \sigma^2 \boldsymbol{I}_{n_p}-\text{diag}\big(({\mathbf{\Sigma}}_{n{n_p}}^\dagger)^{\mathrm{T}}\Tilde{\mathbf{\Sigma}}^{-1}_\dagger{\mathbf{\Sigma}}_{n{n_p}}^\dagger\big)\in\mathbb{R}^{{n_p}},\label{MuPredFSA}
\end{align}
where $\sigma_1^2$ is the marginal variance parameter, can be calculated in $ \mathcal{O}\big(m\cdot n_\gamma\cdot n + m\cdot {n}_\gamma^p\cdot n_p + m^2\cdot n + m^2\cdot n_p + n\cdot n_p \cdot n_\gamma\big)$ operations; see Appendix~\ref{AppVar} for a derivation. These computations can become prohibitively slow for large $n$ and $n_p$ due to the product $n\cdot n_p$ term. The required storage for computing the predictive mean and variances is of order $\mathcal{O}\big({n_p}\cdot (m + {n}_\gamma^p) + n\cdot (m + {n}_\gamma)\big)$.

\section{Iterative methods for the full-scale approximation}\label{sect3}

In this section, we show how iterative methods and stochastic approximations can be used for fast computations with the FSA.

\subsection{Parameter estimation}\label{SectionIterPE}
To compute linear solves involving $\Tilde{\mathbf{\Sigma}}_{\dagger}$, we apply the conjugate gradient (CG) method. In each iteration of the CG method, we calculate matrix-vector products of the form $(\Tilde{\mathbf{\Sigma}}_{\mathrm{s}} + \mathbf{\Sigma}_{mn}^{\mathrm{T}}\mathbf{\Sigma}_{m}^{-1}\mathbf{\Sigma}_{mn})\boldsymbol{v}$ for $\boldsymbol{v}\in\mathbb{R}^n$, which require $\mathcal{O}\big(n\cdot( m + n_\gamma)\big)$ operations. To approximate the log-determinant $\log\det\big(\Tilde{\mathbf{\Sigma}}_{\dagger}\big)$ in \eqref{NEGLLDNEGLLFSA}, we employ the stochastic Lanczos quadrature (SLQ) method \citep{ubaru2017fast} which combines a quadrature technique based on the Lanczos algorithm with Hutchinson's estimator \citep{hutchinson1989stochastic} as follows:
\begin{align*}
    \log\det\big(\Tilde{\mathbf{\Sigma}}_{\dagger}\big) & =\Tr\Big(\log\big(\Tilde{\mathbf{\Sigma}}_{\dagger}\big)\Big) \approx \frac{1}{\ell} \sum_{i=1}^{\ell} \boldsymbol{z}_i^{\mathrm{T}} \log\big(\Tilde{\mathbf{\Sigma}}_{\dagger}\big) \boldsymbol{z}_i \approx \frac{1}{\ell} \sum_{i=1}^{\ell} \boldsymbol{z}_i^{\mathrm{T}} \tilde{\boldsymbol{Q}}_{i}\log\big(\tilde{\boldsymbol{T}}_{i}\big)\tilde{\boldsymbol{Q}}_{i}^\mathrm{T} \boldsymbol{z}_i\\ &= \frac{1}{\ell} \sum_{i=1}^{\ell} ||\boldsymbol{z}_i||_2^2\boldsymbol{e}_1^{\mathrm{T}} \log\big(\tilde{\boldsymbol{T}}_{i}\big)\boldsymbol{e}_1 \approx \frac{n}{\ell} \sum_{i=1}^{\ell} \boldsymbol{e}_1^{\mathrm{T}} \log\big(\tilde{\boldsymbol{T}}_{i}\big)\boldsymbol{e}_1,
\end{align*}
where $\tilde{\boldsymbol{Q}}_{i}\tilde{\boldsymbol{T}}_{i}\tilde{\boldsymbol{Q}}_{i}^\mathrm{T}$ is a partial Lanczos tridiagonalization of $\Tilde{\mathbf{\Sigma}}_{\dagger}$, where $\tilde{\boldsymbol{T}}_{i}\in\mathbb{R}^{k\times k}$ is tridiagonal, and $\tilde{\boldsymbol{Q}}_{i}\in\mathbb{R}^{n\times k}$ has orthonormal columns obtained by running the Lanczos algorithm for $k$ steps with initial vector $\boldsymbol{z}_i/||\boldsymbol{z}_i||_2$, where $\mathbb{E}\left[\boldsymbol{z}_i\right]=\mathbf{0}$ and $\mathbb{E}\left[ \boldsymbol{z}_i\boldsymbol{z}_i^\mathrm{T}\right]=\boldsymbol{I}_n$, for instance, $\boldsymbol{z}_i\sim\mathcal{N}(\boldsymbol{0},\boldsymbol{I}_n)$. \citet{dong2017scalable} analyzed different approaches to estimate the $\log\det$ and found that SLQ is superior to other methods.

As in \citet{gardner2018gpytorch}, we use the connection between the Lanczos and the CG algorithm to calculate the $\ell$ partial Lanczos tridiagonal matrices $\Tilde{\boldsymbol{T}}_i$ and thus avoid the necessity for explicitly employing the Lanczos tridiagonalization method. This is advantageous since the Lanczos algorithm can be numerically unstable and can have high storage requirements. In each iteration of this modified CG method, we calculate matrix-matrix products, which involve $\mathcal{O}\big(n\cdot (m + n_\gamma)\big)$ operations. For a more detailed description, we refer to the pseudo-algorithm provided in Appendix \ref{appendix:CGalgo}.

In addition to the $\ell$ partial Lanczos tridiagonal matrices, the modified CG method computes the $\ell$ linear solves $(\Tilde{\mathbf{\Sigma}}_{\mathrm{s}} + \mathbf{\Sigma}_{mn}^{\mathrm{T}}\mathbf{\Sigma}_{m}^{-1}\mathbf{\Sigma}_{mn})^{-1}\boldsymbol{z}_i$ for the probe vectors $\boldsymbol{z}_i$. This means that once the log-likelihood is calculated, gradients can be calculated with minimal computational overhead by approximating the trace term in \eqref{NEGLLDNEGLLFSA} using stochastic trace estimation (STE), as follows:
\begin{align*}
\Tr \Big(\Tilde{\mathbf{\Sigma}}^{-1}_{\dagger}\frac{\partial \Tilde{\mathbf{\Sigma}}_{\dagger}}{\partial \boldsymbol{\theta}}\Big)\approx\frac{1}{\ell}\sum_{i=1}^{\ell} \big(\boldsymbol{z}_i^\mathrm{T} \Tilde{\mathbf{\Sigma}}^{-1}_{\dagger}\big)\big(\frac{\partial \Tilde{\mathbf{\Sigma}}_{\dagger}}{\partial \boldsymbol{\theta}} \boldsymbol{z}_i\big).
\end{align*}

Therefore, for calculating the negative log-likelihood and its derivatives, we have storage requirements of order $\mathcal{O}\big(n\cdot(m+n_\gamma)\big)$ and computational costs of $\mathcal{O}\big(n\cdot (m^2 + m\cdot t + n_\gamma \cdot t)\big)$, where $t$ is the number of iterations in the CG algorithms. The number $t$ has a slight dependence on $n$, $m$, and $n_\gamma$; see the simulated experiments in Section \ref{comparison_PC} and Theorems $\ref{thm1}$ and $\ref{th2}$. In Appendix \ref{Fishapp}, we also show how the Fisher information can be calculated efficiently using STE.

\subsection{Predictive variances}\label{sec_pred_var}

Calculating predictive variances in \eqref{MuPredFSA} is computationally expensive when $n$ and $n_p$ are large, even with iterative methods due to $n_p$ right-hand sides. In Algorithm \ref{alg:pred_var}, we propose a more efficient simulation-based approach. This algorithm results in an unbiased and consistent approximation for $\text{diag}(\boldsymbol{\Sigma}^p_\dagger)$; see Appendix \ref{AppVar} for a proof of
Proposition \ref{PropPredVar}. The diagonals ${\boldsymbol{D}}^d_1,\dots, {\boldsymbol{D}}^d_3$ in Algorithm \ref{alg:pred_var} are calculated deterministically using the preconditioned CG method for the linear solves $\boldsymbol{\Sigma}_\dagger^{-1}\mathbf{\Sigma}_{mn}^\mathrm{T}$ and $\Tilde{\mathbf{\Sigma}}_{\mathrm{s}}^{-1}\mathbf{\Sigma}_{mn}^\mathrm{T}$. We estimate the remaining diagonal term ${\boldsymbol{D}}_\ell \approx \text{diag}\big(({\mathbf{\Sigma}}_{n{n_p}}^\mathrm{s})^\mathrm{T}\Tilde{\mathbf{\Sigma}}_{\mathrm{s}}^{-1}{\mathbf{\Sigma}}_{n{n_p}}^\mathrm{s}\big)$ stochastically \citep{bekas2007estimator}, where the linear solves $\Tilde{\mathbf{\Sigma}}_{\mathrm{s}}^{-1}{\mathbf{\Sigma}}_{n{n_p}}^\mathrm{s}\boldsymbol{z}_i^{(1)}$
can be done using the preconditioned CG method. The computational complexity of Algorithm \ref{alg:pred_var} is $\mathcal{O}\big({n_p}\cdot (m \cdot {n}_\gamma^p + m^2) + n\cdot (m \cdot {n}_\gamma\cdot t +  m^2\cdot t)\big)$ with storage requirements of $\mathcal{O}\big({n_p}\cdot (m + {n}_\gamma^p) + n\cdot (m + {n}_\gamma)\big)$. In addition, the algorithm can be easily parallelized because it relies on matrix-matrix multiplications. Note that we use Rademacher random vectors with entries $\pm 1$; see Appendix \ref{radapp}.

\begin{proposition}\label{PropPredVar}
    Algorithm \ref{alg:pred_var} produces unbiased and consistent estimates $\boldsymbol{D}^p$ of the predictive
variance $\text{diag}(\boldsymbol{\Sigma}^p_\dagger)$ given in \eqref{MuPredFSA}.
\end{proposition}

\begin{algorithm}
\caption{Approximate predictive variances using simulation}\label{alg:pred_var}
\begin{algorithmic}[1]
\Require Matrices ${\mathbf{\Sigma}}_{n{n_p}}^\mathrm{s}$, ${\mathbf{\Sigma}}_{mn}$, ${\mathbf{\Sigma}}_{m}$, ${\mathbf{\Sigma}}_{m{n_p}}$, $\Tilde{\mathbf{\Sigma}}_\mathrm{s}$
\Ensure Approximated predictive variances $\boldsymbol{D}^p\approx\text{diag}(\boldsymbol{\Sigma}^p_\dagger)$
\State ${\boldsymbol{D}}^d_1 \gets \text{diag}({\mathbf{\Sigma}}_{m{n_p}}^\mathrm{T}\mathbf{\Sigma}_{m}^{-1}{\mathbf{\Sigma}}_{mn}\boldsymbol{\Sigma}_\dagger^{-1}\mathbf{\Sigma}_{mn}^\mathrm{T}{\mathbf{\Sigma}}_{m}^{-1}{\mathbf{\Sigma}}_{m{n_p}})$
\State ${\boldsymbol{D}}^d_2 \gets\text{diag}\big(({\mathbf{\Sigma}}_{n{n_p}}^\mathrm{s})^\mathrm{T}\boldsymbol{\Sigma}_\dagger^{-1}\mathbf{\Sigma}_{mn}^\mathrm{T}\mathbf{\Sigma}_{m}^{-1}{\mathbf{\Sigma}}_{m{n_p}}\big)$
\State ${\boldsymbol{D}}^d_3 \gets\text{diag}\big(({\mathbf{\Sigma}}_{n{n_p}}^\mathrm{s})^\mathrm{T}\Tilde{\mathbf{\Sigma}}_{\mathrm{s}}^{-1}\mathbf{\Sigma}_{mn}^\mathrm{T}(\mathbf{\Sigma}_{m} + \mathbf{\Sigma}_{mn} \Tilde{\mathbf{\Sigma}}_{\mathrm{s}}^{-1}\mathbf{\Sigma}_{mn}^\mathrm{T})^{-1}(\Tilde{\mathbf{\Sigma}}_{\mathrm{s}}^{-1}\mathbf{\Sigma}_{mn}^\mathrm{T})^\mathrm{T}{\mathbf{\Sigma}}_{n{n_p}}^\mathrm{s}\big)$
\State ${\boldsymbol{D}}_\ell\gets \boldsymbol{0}\in\mathbb{R}^{n_p}$
\For{$i = 1$ to $\ell$}
\State $\boldsymbol{z}_i^{(2)} \gets ({\mathbf{\Sigma}}_{n{n_p}}^\mathrm{s})^\mathrm{T}\Tilde{\mathbf{\Sigma}}_{\mathrm{s}}^{-1}{\mathbf{\Sigma}}_{n{n_p}}^\mathrm{s}\boldsymbol{z}_i^{(1)}$, where $\boldsymbol{z}_i^{(1)}\sim$ Rademacher
\State ${\boldsymbol{D}}_\ell\gets {\boldsymbol{D}}_\ell + \boldsymbol{z}_i^{(1)} \circ \boldsymbol{z}_i^{(2)}$
\EndFor
\State ${\boldsymbol{D}}_\ell\gets \frac{1}{\ell}{\boldsymbol{D}}_\ell$\qquad $\Big(\approx \text{diag}\big(({\mathbf{\Sigma}}_{n{n_p}}^\mathrm{s})^\mathrm{T}\Tilde{\mathbf{\Sigma}}_{\mathrm{s}}^{-1}{\mathbf{\Sigma}}_{n{n_p}}^\mathrm{s}\big)\Big)$
\State $\boldsymbol{D}^p\gets \sigma_1^2 \boldsymbol{1}_{n_p} + \sigma^2 \boldsymbol{1}_{n_p} - {\boldsymbol{D}}^d_1 - 2\cdot {\boldsymbol{D}}^d_2 +  {\boldsymbol{D}}^d_3 - {\boldsymbol{D}}_\ell$
\end{algorithmic}
\end{algorithm}

Another approach for approximating the costly diagonal $\text{diag}\big(({\mathbf{\Sigma}}_{n{n_p}}^\mathrm{s})^\mathrm{T}\Tilde{\mathbf{\Sigma}}_{\mathrm{s}}^{-1}{\mathbf{\Sigma}}_{n{n_p}}^\mathrm{s}\big)$ is to use the Lanczos algorithm as proposed by \citet{pleiss2018constant}:
\begin{align}\label{Landiag}
\text{diag}\big(({\mathbf{\Sigma}}_{n{n_p}}^\mathrm{s})^\mathrm{T}\Tilde{\mathbf{\Sigma}}_{\mathrm{s}}^{-1}{\mathbf{\Sigma}}_{n{n_p}}^\mathrm{s}\big) \approx \text{diag}\big(({\mathbf{\Sigma}}_{n{n_p}}^\mathrm{s})^\mathrm{T}\Tilde{\boldsymbol{Q}}_k^\mathrm{s} (\Tilde{\boldsymbol{T}}_k^\mathrm{s})^{-1}(\Tilde{\boldsymbol{Q}}_k^\mathrm{s})^\mathrm{T}{\mathbf{\Sigma}}_{n{n_p}}^\mathrm{s}\big),
\end{align}
where $\Tilde{\boldsymbol{Q}}_k^\mathrm{s}\in\mathbb{R}^{n\times k}$ and $\Tilde{\boldsymbol{T}}_k^\mathrm{s}\in\mathbb{R}^{k\times k}$ denote the partial Lanczos tridiagonalization of $\Tilde{\mathbf{\Sigma}}_{\mathrm{s}}$. This is currently a state-of-the-art approach in machine learning and allows us to compute the diagonal $\text{diag}\big(({\mathbf{\Sigma}}_{n{n_p}}^\mathrm{s})^\mathrm{T}\Tilde{\mathbf{\Sigma}}_{\mathrm{s}}^{-1}{\mathbf{\Sigma}}_{n{n_p}}^\mathrm{s}\big)$ in $\mathcal{O}\big(k^2\cdot (n+{n_p}) + k\cdot ({n_p}\cdot{n}_\gamma^p + n\cdot{n}_\gamma)\big)$. 

The Lanczos tridiagonalization method can work well for approximating covariance matrices because their eigenvalue distribution typically contains a small number of large eigenvalues and a larger number of small ones. This structure makes covariance matrices well-suited for low-rank approximations because the Lanczos method effectively captures the dominant eigenvalues with relatively few iterations. However, for predictive variances in \eqref{Landiag}, the Lanczos tridiagonalization method is used to approximate the inverse of a covariance matrix. This approach leads to a distribution of eigenvalues for the inverse that is essentially inverted, resulting in many large eigenvalues and only a few small ones. This makes low-rank approximations less effective because capturing the numerous large eigenvalues necessitates a very high rank $k$ to ensure accurate approximations. Therefore, the Lanczos tridiagonalization is typically not well-suited for approximating the inverse of covariance matrices because it requires significantly more computational effort and a higher rank to achieve accurate results; see our experiments in Section \ref{exp_pred_var}.

\subsection{Preconditioning}
Iterative methods use preconditioners to improve convergence properties and to reduce variance in stochastic approximations. When using a preconditioner $\boldsymbol{P}$, the CG algorithm solves the equivalent preconditioned system
$$\boldsymbol{P}^{-\frac{1}{2}}\Tilde{\mathbf{\Sigma}}_{\dagger}\boldsymbol{P}^{-\frac{1}{2}}\hat{\boldsymbol{x}} = \boldsymbol{P}^{-\frac{1}{2}}\boldsymbol{b},~~~~\hat{\boldsymbol{x}} = \boldsymbol{P}^{\frac{1}{2}}{\boldsymbol{x}},$$  instead of $\Tilde{\mathbf{\Sigma}}_{\dagger}\boldsymbol{x} = \boldsymbol{b}$. The corresponding stochastic Lanczos quadrature approximation for log-determinants is given by \citet{gardner2018gpytorch}, as follows:
\begin{align*}
    \log\det\big(\Tilde{\mathbf{\Sigma}}_{\dagger}\big) &= \log\det\big(\boldsymbol{P}^{-\frac{1}{2}}\Tilde{\mathbf{\Sigma}}_{\dagger}\boldsymbol{P}^{-\frac{1}{2}}\big) + \log\det(\boldsymbol{P})\\&\approx \frac{n}{\ell} \sum_{i=1}^{\ell} \boldsymbol{e}_1^\mathrm{T} \log\big(\tilde{\boldsymbol{T}}_{i}\big) \boldsymbol{e}_1 + \log\det(\boldsymbol{P}),
\end{align*}
where $\tilde{\boldsymbol{T}}_{i}\in\mathbb{R}^{k\times k}$ is the Lanczos tridiagonal matrix of $\boldsymbol{P}^{-\frac{1}{2}}\Tilde{\mathbf{\Sigma}}_{\dagger}\boldsymbol{P}^{-\frac{1}{2}}$ obtained by running the Lanczos algorithm for $k$ steps with initial vector $\boldsymbol{P}^{-\frac{1}{2}}\boldsymbol{z}_i / \|\boldsymbol{P}^{-\frac{1}{2}}\boldsymbol{z}_i\|_2$, where $\boldsymbol{z}_i \sim \mathcal{N}(\boldsymbol{0},\boldsymbol{P})$. Gradients of log-determinants can be calculated using stochastic trace estimation as follows:
\begin{align*}
    \Tr \Big(\Tilde{\mathbf{\Sigma}}^{-1}_{\dagger}\frac{\partial \Tilde{\mathbf{\Sigma}}_{\dagger}}{\partial \boldsymbol{\theta}}\Big) &= \Tr \Big(\Tilde{\mathbf{\Sigma}}^{-1}_{\dagger}\frac{\partial \Tilde{\mathbf{\Sigma}}_{\dagger}}{\partial \boldsymbol{\theta}}\mathbb{E}_{\boldsymbol{z}_i\sim\mathcal{N}(\boldsymbol{0},\boldsymbol{P})}\left[\boldsymbol{P}^{-1} \boldsymbol{z}_i\boldsymbol{z}_i^\mathrm{T}\right]\Big)\\
    &= \mathbb{E}_{\boldsymbol{z}_i\sim\mathcal{N}(\boldsymbol{0},\boldsymbol{P})}\Big[\big(\boldsymbol{z}_i^\mathrm{T}\Tilde{\mathbf{\Sigma}}^{-1}_{\dagger}\big)\big(\frac{\partial\Tilde{\mathbf{\Sigma}}_{\dagger}}{\partial \boldsymbol{\theta}}\boldsymbol{P}^{-1} \boldsymbol{z}_i\big)\Big]\\&\approx \frac{1}{\ell} \sum_{i=1}^{\ell} \big(\boldsymbol{z}_i^\mathrm{T}\Tilde{\mathbf{\Sigma}}^{-1}_{\dagger}\big)\big(\frac{\partial\Tilde{\mathbf{\Sigma}}_{\dagger}}{\partial \boldsymbol{\theta}}\boldsymbol{P}^{-1} \boldsymbol{z}_i\big) = \Tilde{T}_\ell.
\end{align*}
Further, we can reduce the variance of this stochastic estimate by using the preconditioner $\boldsymbol{P}$ to build a control variate; see Appendix \ref{contvarapp}. 

In the FSA, we have a low-rank approximation part $\mathbf{\Sigma}_{mn}^{\mathrm{T}}\mathbf{\Sigma}_{m}^{-1}\mathbf{\Sigma}_{mn}$ with rank $m$. We use this to define the fully independent training conditional (FITC) preconditioner as
\begin{align}\label{eqfitcp}
    \widehat{\boldsymbol{P}} = \boldsymbol{D}_{\mathrm{s}} + \mathbf{\Sigma}_{mn}^{\mathrm{T}}\mathbf{\Sigma}_{m}^{-1}\mathbf{\Sigma}_{mn},
\end{align}
where $\boldsymbol{D}_{\mathrm{s}} = \text{diag}(\Tilde{\mathbf{\Sigma}}_{\mathrm{s}}) = \text{diag}(\mathbf{\Sigma} - \mathbf{\Sigma}_{mn}^{\mathrm{T}}\mathbf{\Sigma}_{m}^{-1}\mathbf{\Sigma}_{mn} + \sigma^2 \boldsymbol{I}_n)$. We also use the latter diagonal matrix as preconditioner $\boldsymbol{D}_{\mathrm{s}} = \widehat{\boldsymbol{P}}_{\mathrm{s}} \approx \Tilde{\mathbf{\Sigma}}_{\mathrm{s}}$ in the computations of the predictive variances, as outlined in Section \ref{sec_pred_var}, where linear solves with $\Tilde{\mathbf{\Sigma}}_{\mathrm{s}}$ are required. Furthermore, in Appendix \ref{contvarapp}, we show how $\widehat{\boldsymbol{P}}$ can be used to construct a control variate in the stochastic estimation of the diagonal $\text{diag}\big(({\mathbf{\Sigma}}_{n{n_p}}^\mathrm{s})^\mathrm{T}\Tilde{\mathbf{\Sigma}}_{\mathrm{s}}^{-1}{\mathbf{\Sigma}}_{n{n_p}}^\mathrm{s}\big)$ for predictive variances. 

Concerning computational costs, $\log\det\big(\widehat{\boldsymbol{P}}\big)$ and calculating $\Tr \big(\widehat{\boldsymbol{P}}^{-1}\frac{\partial \widehat{\boldsymbol{P}}}{\partial \boldsymbol{\theta}}\big)$ are of complexity $\mathcal{O}(n\cdot m^2)$, while linear solves are of complexity $\mathcal{O}(n\cdot m)$. Moreover, the computational overhead for sampling from $\mathcal{N}(\boldsymbol{0},\widehat{\boldsymbol{P}})$ is $\mathcal{O}(n\cdot m)$; see Appendix \ref{AppRT}.

\subsection{Convergence theory}\label{secCG}
Next, we analyze the convergence properties of the preconditioned CG method when applied to the FSA. For the following statements and proofs, we denote the Frobenius and the 2-norm (spectral norm) by $||\cdot||_F$ and $||\cdot||_2$, respectively, and define the vector norm $||\boldsymbol{v}||_{\boldsymbol{A}} = \sqrt{\boldsymbol{v}^\mathrm{T}\boldsymbol{A}\boldsymbol{v}}$ for $\boldsymbol{v}\in\mathbb{R}^n$ and a positive semidefinite matrix $\boldsymbol{A}\in\mathbb{R}^{n\times n}$. We make the following assumptions.
\begin{assumption}\label{assumpt1}
  $n\geq 2$.
\end{assumption}
\begin{assumption}\label{assumpt2}
  The $m \in \{1,2,\dots,n\}$ inducing points are sampled uniformly without replacement from the set of locations $\mathcal{S}$.
\end{assumption}
\begin{assumption}\label{assumpt3}
  The covariance matrix $\boldsymbol{\Sigma}$ is of the form $\Sigma_{ij}=\sigma_1^2 \cdot r\left(\boldsymbol{s}_i,\boldsymbol{s}_j\right)$, where $r(\cdot)$ is positive and continuous, and $r(0)=1$. Additionally, the matrix $\boldsymbol{\Sigma}$ has eigenvalues $\lambda_1 \geq ... \geq \lambda_n > 0$.
\end{assumption}
First, we analyze the convergence speed of the CG method for linear solves with $\Tilde{\boldsymbol{\Sigma}}_{\dagger}$.

\begin{theorem}\label{thm1}
    \textbf{Convergence of the CG method:}
    Let $\Tilde{\boldsymbol{\Sigma}}_\dagger\in \mathbb{R}^{n \times n}$ be the full-scale approximation of a covariance matrix $\Tilde{\boldsymbol{\Sigma}} =\boldsymbol{\Sigma} + \sigma^2\boldsymbol{I}_n  \in \mathbb{R}^{n \times n}$ with $m$ inducing points, taper range $\gamma$, and $\sigma^2>0$. Consider the linear system $\Tilde{\boldsymbol{\Sigma}}_\dagger \mathbf{u}^*=\mathbf{y}$, where $\mathbf{y}\in\mathbb{R}^n$. Let $\mathbf{u}_k$ be the approximation in the $k$th iteration of the CG method. Given Assumptions~\ref{assumpt1}--\ref{assumpt3}, the following holds for the relative error:
    \begin{align*}
        \frac{\left\|\mathbf{u}^*-\mathbf{u}_k\right\|_{\Tilde{\boldsymbol{\Sigma}}_\dagger}}{\left\|\mathbf{u}^*-\mathbf{u}_0\right\|_{\Tilde{\boldsymbol{\Sigma}}_\dagger}} \leq 2\Bigg({1 + \mathcal{O}_P\Big(\sigma\cdot\big((\lambda_{m+1} + \frac{n}{\sqrt{m}}\cdot\sigma_1^2)\cdot\sqrt{n\cdot n_\gamma}+\lambda_1\big)^{-\frac{1}{2}}\Big)}\Bigg)^{-k},
    \end{align*}
    where $\mathcal{O}_P$ is the $\mathcal{O}$-notation in probability (see Definition \ref{defO} in Appendix \ref{AppConv}).
\end{theorem}
For the proof, see Appendix \ref{AppConv}. Theorem \ref{thm1} shows that selecting a narrower taper range $\gamma$ leads to improved convergence in the CG method. However, the relationship is more complicated with respect to the number of inducing points. On the one hand, terms such as $\frac{n}{\sqrt{m}}$ and $\lambda_{m+1}$ decrease with larger $m$ for a given $n$ and thus lead to faster convergence. On the other hand, we bound $||\boldsymbol{\Sigma}_{mn}^\mathrm{T}\boldsymbol{\Sigma}^{-1}_m\boldsymbol{\Sigma}_{mn}||_2$ in the proof of Theorem \ref{thm1} by $||\boldsymbol{\Sigma}||_2$; see Appendix \ref{AppConv}. However, $||\boldsymbol{\Sigma}_{mn}^\mathrm{T}\boldsymbol{\Sigma}^{-1}_m\boldsymbol{\Sigma}_{mn}||_2$ grows with $m$. Furthermore, we observe that the convergence is slower with decreasing $\sigma^2$ and increasing values of $\lambda_1$ and $\lambda_{m+1}$, which are positively related to the covariance parameters.

\begin{theorem}\label{th2}
    \textbf{Convergence of the CG method with the FITC preconditioner:}
    Let $\Tilde{\boldsymbol{\Sigma}}_\dagger\in \mathbb{R}^{n \times n}$ be the full-scale approximation of a covariance matrix $\Tilde{\boldsymbol{\Sigma}} =\boldsymbol{\Sigma} + \sigma^2\boldsymbol{I}_n  \in \mathbb{R}^{n \times n}$ with $m$ inducing points, taper range $\gamma$, and $\sigma^2>0$. Consider the linear system $\Tilde{\boldsymbol{\Sigma}}_\dagger \mathbf{u}^*=\mathbf{y}$, where $\mathbf{y}\in\mathbb{R}^n$. Let $\mathbf{u}_k$ be the approximation in the $k$th iteration of the preconditioned CG method with the FITC preconditioner. Given Assumptions~\ref{assumpt1}--\ref{assumpt3}, the following holds for the relative error:
    \begin{align*}
       \frac{\left\|\mathbf{u}^*-\mathbf{u}_k\right\|_{\Tilde{\boldsymbol{\Sigma}}_\dagger}}{\left\|\mathbf{u}^*-\mathbf{u}_0\right\|_{\Tilde{\boldsymbol{\Sigma}}_\dagger}} \leq 2\Bigg({1 + \mathcal{O}_P\Big(\sigma^2 \cdot\big((\lambda_{m+1} + \frac{n}{\sqrt{m}}\cdot \sigma_1^2)\cdot \sqrt{n\cdot (n_\gamma-1)}\big)^{-1}\Big)}\Bigg)^{-k},
    \end{align*}
    where $\mathcal{O}_P$ is the $\mathcal{O}$-notation in probability (see Definition \ref{defO} in Appendix \ref{AppConv}).
\end{theorem}
For the proof, see Appendix \ref{AppConv}. We note that the upper bound for the relative error in Theorem \ref{th2} also holds for $\frac{\left\|\mathbf{u}^*-\mathbf{u}_k\right\|_{\Tilde{\boldsymbol{\Sigma}}_{\mathrm{s}}}}{\left\|\mathbf{u}^*-\mathbf{u}_0\right\|_{\Tilde{\boldsymbol{\Sigma}}_{\mathrm{s}}}}$ using the CG method with preconditioner $\widehat{\boldsymbol{P}}_{\mathrm{s}} = \text{diag}(\Tilde{\boldsymbol{\Sigma}}_{\mathrm{s}})$ to solve the linear system $\Tilde{\boldsymbol{\Sigma}}_{\mathrm{s}}\mathbf{u}^*= \mathbf{y}$. Similarly to Theorem \ref{thm1}, Theorem \ref{th2} shows that selecting a narrower taper range $\gamma$ leads to improved convergence in the preconditioned CG method. In particular, this improvement is characterized by a relative error approaching zero as the taper range goes to zero, $\gamma \rightarrow 0$, a logical outcome considering that, in this scenario, $\Tilde{\boldsymbol{\Sigma}}_\dagger\xrightarrow{\gamma\rightarrow 0}\widehat{\boldsymbol{P}}$. Furthermore, our findings indicate that the convergence rate of the CG method with the FITC preconditioner does not depend on the largest eigenvalue $\lambda_1$. Consequently, using the FITC preconditioner exhibits less sensitivity to the sample size and the correlation function parameterized, e.g., by a range parameter. Furthermore, a higher number of inducing points $m$ leads to faster convergence. These theoretical results align with the empirical results in our simulation study; see Figure \ref{fig:IT}.

\section{Simulation study}\label{sect4}
In the following, we analyze our methods in experiments with simulated data. Unless stated otherwise, we simulate a sample of size $n=100'000$ from a zero-mean GP with a Matérn covariance function with smoothness parameter $\nu = \frac{3}{2}$, marginal variance $\sigma_1^2 = 1$, and nugget effect $\sigma^2 = 1$, and locations are sampled uniformly on the unit square $[0,1]\times[0,1]$. We consider three different choices of effective ranges $(0.5, 0.2, 0.05)$ corresponding to range parameters $\rho$ of approximately $(\frac{0.5}{2.7}, \frac{0.2}{2.7}, \frac{0.05}{2.7}) = (0.1852, 0.0741, 0.0185)$. The effective range is defined as the distance at which the correlation drops to 0.05. We use a tolerance level of $\delta = 0.001$ for checking convergence in the CG algorithm, in line with the recommendations of \citet{maddox2021iterative}. Further, unless specified otherwise, we use $\ell = 50$ sample vectors for the stochastic approximations. All of the following calculations are done with an AMD EPYC 7742 processor and 512 GB of random-access memory using the GPBoost library version 1.5.5 compiled with the GCC compiler version 11.2.0. The code to reproduce the experiments is available at \url{https://github.com/TimGyger/iterativeFSA}. 

\subsection{Methods for choosing inducing points}
There are various ways of choosing inducing points. We compare three different methods that do not use any response variable information but only the spatial input locations: random selection, kmeans++ \citep{arthur2007k}, and the cover tree algorithm \citep{terenin2022numerically}. The computational complexities of these methods are $\mathcal{O}(m)$, $\mathcal{O}(n\cdot m)$, and $\mathcal{O}\big(n\cdot\log(n)\big)$, respectively. Specifically, we compare the negative log-likelihood evaluated at the true population parameters. As all three methods contain randomness, we repeat the evaluation 25 times with different random number generator seeds, and all calculations are done using the Cholesky decomposition to ensure that the only source of randomness arises from the methods used to select the inducing points.
\begin{figure}[ht!]
    \centering
    \includegraphics[width=\linewidth]{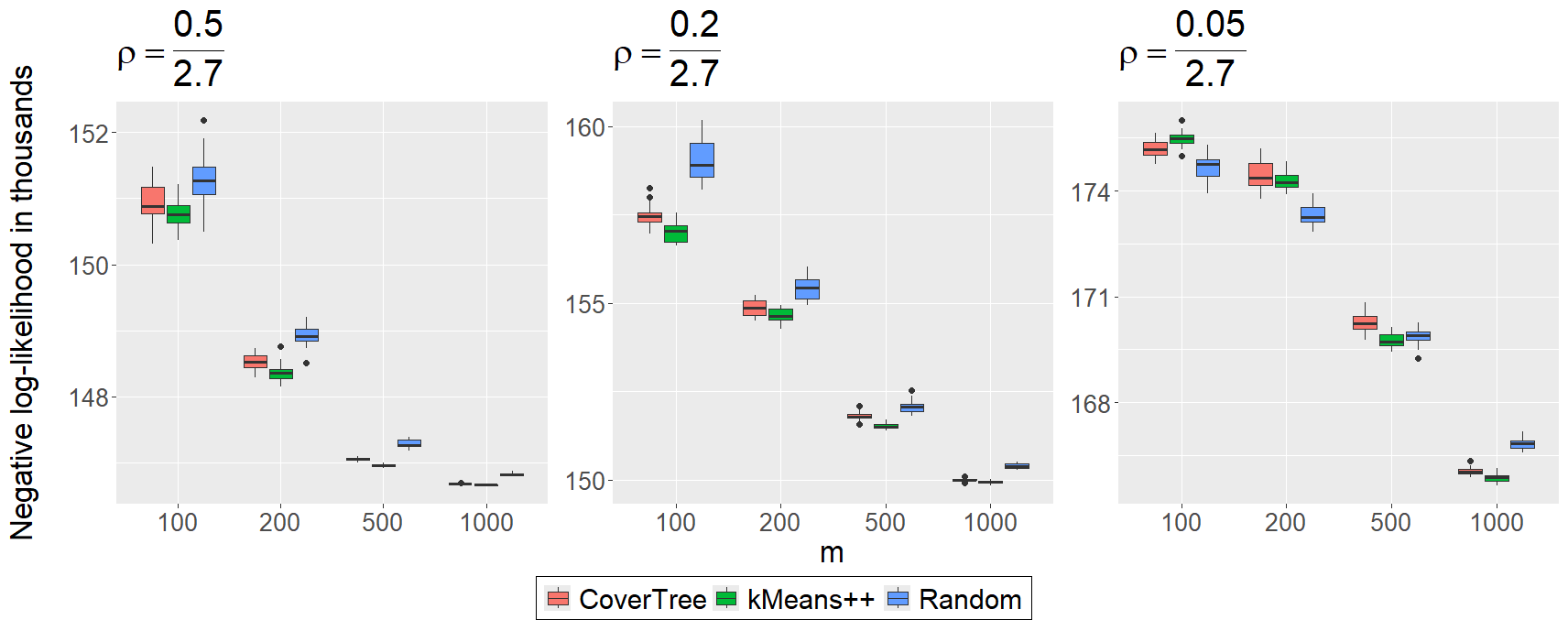}
    \caption{Box-plots of the negative log-likelihood for the FITC approximation for different effective ranges (0.5, 0.2, 0.05 from left to right) and numbers of inducing points $m$ ($n = 100'000$).}
    \label{fig:FITC}
\end{figure}

In Figure \ref{fig:FITC}, we present the results for the FITC approximation, which combines the predictive process with the diagonal of the residual process for varying numbers of inducing points. We find that the kmeans++ algorithm usually results in the lowest negative log-likelihood for a given number of inducing points. The cover tree algorithm achieves slightly worse results, and a random selection is clearly inferior. We observe similar results for the FSA; see Figure \ref{fig:FSA} in Appendix \ref{AppIP}. Therefore, we use the kmeans++ algorithm in the following to determine the inducing points. Note that the experiments in this section, Section \ref{exp_pred_var}, and Section \ref{comparison_PC_Vecchia} are conducted on only one simulated data set since we do not want to mix sampling variability and randomness of the methods that are analyzed. However, the results do not change when using other samples (results not shown).

\subsection{Comparison of preconditioners}\label{comparison_PC}
We compare the proposed FITC preconditioner to the state-of-the-art pivoted Cholesky preconditioner \citep{harbrecht2012low, gardner2018gpytorch} for different ranks $k\in\{200,500,1'000\}$. Note that these ranks exceed the recommendations of \citet{maddox2021iterative}. We analyze the number of CG iterations, and the runtime. Table \ref{Table1} presents the results for simulated data using an effective range of 0.2. In Table \ref{Table1Precapp} in Appendix \ref{simstudapp}, we present additionally the results for effective ranges of 0.05 and 0.5. We observe that the FITC preconditioner clearly outperforms the pivoted Cholesky preconditioner in terms of number of CG iterations and runtime, even for a very high rank $k$. Note that in addition to requiring more CG iterations for convergence, the pivoted Cholesky preconditioner is considerably slower to construct, limiting the rank to small values in practice. In contrast, the FITC preconditioner can be computed more efficiently, and larger ranks can thus be used. Consequently, we use the FITC preconditioner in the following.
\begin{table}[ht!]
\centering
\begin{tabular}{ |p{3.1cm}||p{1.cm}||p{1.2cm}||p{1.6cm}||p{1.6cm}||p{1.8cm}|  }
 \hline
\textbf{Preconditioner:}& \text{None} & \text{FITC} & \multicolumn{3}{|c|}{Pivoted Cholesky}\\
 \hline
 & & &$k = 200$ & $k = 500$& $k = 1'000$\\
 \hhline{|=||=||=||=||=||=|}
 CG-Iterations & 279  & 9   & 91& 52 & 32  \\
 Time (s) & 54  & 10   & 52& 83 & 266  \\ 
 
 \hline
 
\end{tabular}
\caption{Number of (preconditioned) CG-iterations for the linear solve $\Tilde{\boldsymbol{\Sigma}}_{\dagger}^{-1}\boldsymbol{y}$ and the time in seconds (s) for computing the negative log-likelihood for the true population parameters ($n = 100'000$, $m = 500$, $n_\gamma = 80$).}\label{Table1} 
\end{table}

In Figure \ref{fig:IT}, we show the number of iterations required by the CG method with and without the FITC preconditioner for calculating $\Tilde{\boldsymbol{\Sigma}}_{\dagger}^{-1}\boldsymbol{y}$ depending on the parameters $n$, $m$, and $\gamma$. These results align with Theorems \ref{thm1} and \ref{th2}. Specifically, we observe that the convergence of the CG method with the FITC preconditioner is less sensitive to the sample size than without a preconditioner. Furthermore, the number of iterations in the CG algorithm without a preconditioner increases with a growing number of inducing points $m$. On the other hand, the number of iterations decreases in $m$ when using the FITC preconditioner. This is consistent with our convergence theorems. 
\begin{figure}[ht!]
    \centering
    \includegraphics[width=\linewidth]{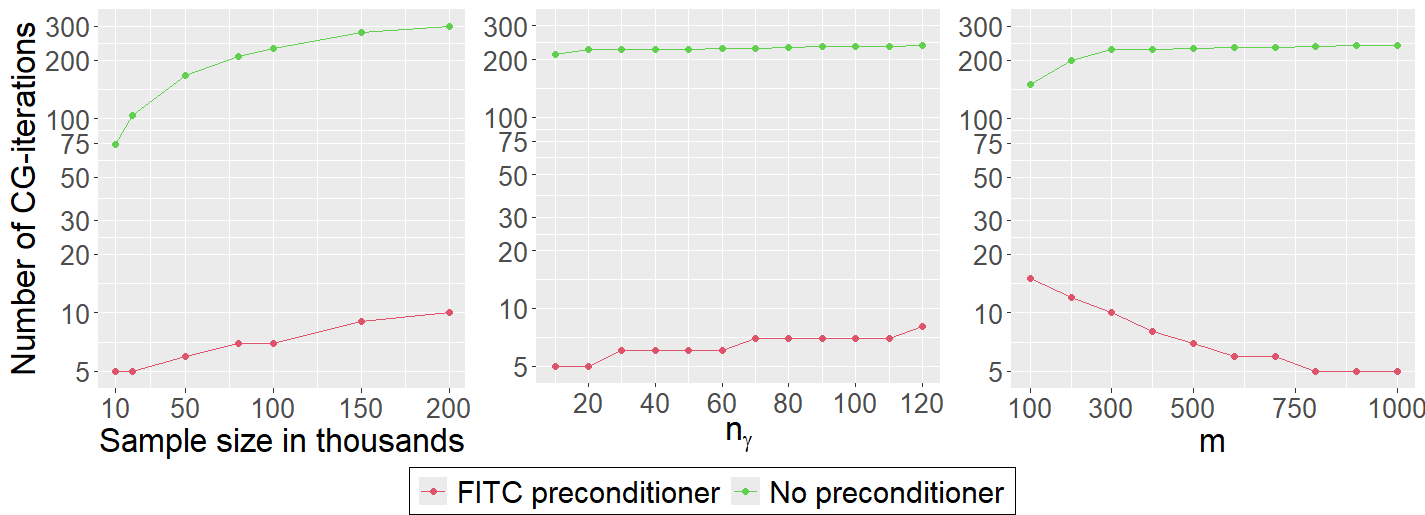}
    \caption{Number of iterations (in log-scale) used in the CG method with and without the FITC preconditioner for calculating $\Tilde{\boldsymbol{\Sigma}}_{\dagger}^{-1}\boldsymbol{y}$ for simulated data with an effective range of 0.2. Left: Different sample sizes $n$ with constant $n_\gamma = 80$ and $m = 500$. Middle: Different taper ranges $\gamma$ with constant $n = 100'000$ and $m = 500$. Right: Different numbers of inducing points $m$ with constant $n = 100'000$ and $n_\gamma = 80$.}
    \label{fig:IT}
\end{figure}

\subsection{Comparison of methods for calculating predictive variances}\label{exp_pred_var}

In the following, we compare the simulation-based approach introduced in Section \ref{sec_pred_var} for calculating predictive variances and the Lanczos algorithm-based method. We use $n=n_p=100'000$ training and prediction locations for doing this. Evaluation is done using the  log-score (LS), $-\frac{1}{n_p}\sum\nolimits_{i=1}^{n_p}\log\big(\phi({\boldsymbol{y}}^*_i ; \boldsymbol{\mu}^p_{\dagger,i}, {\boldsymbol{\sigma}}_{p,i})\big)$, 
where $\phi(x;\mu,\sigma)$ is the density function of a normal distribution with mean $\mu$ and variance $\sigma^2$, ${\boldsymbol{y}}^*$ is the test response, $\boldsymbol{\mu}^p_\dagger$ is the predictive mean, and ${\boldsymbol{\sigma}}^2_p$ is the predictive variance $\text{diag}(\boldsymbol{\Sigma}^p_\dagger)$. We also consider the root mean squared error (RMSE) of the approximate predictive variances compared to the Cholesky-based ones. In Figure \ref{fig:Pred_Sim}, we report the accuracy versus time for varying numbers of ranks $k$ in the Lanczos algorithm and different numbers of samples $\ell$ in the stochastic approach, with $k$, $\ell \in\{50,200,500,1'000,2'000,5'000\}$ when using an effective range of 0.2 and the true population parameters. The figure also reports the runtime for calculating the predictive means and variances using the Cholesky factorization. Additionally, we report results for other range parameters and the standard deviations of the stochastic computations in parentheses in Table \ref{table2} in Appendix \ref{AppPD}.
\begin{figure}[ht!]
    \centering
    \includegraphics[width=\linewidth]{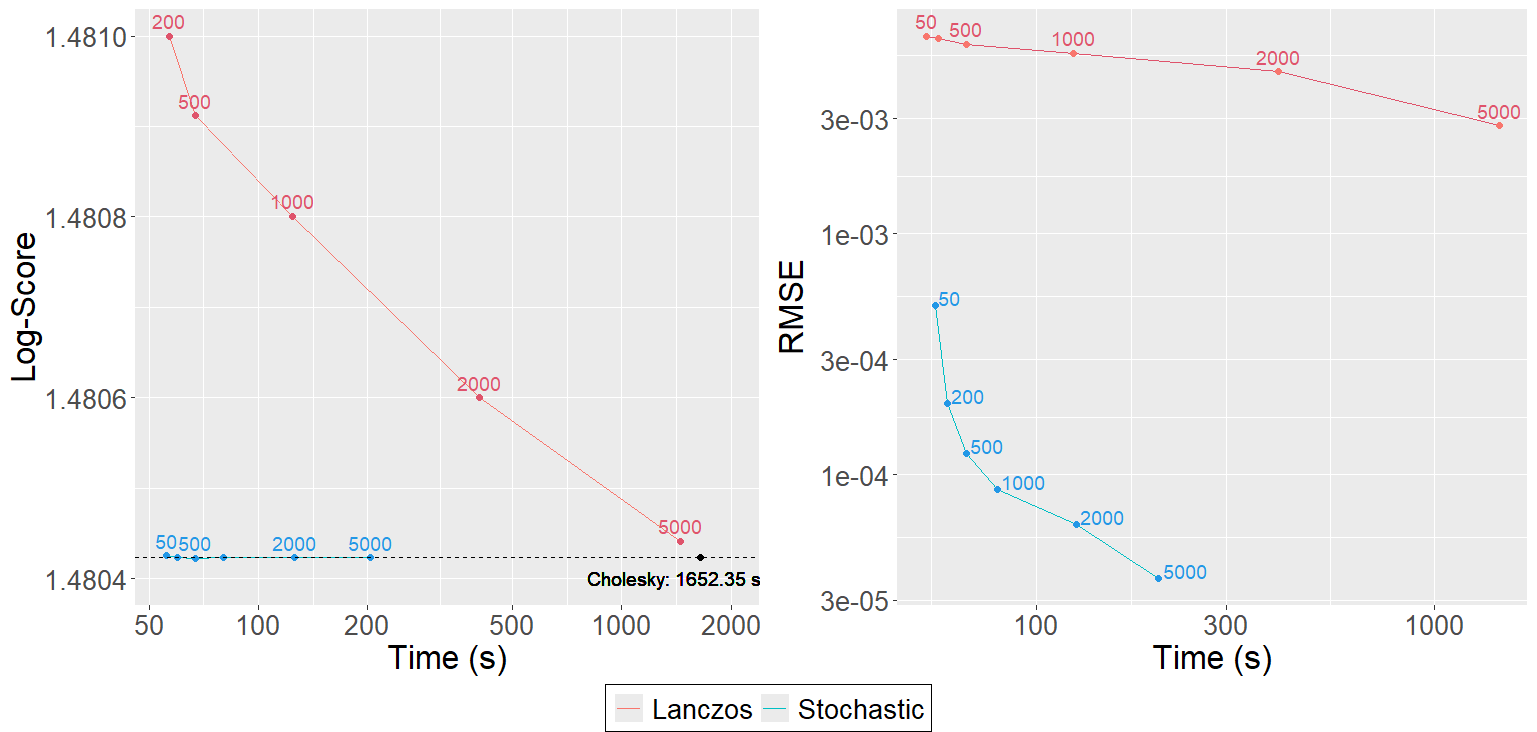}
    \caption{Comparison of the Lanczos and stochastic estimation methods for predictive variances when using an effective range of 0.2. The dashed black line corresponds to the computations based on Cholesky decomposition. The numbers next to the points correspond to the number of sample vectors or the rank, respectively. For the stochastic approach, the respective mean is shown ($n_p = 100'000$, $n = 100'000$, $m = 500$, $n_\gamma = 80$).}
    \label{fig:Pred_Sim}
\end{figure}
The results show that even for a large rank $k = 5'000$, as recommended by \citet{maddox2021iterative}, the Lanczos method is not able to approximate the predictive variances accurately. This limitation arises from the presence of numerous small eigenvalues in the matrix $\Tilde{\mathbf{\Sigma}}_{\mathrm{s}}$ and, generally, in covariance matrices. On the other hand, the stochastic approach approximates the predictive variance accurately with low variances in its approximations. Therefore, for the following experiments, we use the stochastic approach.

\subsection{Computational time and accuracy of log-marginal likelihoods, parameter estimates, and predictive distributions}\label{subsect:sim_all}
Next, we analyze the runtime and accuracy of log-marginal likelihoods, parameter estimates, and predictive distributions of iterative methods compared to Cholesky-based calculations. Parameter estimates are obtained by minimizing the negative log-likelihood using a limited-memory BFGS (LBFGS) algorithm. We have also tried using Fisher scoring in two variants: one utilizing Cholesky-based computations and another one employing iterative techniques as described in Section \ref{SectionIterPE}, but both options were slower than the LBFGS algorithm (results not shown).
For the stochastic estimation of the predictive variance, we use $500$ Rademacher sample vectors. For evaluating predictive distributions, we calculate the RMSE, the log-score, and the continuous ranked probability score (CRPS) given by
\begin{align*}
    \frac{1}{n_p} \sum_{i=1}^{n_p} \boldsymbol{\sigma}_{p,i}\Bigg(\frac{-1}{\sqrt{\pi}}+2\cdot \phi\Big(\frac{\boldsymbol{y}^*_i-\boldsymbol{\mu}^p_{\dagger,i}}{{\boldsymbol{\sigma}_{p,i}}};0,1\Big)+\frac{\boldsymbol{y}^*_i-\boldsymbol{\mu}^p_{\dagger,i}}{{\boldsymbol{\sigma}_{p,i}}}\bigg(2\cdot \Phi\Big(\frac{\boldsymbol{y}^*_i-\boldsymbol{\mu}^p_{\dagger,i}}{{\boldsymbol{\sigma}_{p,i}}};0,1\Big)-1\bigg)\Bigg),
\end{align*}
where $\Phi(x;\mu,\sigma)$ is the cumulative distribution function of a normal distribution with mean $\mu$ and variance $\sigma^2$. 

In Figure \ref{fig:Runtime_k}, we report the runtimes for calculating the negative log-likelihood using both Cholesky-based and iterative methods with and without the FITC preconditioner ($\ell = 50$). We vary the sample sizes $n$, taper ranges $\gamma$, and numbers of inducing points $m$ in turn, while keeping the other quantities fixed at $n_\gamma = 80$, $m = 500$, and $n = 100'000$. We observe that the speedup of the iterative approach with the FITC preconditioner, relative to Cholesky-based computations, grows with the sample size $n$ and the average number of nonzero entries per row $n_\gamma$. The speedup for different numbers of inducing points $m$ is constant. Furthermore, the runtime for the iterative calculations when using the FITC preconditioner exhibits a linear growth pattern concerning both $n$ and $n_\gamma$ and a nonlinear growth concerning $m$, consistent with the theoretical computational complexity of $\mathcal{O}\big(n\cdot(m^2 + m\cdot t + n_\gamma\cdot t)\big)$. On the other hand, the runtime for the Cholesky-based calculations is not linear in $n$, consistent with the theoretical computational complexity of $\mathcal{O}\big(n\cdot(m^2 + n_\gamma^2) + n^{3/2}\big)$. 
\begin{figure}[ht!]
    \centering
    \includegraphics[width=\linewidth]{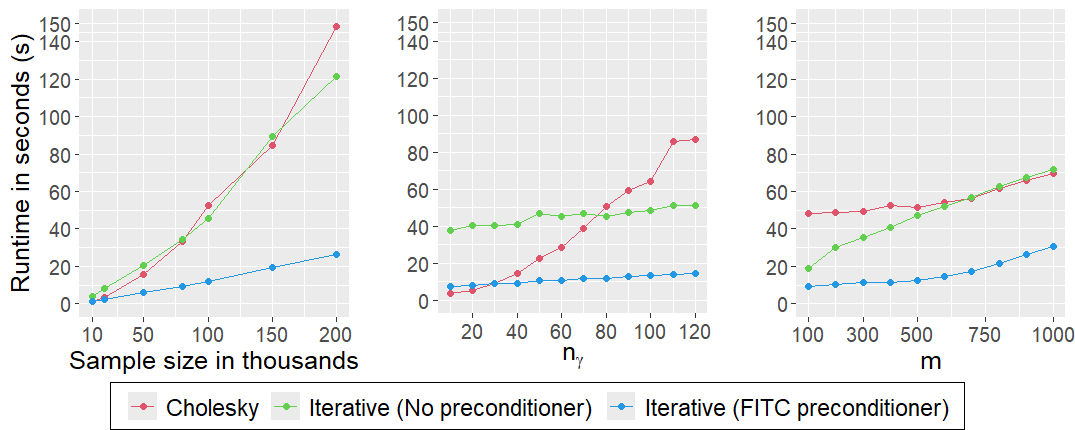}
    \caption{Time (s) for computing the negative log-likelihood using a Cholesky decompostion and iterative methods for simulated data for varying samples sizes $n$, taper ranges $\gamma$, and numbers of inducing points $m$.}
    \label{fig:Runtime_k}
\end{figure}

In Figure \ref{fig:NEGLL_Var}, we show the relative error in the negative log-likelihood calculated with our iterative methods compared to Cholesky-based calculations on simulated random fields with effective ranges 0.05, 0.2, and 0.5, respectively, for the true population parameters. We consider both the FITC preconditioner and the CG method without a preconditioner. We observe only minor deviations between the Cholesky and iterative methods-based log-likelihoods, and the preconditioner results in variance reduction. Concerning the latter, we find that the range parameter influences the amount of variance reduction. We have also compared Rademacher-distributed and normally distributed sample vectors, but no significant differences in the stochastic estimates were evident (results not shown).
\begin{figure}[ht!]
    \centering
    \includegraphics[width=\linewidth]{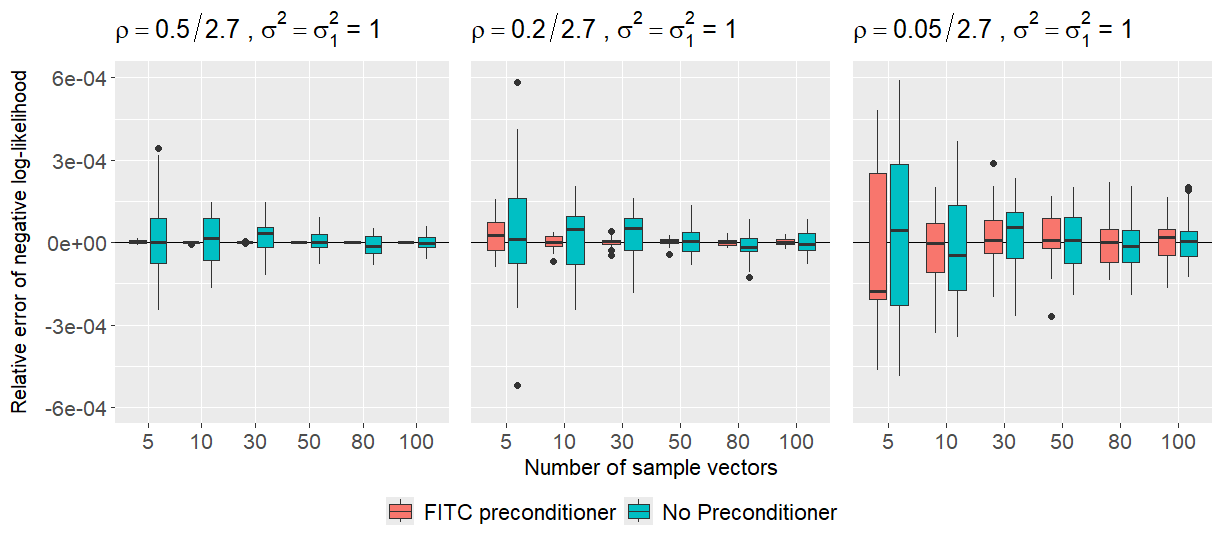}
    \caption{Box-plots of the relative error of the negative log-likelihood computed with and without the FITC preconditioner on simulated random fields with effective ranges of 0.5, 0.2, and 0.05, respectively, for the true population parameters ($n = 100'000$, $n_\gamma = 80$, $m = 500$).}
    \label{fig:NEGLL_Var}
\end{figure}

Next, we analyze the accuracy of the stochastic approximations of the derivatives of the negative log-likelihood. We consider the relative error of iterative methods compared to that of Cholesky-based calculations. This analysis is carried out for the true population parameters on simulated random fields with an effective range 0.05, 0.2, and 0.5, respectively. For the iterative methods, we consider three distinct approaches: employing no preconditioner, as well as using the FITC preconditioner and variance reduction techniques with $\hat{c}_\text{opt} = 1$ and $\hat{c}_\text{opt}$, as outlined in Appendix \ref{contvarapp}. Figure \ref{GradNEGLL} shows the results. The FITC preconditioner method results in lower variance, with a slight further reduction observed when employing the optimal variance reduction parameter $\hat{c}_\text{opt}$.

\begin{figure}[ht!]
    \centering
    \includegraphics[width=\linewidth]{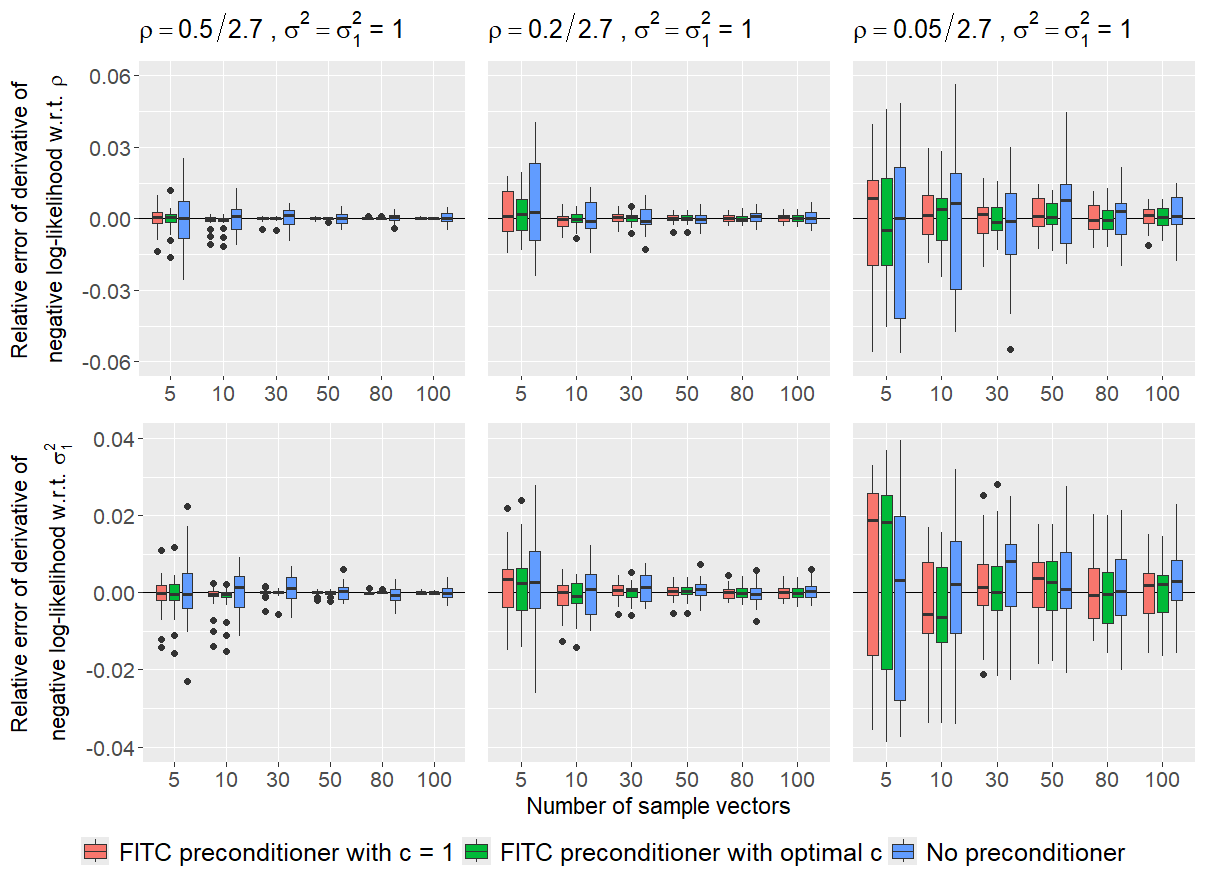}
    \caption{Box-plots of the relative error of the derivatives of the negative log-likelihood with respect to $\rho$ (top row) and $\sigma^2_1$ (bottom row) computed without a preconditioner, with the FITC-preconditioner, and additionally with estimating $\hat{c}_\text{opt}$ for the effective ranges of 0.5, 0.2, and 0.05 (left to right) and the true population parameters  ($n = 100'000$, $n_\gamma = 80$, $m = 500$).}\label{GradNEGLL}
\end{figure}

In Figure \ref{fig:Comp_GP_Inf}, we present the estimated covariance parameters and the measures for evaluating the predictive distributions over 10 simulated random fields for an effective range of 0.2. We observe almost identical parameter estimates and prediction accuracy measures for iterative methods and Cholesky-based computations. Further, iterative methods are faster by approximately one order of magnitude compared to Cholesky-based calculations for both parameter estimation and prediction. Specifically, we find average speedups of 7.5, 27.2, and 11.2 for estimating the parameters, computing the predictive distribution, and both together, respectively; see Table \ref{Table5} in Appendix \ref{Appfull}. In Table \ref{Table5} in Appendix \ref{Appfull}, we also report the bias and RMSE of the estimated covariance parameters for each of the effective ranges 0.5, 0.2, and 0.05, as well as the averages and standard errors of the RMSE of the predictive mean, the log-score, and the CRPS.


\begin{figure}[ht!]
     \centering
     \includegraphics[width=\linewidth]{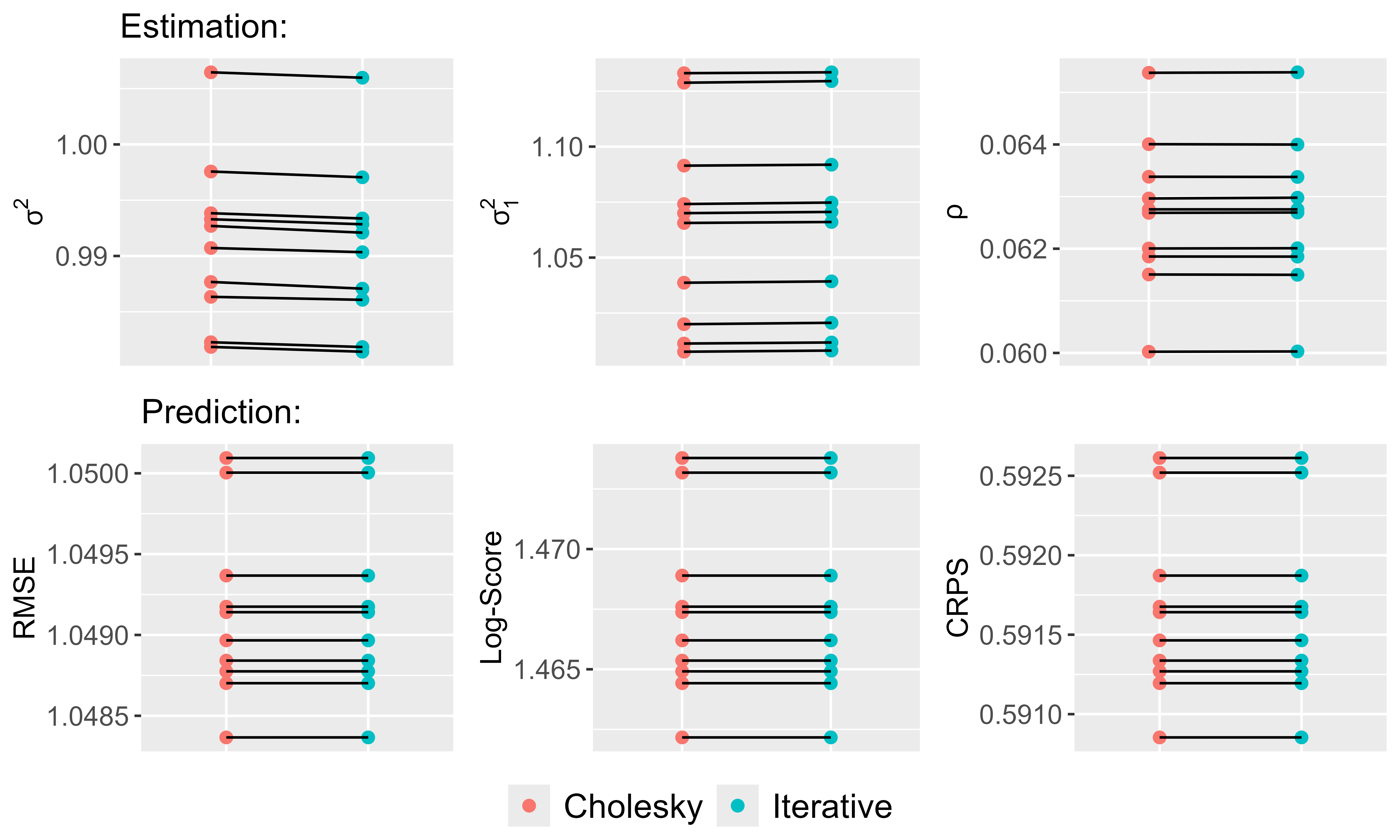}
     \caption{Comparison of the estimated covariance parameters ($\sigma^2$, $\sigma^2_1$, and $\rho$ in top row) and predictive distributions (RMSE, log-score, and CRPS in bottom row) between Cholesky-based and the iterative computations over 10 simulated random fields ($n_p = 100'000$, $n = 100'000$, $m = 500$, $n_\gamma = 80$).}
     \label{fig:Comp_GP_Inf}
 \end{figure}
\section{The FITC preconditioner for Vecchia approximations}\label{sectVecchia}
The FITC preconditioner proposed in this article is useful for not just FSAs but also other GP approximations and exact GPs when using iterative methods. To demonstrate this, we show in the following how the FITC preconditioner can be used in iterative methods for Vecchia approximations \citep{vecchia1988estimation, datta2016hierarchical, katzfuss2021general, schaefer2021compression, kundig2024iterative}, and we compare it to other state-of-the-art preconditioners for Vecchia approximations. 

\subsection{Vecchia approximations}
Vecchia approximations result in approximate sparse (reverse) Cholesky factors of precision matrices. We apply a Vecchia approximation to the latent GP $\boldsymbol{b} \sim \mathcal{N}(\boldsymbol{0}, \mathbf{\Sigma})$ by approximating the density $p(\boldsymbol{b} \mid \boldsymbol{\theta})$ as $p(\boldsymbol{b} \mid \boldsymbol{\theta}) \approx \prod_{i=1}^n p\left(\boldsymbol{b}_i \mid \boldsymbol{b}_{N(i)}, \boldsymbol{\theta}\right)$, where $\boldsymbol{b}_{N(i)}$ are subsets of $\left(\boldsymbol{b}_1, \ldots, \boldsymbol{b}_{i-1}\right)$, and $N(i) \subseteq\{1, \ldots, i-1\}$ with $|N(i)| \leq m_v$. If $i>m_v+1, N(i)$ is often chosen as the $m_v$ nearest neighbors of $\boldsymbol{s}_i$ among $\boldsymbol{s}_1, \ldots, \boldsymbol{s}_{i-1}$. It follows that $p\left(\boldsymbol{b}_i \mid \boldsymbol{b}_{N(i)}, \boldsymbol{\theta}\right)=\mathcal{N}\left(\mathbf{A}_i \boldsymbol{b}_{N(i)}, \mathbf{D}_i\right)$, where $\mathbf{A}_i=\mathbf{\Sigma}_{N(i), i}^{\mathrm{T}} \mathbf{\Sigma}_{N(i)}^{-1},$ $\mathbf{D}_i=\mathbf{\Sigma}_{i, i}-\mathbf{A}_i \mathbf{\Sigma}_{N(i), i}$, $\mathbf{\Sigma}_{N(i), i}$ is a subvector of $\mathbf{\Sigma}$ with the $i$th column and row indices $N(i)$, and $\mathbf{\Sigma}_{N(i)}$ denotes a submatrix of $\mathbf{\Sigma}$ consisting of rows and columns $N(i)$. Defining a sparse lower triangular matrix $\mathbf{B} \in \mathbb{R}^{n \times n}$ with 1s on the diagonal, off-diagonal entries $\mathbf{B}_{i, N_{(i)}}=-\mathbf{A}_i$ and $0$ otherwise, and a diagonal matrix $\mathbf{D} \in \mathbb{R}^{n \times n}$ with $\mathbf{D}_i$ on the diagonal, one obtains the Vecchia approximation $\boldsymbol{b} \stackrel{\text { approx }}{\sim} \mathcal{N}(\boldsymbol{0}, {\mathbf{\Sigma}}_V)$, where ${\mathbf{\Sigma}}_V^{-1}=\mathbf{B}^{\mathrm{T}} \mathbf{D}^{-1} \mathbf{B}$. Calculating a Vecchia approximation has $\mathcal{O}\left(n\cdot m_v^3\right)$ computational and $\mathcal{O}(n \cdot m_v)$ memory cost. Often, accurate approximations are obtained for small values of $m_v$. Computing linear solves of the form 
\begin{equation}\label{eqvecchia}
(\mathbf{\Sigma}_V^{-1} + \mathbf{W})^{-1}\boldsymbol{v},
\end{equation}
where \(\boldsymbol{v} \in \mathbb{R}^n\) and $\mathbf{W}$ is a diagonal matrix containing negative second derivatives of the log-likelihood, and log-determinants of $\mathbf{\Sigma}_V^{-1} + \mathbf{W}$ are required when applying Vecchia approximations to a latent GP with Gaussian likelihoods and in Vecchia-Laplace approximations for non-Gaussian likelihoods; see \citet{schaefer2021sparse} and \citet{kundig2024iterative} for more details. To apply the FITC preconditioner, we first note that instead of \eqref{eqvecchia}, we can equivalently solve 
\begin{equation}\label{VL_v2}
    \mathbf{W}(\mathbf{\Sigma}_V + \mathbf{W}^{-1})^{-1}\mathbf{\Sigma}_V^{-1}\boldsymbol{v}.
\end{equation}
We propose to compute these linear solves with the preconditioned CG method and the FITC preconditioner given by
\[
\widehat{\boldsymbol{P}} = \mathbf{\Sigma}_{mn}^{\mathrm{T}} \mathbf{\Sigma}_m^{-1} \mathbf{\Sigma}_{mn} + 
\boldsymbol{D}_\mathrm{s} \approx \mathbf{\Sigma} + \mathbf{W}^{-1},
\]
where $\boldsymbol{D}_\mathrm{s} = \operatorname{diag}(\mathbf{\Sigma} - \mathbf{\Sigma}_{mn}^{\mathrm{T}} \mathbf{\Sigma}_m^{-1} \mathbf{\Sigma}_{mn}) + \mathbf{W}^{-1}$. Further, log-determinants of $\mathbf{\Sigma}_V^{-1} + \mathbf{W}$ can be calculated using SLQ approximations as outlined in \citet{kundig2024iterative}. Similarly to the state-of-the-art pivoted Cholesky preconditioner, the FITC preconditioner offers a low-rank approximation. However, the pivoted Cholesky preconditioner is considerably slower to construct limiting the rank to smaller values in practice. In contrast, the FITC preconditioner can be computed more efficiently, and larger ranks can thus be used.

\subsection{Comparison of preconditioners}\label{comparison_PC_Vecchia}
We compare the proposed FITC preconditioner with the following state-of-the-art alternatives: a pivoted Cholesky preconditioner  \citep{harbrecht2012low, gardner2018gpytorch}, a Vecchia approximation with diagonal update (VADU) preconditioner \citep{kundig2024iterative}, an observable Vecchia approximation preconditioner defined below, and a zero fill-in incomplete Cholesky factorization preconditioner \citep{schaefer2021sparse} using the sparsity pattern of $\mathbf{B}^{\mathrm{T}} \mathbf{D}^{-1} \mathbf{B}$; see Appendix A.7 of \citet{kundig2024iterative} for a precise definition of the incomplete Cholesky factorization algorithm we use. Alternatively, the incomplete Cholesky preconditioner could also use the sparsity pattern of $\mathbf{B}$. While this leads to faster runtimes, it often yields breakdowns \citep{kundig2024iterative}. In fact, we occasionally observe breakdowns even when using the sparsity pattern of $\mathbf{B}^{\mathrm{T}} \mathbf{D}^{-1} \mathbf{B}$, depending on the simulated locations. An insightful reviewer gave us the idea to use a Vecchia approximation to the response variable precision $(\mathbf{\Sigma} + \sigma^2 \boldsymbol{I}_n)^{-1}$ as preconditioner for Gaussian likelihoods. We extend this to non-Gaussian likelihoods by using a Vecchia approximation for $(\mathbf{\Sigma}_V + \mathbf{W}^{-1})^{-1}$ in \eqref{VL_v2}, interpreting $\mathbf{W}^{-1}$ as pseudo nugget effect. We denote this preconditioner as the ``observable Vecchia approximation" preconditioner. To the best of our knowledge, this preconditioner is novel and has not been used before. The idea of using a Vecchia approximation as a preconditioner in the CG method has first been proposed in \citet{stroud2017bayesian}, who have applied this for a model with a Gaussian likelihood and a spectral approximation.

The comparison is conducted in terms of the runtime and the variance of marginal likelihood approximations. Specifically, we replicate the experiments of \citet[Section 4.2 and Appendix A.8]{kundig2024iterative}, where the response variable $\boldsymbol{y} \in \{0,1\}^n$ follows a Bernoulli likelihood. The experiments are conducted for varying numbers of sample vectors \(\ell\). For the pivoted Cholesky algorithm, we use a rank of \(k = 50\) following \citet{maddox2021iterative} and \citet{kundig2024iterative}. For the FITC preconditioner, we use \(m = 200\); see below for a sensitivity analysis regarding this tuning parameter. The marginal likelihood is computed at the true covariance parameters and repeated 100 times using different random vectors to assess variability. 

\begin{figure}[ht!]
    \centering
    \includegraphics[width=\linewidth]{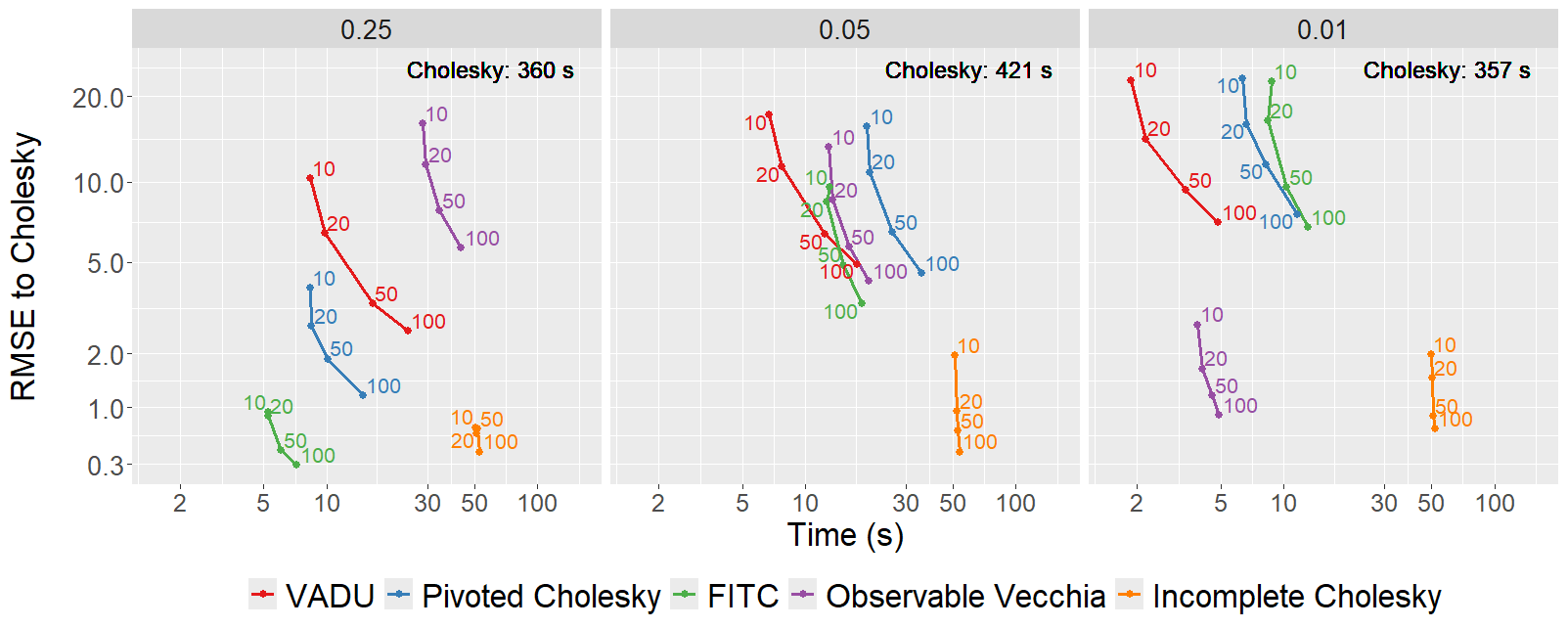}
    \caption{Accuracy-runtime comparison of preconditioners: RMSE between log‑marginal likelihoods computed using iterative methods and those computed using a Cholesky decomposition versus runtime for the VADU, pivoted Cholesky, FITC, observable Vecchia, and zero fill-in incomplete Cholesky preconditioners and varying numbers of sample vectors $\ell$ (annotated in the plot). A binary likelihood, $n = 100'000$ samples, and three range parameters $\rho \in \{0.25, 0.05, 0.01\}$ (from left to right) are used.}
    \label{fig:FITC_vs_PC_0.2}
\end{figure}
In Figure \ref{fig:FITC_vs_PC_0.2}, we show the RMSE between log‑marginal likelihoods computed using iterative methods and those computed using a Cholesky decomposition versus runtimes for the different preconditioners when using varying numbers of sample vectors $\ell$ and three different range parameters \(\rho \in \{0.25, 0.05, 0.01\}\). Note that in relative terms, all differences observed are small. For instance, a log-marginal likelihood difference in $10$ corresponds to a relative error of approximately $10^{-4}$. Overall, the results show that the FITC preconditioner yields a good trade-off between runtime and accuracy. In detail, for the large range ($\rho = 0.25$), the FITC preconditioner is the fastest preconditioner while being more accurate than the VADU, observable Vecchia, and pivoted Cholesky preconditioners. For the intermediate and small ranges, the VADU, pivoted Cholesky, and FITC preconditioners are approximately equally accurate. For the intermediate range ($\rho=0.05$), the VADU preconditioner is slightly faster than the FITC preconditioner, and the pivoted Cholesky preconditioner is slower. For the small range ($\rho=0.01$), all preconditioners, except for the incomplete Cholesky preconditioner, are relatively fast. The incomplete Cholesky preconditioner is highly accurate but very slow for all ranges. This means that a similar level of accuracy can often be achieved with the FITC and VADU preconditioners with faster runtimes when using more sample vectors. The observable Vecchia preconditioner yields a similar level of accuracy as that of the FITC and VADU preconditioners for the intermediate range and is very accurate for the small range. Furthermore, the runtimes of the observable Vecchia preconditioner are comparable to the FITC and VADU preconditioners for these intermediate and small ranges. For the large range, however, the observable Vecchia preconditioner is both slow and inaccurate. In summary, for small ranges, all preconditioners, except for the incomplete Cholesky preconditioner, are fast, and one can obtain computationally efficient and accurate approximations by using a larger number of sample vectors. For large ranges, the FITC preconditioner stands out as being both very accurate and fast.

We also briefly analyze the performance of the FITC preconditioner for varying numbers of inducing points. The marginal likelihood is computed at the true covariance parameters and repeated 100 times with 50 random vectors used in each iteration. The results for the range parameter \(\rho = 0.05\) are displayed in Figure \ref{fig:FITC_vs_PC_005_rank} in Appendix \ref{vecchia_app}. We find that the FITC preconditioner achieves the fastest runtime for \(m = 200\).

\section{Real-world application}\label{sect5}

In the following, we apply our proposed methods to a large spatial satellite data set and compare the full-scale approximation to a modified predictive process, or FITC, approximation, covariance tapering, and two variants of Vecchia approximations. In the first variant, a Vecchia approximation is applied to the latent process, and iterative methods with the FITC preconditioner defined in \eqref{eqfitcp} are used. This approach is motivated by the screening effect, which suggests that for spatial processes, predictions at a location primarily depend on nearby observations. However, perturbations such as i.i.d. measurement noise could weaken this effect \citep{katzfuss2021general, geoga2024scalable}. The second variant, referred to as ``observable" approximation, applies a Vecchia approximation to the covariance matrix including the nugget effect enabling efficient calculations without iterative methods.

The data set comprises observations of daytime land surface temperatures in degrees Celsius acquired on August 20, 2023, by the Terra instrument aboard the MODIS satellite (Level-3 data) with a spatial resolution of 1 kilometer (km) on a $1200\times1200$ km grid, encompassing longitude values from $-95.91153$ to $-91.28381$ and latitude values from 34.29519 to 37.06811. The selection of latitude, longitude ranges, and date was driven primarily by the minimal cloud cover over the region on that specific date. To generate out-of-sample test data, we designated locations covered by clouds on a different day (September 6, 2023) as test locations. Subsequently, we randomly sample $400'000$ data points from the uncovered data set to constitute the training set, while $200'000$ data points were randomly selected from the covered data set to form the test set. This subsampling enables us to conduct Cholesky-based computations within a reasonable time. Figure \ref{fig:Day_Temp}  shows the full and the training data sets, and Figure \ref{fig:four graphs} in Appendix \ref{realwapp} shows the test data. For the sake of reproducibility, both the data and the downloading code are available on GitHub at \url{https://github.com/TimGyger/iterativeFSA}. 

\begin{figure}[htbp]
     \centering
     \begin{subfigure}[b]{0.49\textwidth}
         \centering
         \includegraphics[width=\textwidth]{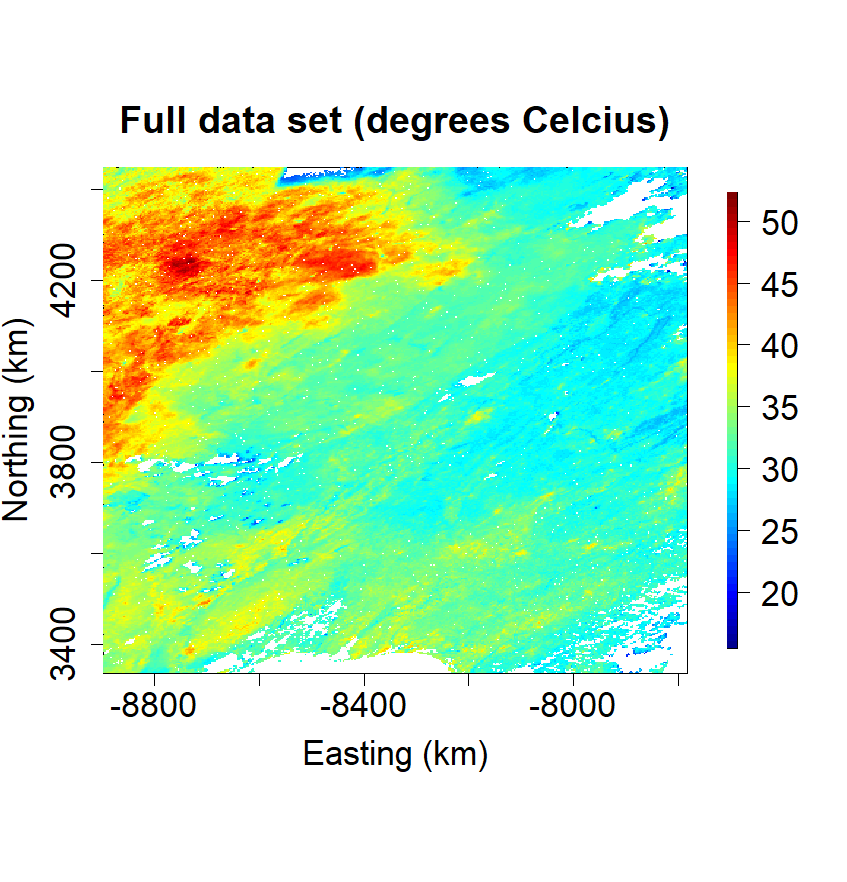}
     \end{subfigure}
     \hfill
     \begin{subfigure}[b]{0.49\textwidth}
         \centering
         \includegraphics[width=\textwidth]{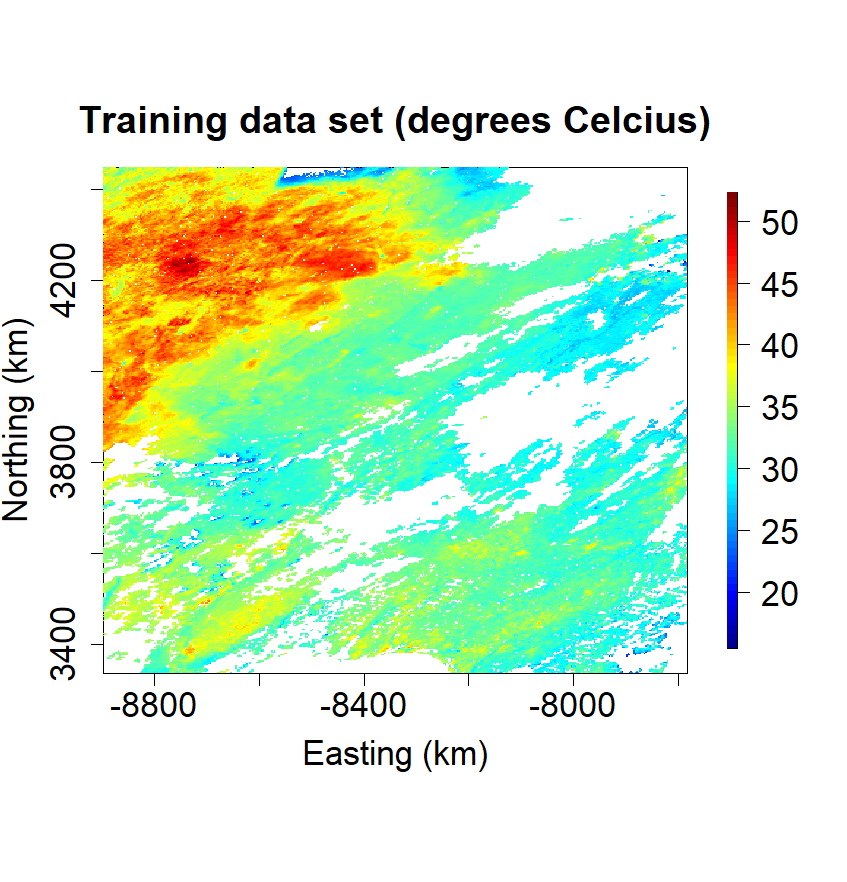}
     \end{subfigure}
     \caption{The full (left) and training (right) satellite data set.}
    \label{fig:Day_Temp}
\end{figure}

For the Vecchia approximations, we use $m_v = 30$ neighbors. Furthermore, we use $m=500$ inducing points determined by the kmeans++ algorithm for both the full-scale and FITC approximations, and we set the taper range to $\gamma = 7500$, resulting in an average of $n_\gamma = 80$ nonzero entries per row for both the full-scale and covariance tapering approximations. The mean function is assumed to be linear in the coordinates, $F(\boldsymbol{X}) = \boldsymbol{X}\boldsymbol{\beta}$, and we use the coordinates as covariates. In our iterative approaches with the FITC preconditioner, we use a convergence tolerance level $\delta = 1$ in the CG algorithm, and we use $\ell=50$ Gaussian sample vectors for the stochastic estimates. Below, we also investigate using both lower and higher convergence tolerance levels, as well as less random vectors, and obtain virtually identical results. Furthermore, for the stochastic estimation of the predictive variance, we use 500 Rademacher sample vectors. For parameter estimation, we use the same optimization method and settings as those in Subsection \ref{subsect:sim_all}. 

\begin{table}[ht!]
\centering
\begin{tabular}{ |p{0.3cm}|p{1.7cm}||p{1.5cm}|p{1.5cm}|p{1.5cm}|p{1.5cm}|p{1.5cm}|p{1.5cm}|}
\hline
 \multicolumn{2}{|c||}{\multirow{2}{*}{}}&\multicolumn{2}{c|}{FSA}& Tapering & FITC&\multicolumn{2}{c|}{Vecchia}\\
 \cline{3-8} 
 \multicolumn{2}{|c||}{}&{\small Iterative}& {\small Cholesky} & {\small Cholesky} & {\small Cholesky} & {\small Latent}&{\small Observ.} \\
 \hhline{|=|=|=|=|=|=|=|=|}
\multirow[c]{8}{*}[0in]{\rotatebox{90}{\textbf{Estimation}}} &
 $\beta_{\text{intercept}}$ & -26.30  & -26.18&  -48.32 & 61.51& -49.57 & -50.27\\
 &$\beta_{\text{east}}$ & -0.0066 & -0.0066 &  -0.0076 & -0.0200 & -0.0080 & -0.0080\\
 &$\beta_{\text{north}}$ & 0.00090 & 0.00090 &  0.00484 & -0.0513 & 0.00438 & 0.00439\\
 &$\sigma^2$ & 0.1428 & 0.1429 & 0.01968 & 1.909 & 0.1279 & 0.1298\\
 &$\sigma^2_1$& 4.093 &4.099 & 4.527 & 111.1 &5.685 & 5.738\\ 
  &$\rho$& 24.13 & 24.16 & 5.494 & 211.1 & 5.044 & 5.089\\
 \cline{2-8} 
  &Time (s)& 5054 & 33212 &  13086 & 12445  & 9733 & 426\\ 
  &Speedup & 6.6  & & 2.5  & 2.7  & 3.4 & 78.0\\
 \hhline{|=|=|=|=|=|=|=|=|}
 \multirow[c]{5}{*}[0in]{\rotatebox{90}{\textbf{Prediction}}} &
 RMSE &  1.4177 &1.4175 & 2.3349 & 1.8171 & 1.9182 & 1.9155\\
 \cline{2-8} 
  &LS &  1.7244 &1.7245 & 2.1735 & 1.9444 & 1.8749 & 1.8737\\
 \cline{2-8} 
  &CRPS &  0.75451 &0.75449 &  1.2723 & 0.94036 & 0.98518 & 0.98351\\
 \cline{2-8} 
  &
 Time (s) &  587 & 5513 &  199 & 68 & 1631 & 485\\
  &Speedup & 9.4  &  & 27.8 & 81.6 &  3.4 & 11.4\\
 \hhline{|=|=|=|=|=|=|=|=|}
 \multicolumn{2}{|c||}{ Total time (s)} &  5641 & 38725 &  13285 &12512 & 11364 & 911\\ 
  \multicolumn{2}{|c||}{Speedup} & 6.9  &  & 2.9 &3.1 & 3.4 & 42.5\\
 \hline
\end{tabular}
\caption{Real-world satellite data results.}\label{TableCG2} 
\end{table}
In Table \ref{TableCG2}, we report the estimated parameters, the prediction accuracy measures, and computation times. We find that for the full-scale approximation, the iterative methods result in almost identical parameter estimates and prediction accuracy measures while being approximately one order of magnitude faster compared to Cholesky-based computations. As expected, the FITC, covariance tapering, and Vecchia approximations are faster than the full-scale approximation for the Cholesky-based computations. However, these approximations have considerably worse prediction accuracy than the full-scale approximation. This is most likely due to the inability of these two approximations to capture both the short-scale and large-scale dependencies present in the data. Moreover, the results show that both Vecchia variants produce nearly identical parameter estimates and prediction accuracy measures, but the variant where the approximation is applied to the observable process is substantially faster. That is, although the noise variance could potentially weaken the screening effect, we do not find evidence that this leads to a less accurate approximation on this data set.

\begin{figure}[ht!]
     \centering
     \begin{subfigure}[b]{0.49\textwidth}
         \centering
         \includegraphics[width=\textwidth]{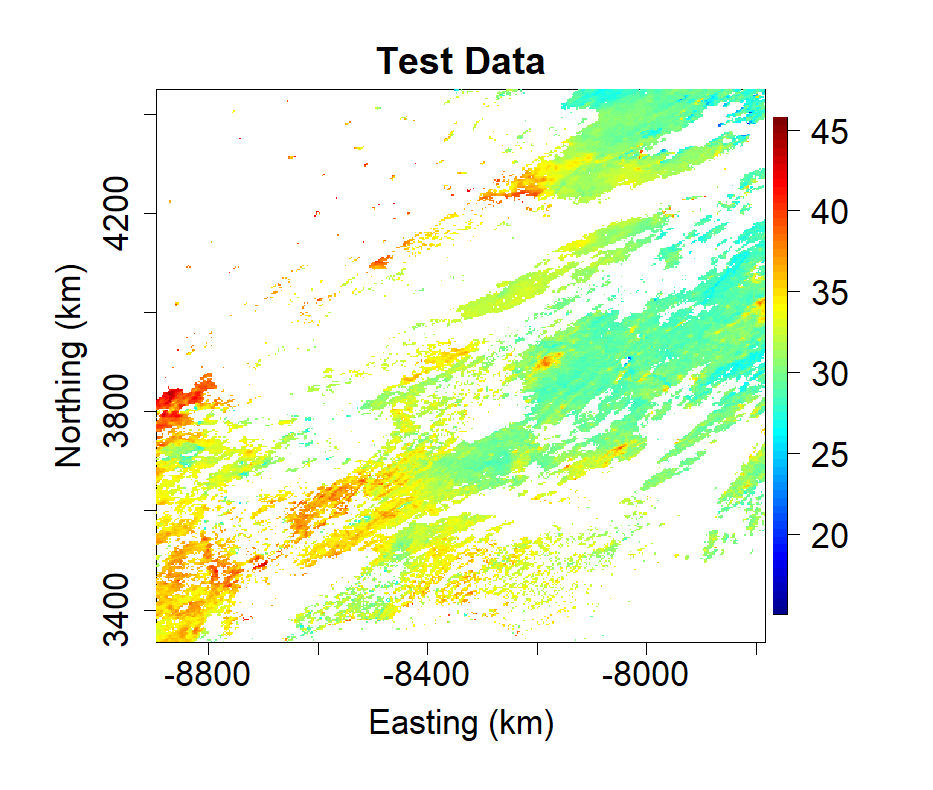}
     \end{subfigure}
     \hfill
     \begin{subfigure}[b]{0.49\textwidth}
         \centering
         \includegraphics[width=\textwidth]{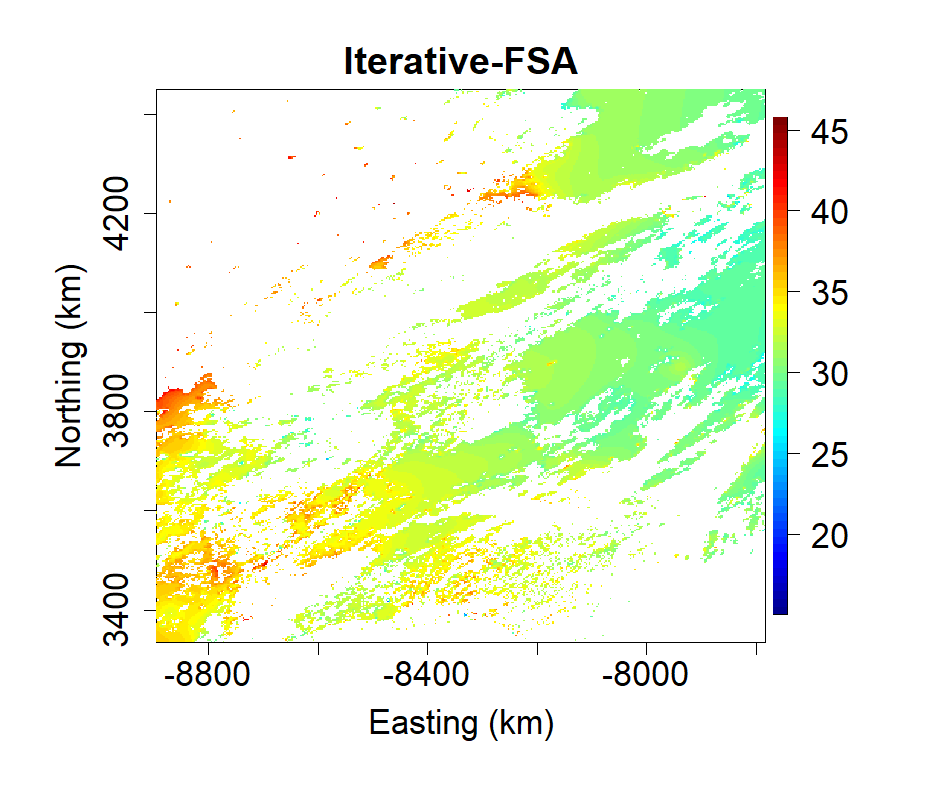}
     \end{subfigure}
     \hfill
     \begin{subfigure}[b]{0.49\textwidth}
         \centering
         \includegraphics[width=\textwidth]{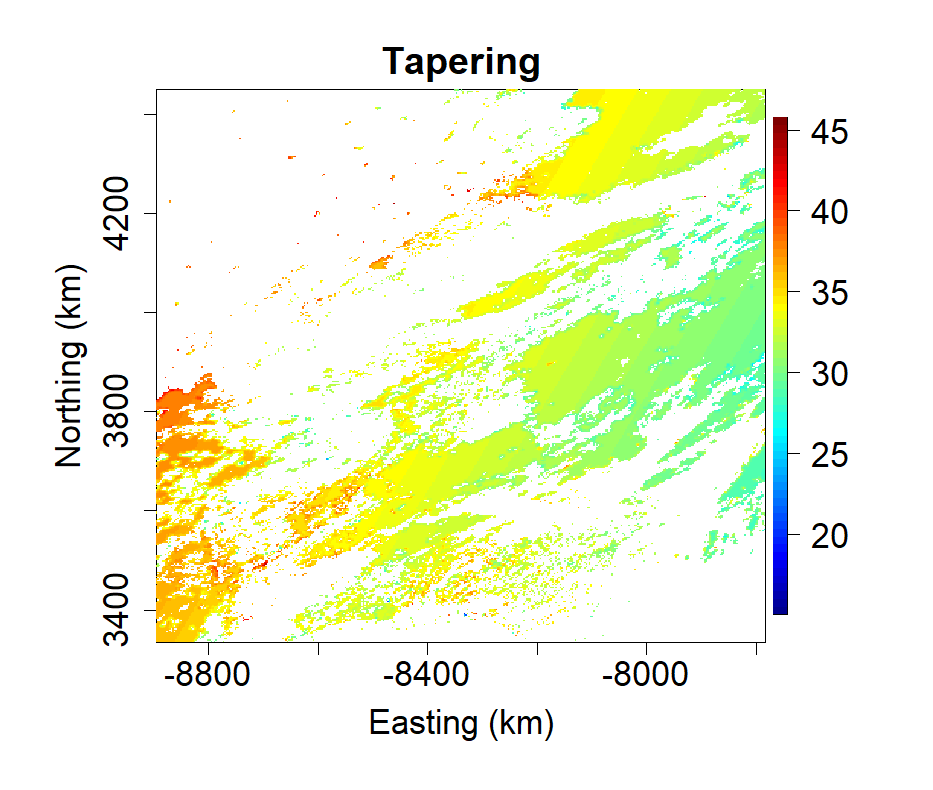}
     \end{subfigure}
     \hfill
     \begin{subfigure}[b]{0.49\textwidth}
         \centering
         \includegraphics[width=\textwidth]{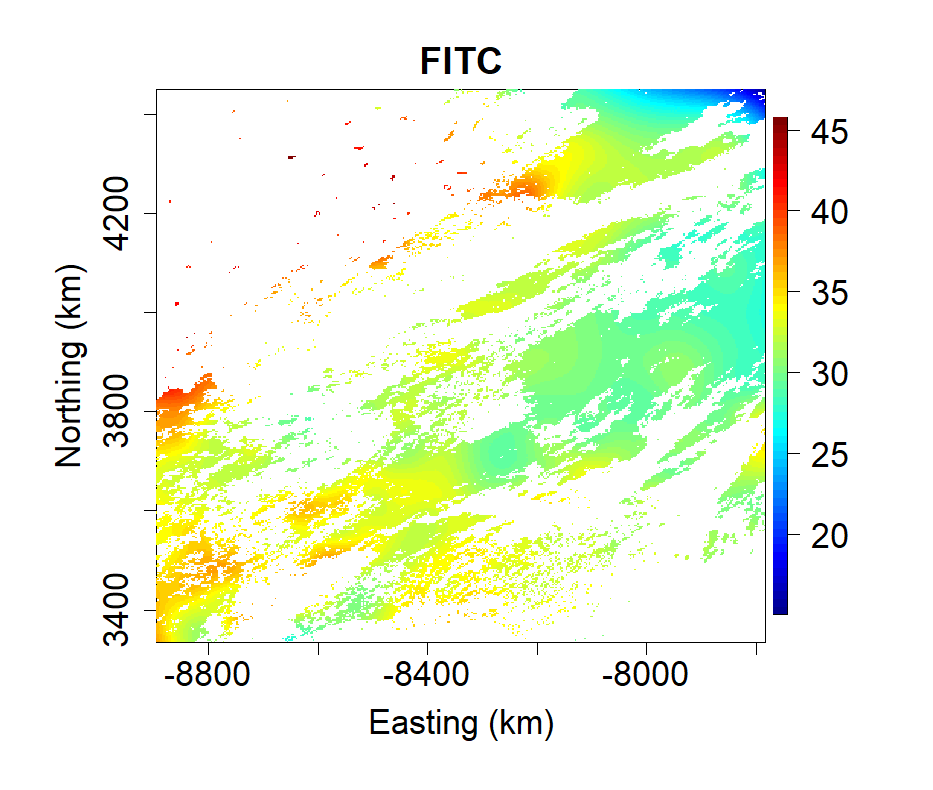}
     \end{subfigure}
     \hfill
     \begin{subfigure}[b]{0.49\textwidth}
         \centering
         \includegraphics[width=\textwidth]{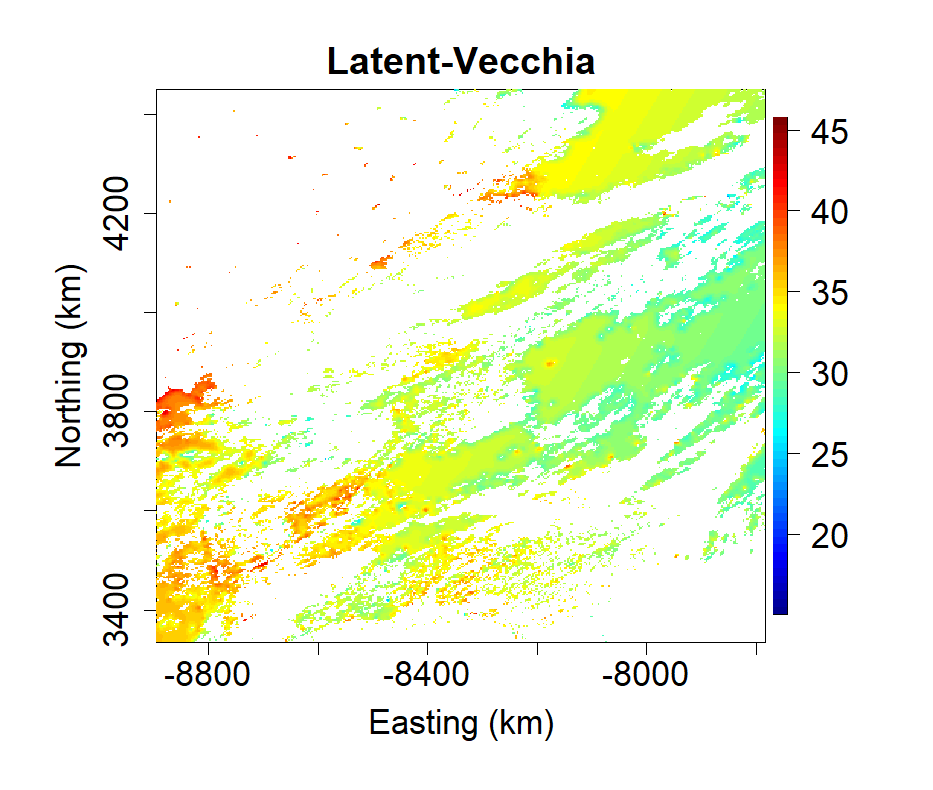}
     \end{subfigure}
     \hfill
     \begin{subfigure}[b]{0.49\textwidth}
         \centering
         \includegraphics[width=\textwidth]{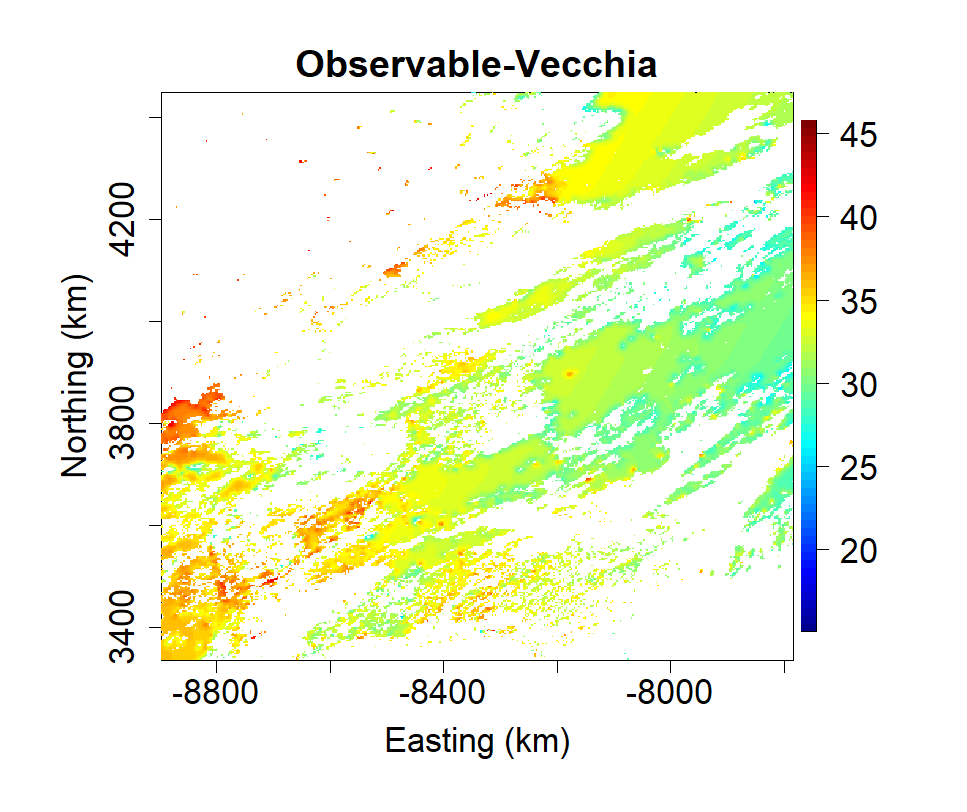}
     \end{subfigure}
     \caption{Test data and predictive means in the FSA (iterative), tapering, FITC, and Vecchia (latent and observable) framework.}
        \label{fig:four graphs}
\end{figure}

\begin{figure}[ht!]
     \centering
     \begin{subfigure}[b]{0.49\textwidth}
         \centering
         \includegraphics[width=\textwidth]{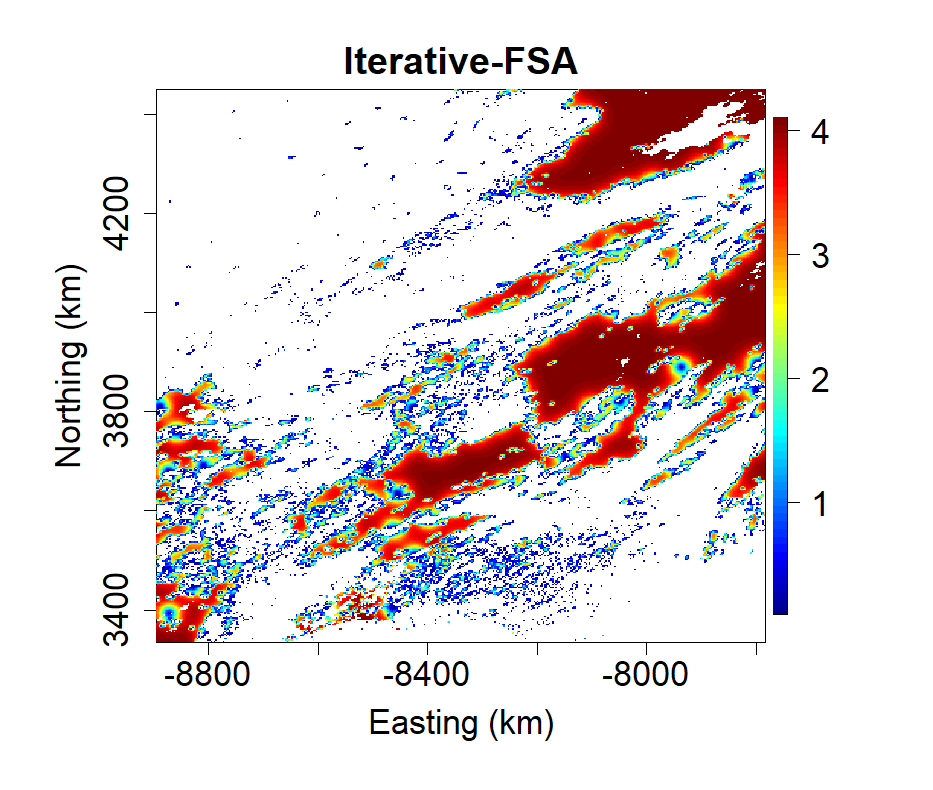}
     \end{subfigure}
     \hfill
     \begin{subfigure}[b]{0.49\textwidth}
         \centering
         \includegraphics[width=\textwidth]{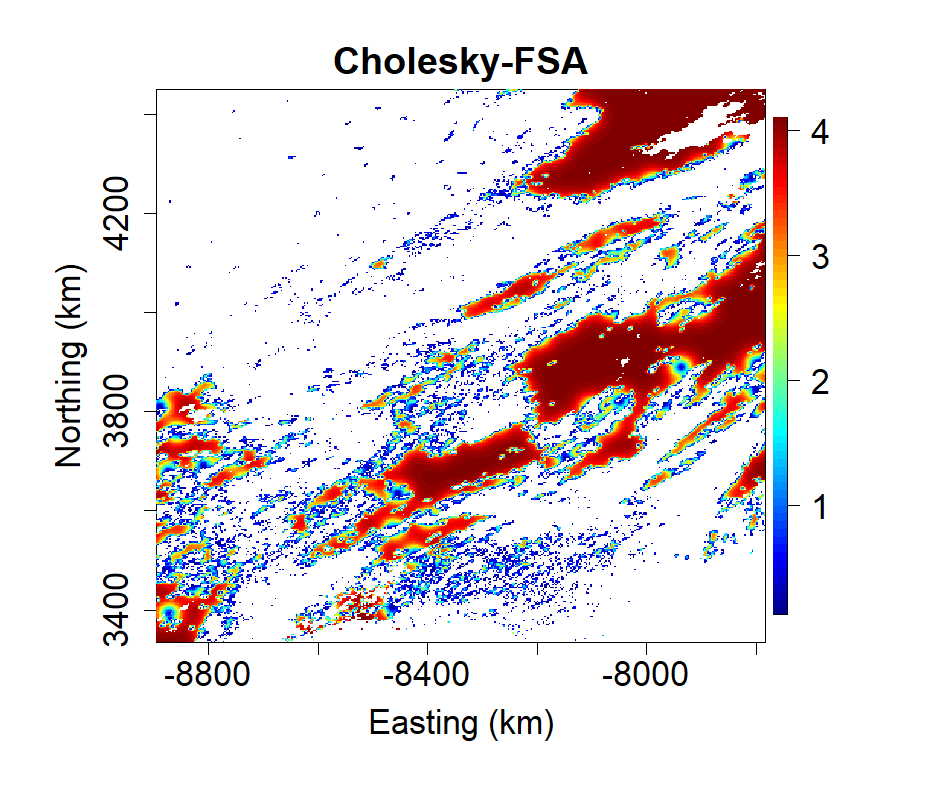}
     \end{subfigure}
     \hfill
     \begin{subfigure}[b]{0.49\textwidth}
         \centering
         \includegraphics[width=\textwidth]{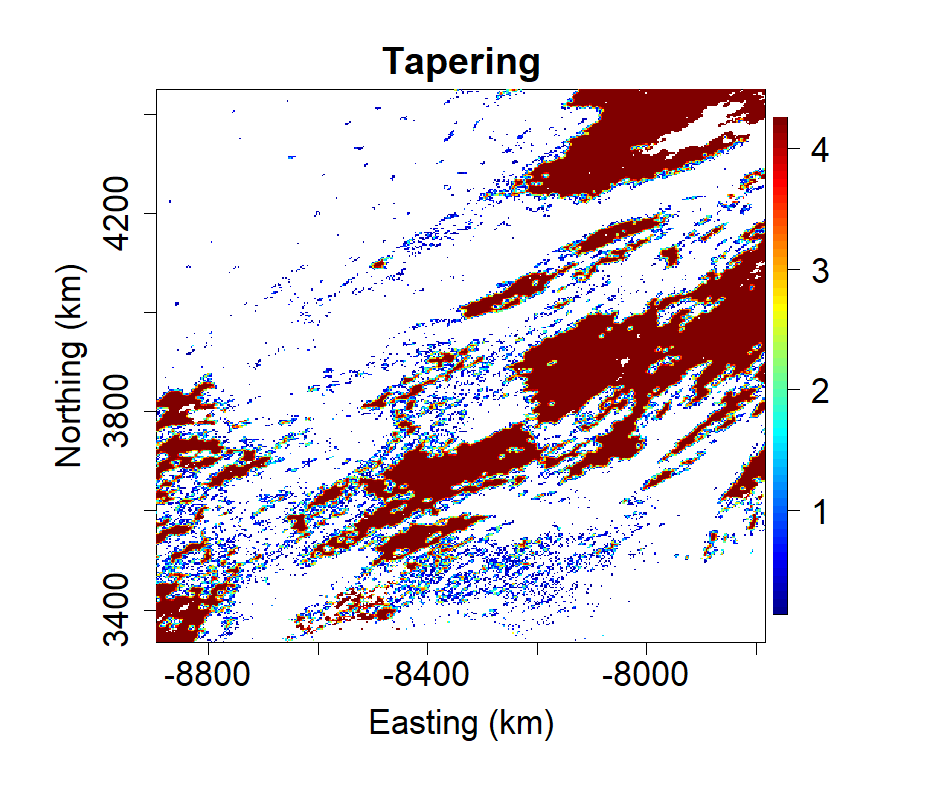}
     \end{subfigure}
     \hfill
     \begin{subfigure}[b]{0.49\textwidth}
         \centering
         \includegraphics[width=\textwidth]{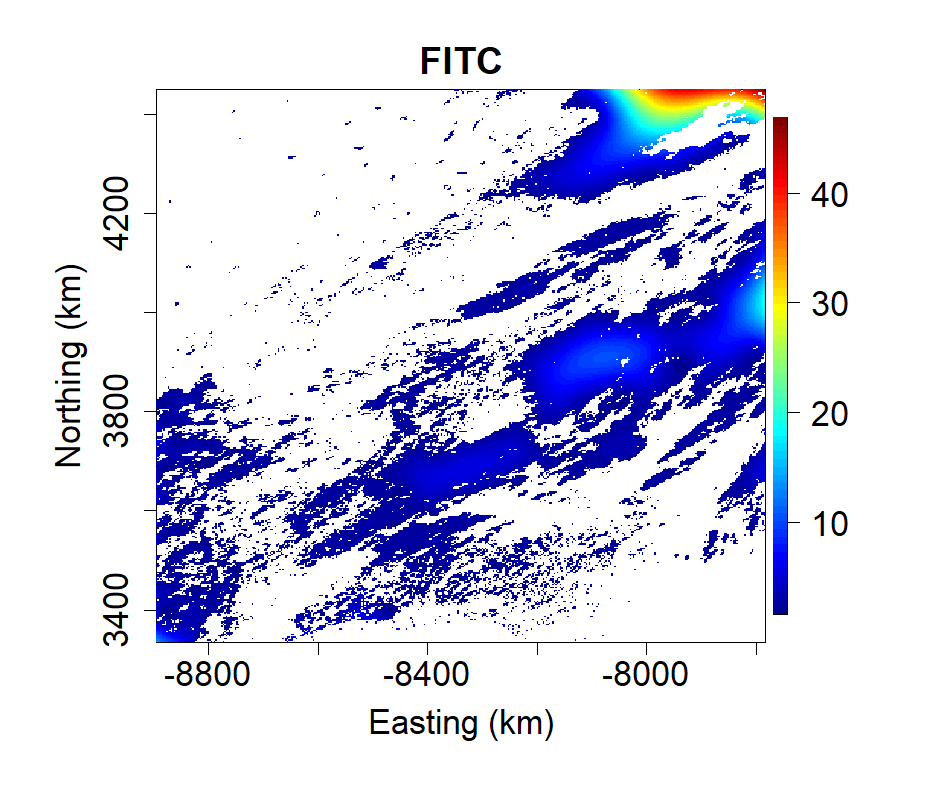}
     \end{subfigure}
     \hfill
     \begin{subfigure}[b]{0.49\textwidth}
         \centering
         \includegraphics[width=\textwidth]{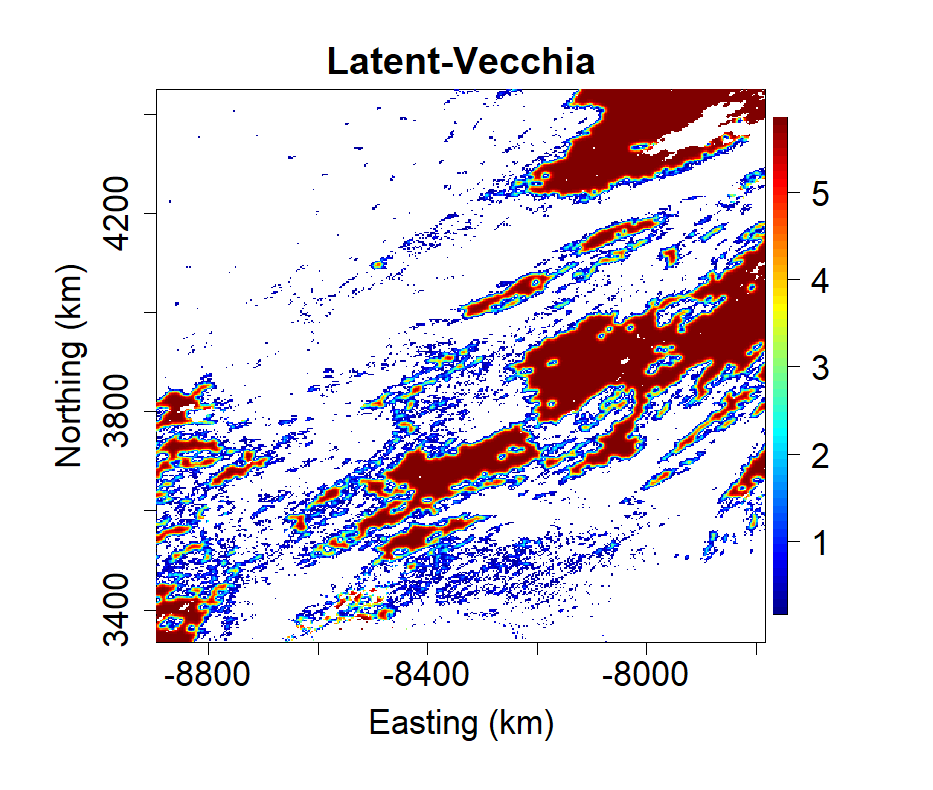}
     \end{subfigure}
     \hfill
     \begin{subfigure}[b]{0.49\textwidth}
         \centering
         \includegraphics[width=\textwidth]{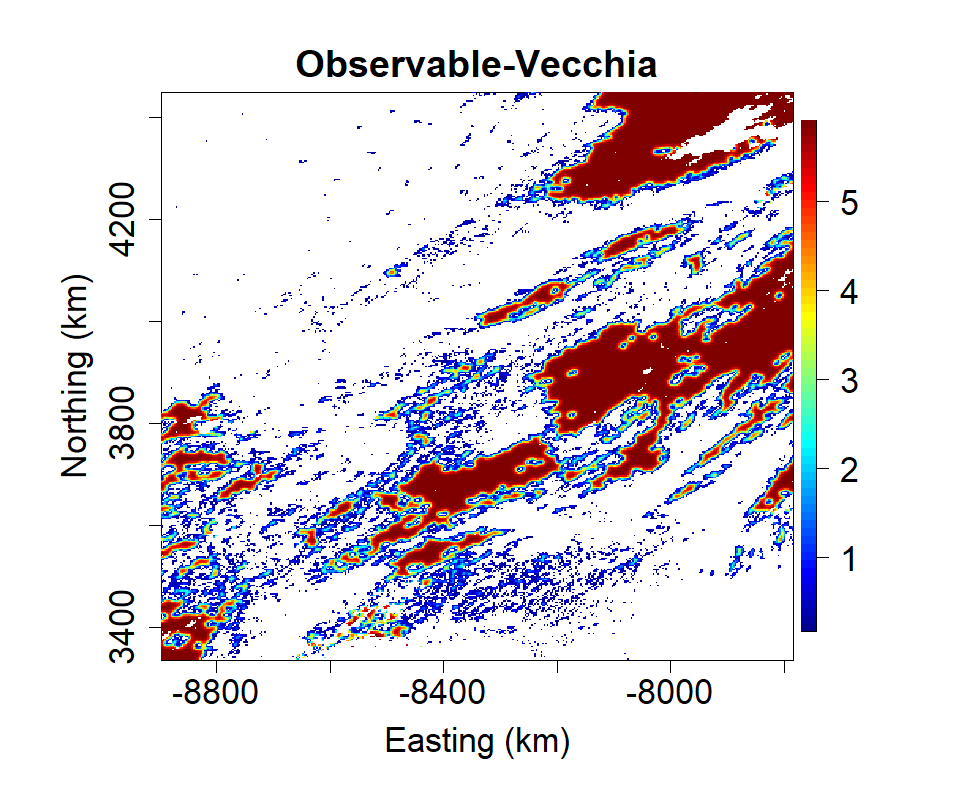}
     \end{subfigure}\caption{Predictive variances in the FSA (iterative and Cholesky-based), tapering, FITC, and Vecchia (latent and observable) framework.}
        \label{fig:four graphs2}
\end{figure}

Figures \ref{fig:four graphs} and \ref{fig:four graphs2} illustrate the predictive means and variances of different covariance matrix approximation approaches and methodologies. These maps visually confirm the results from Table \ref{TableCG2}. We see that the FITC approximation's predictive mean is smooth and captures the large-scale spatial structure but misses the rough short-scale dependence. The tapering and Vecchia approximation are both unable to make good predictions at locations where there is no observed data nearby. Overall, the FSA predicts the spatial structure more accurately than tapering, FITC, and the Vecchia approximation.

Additionally, Table \ref{TableCG} in Appendix \ref{realwapp} presents the results for various convergence tolerance levels and sample vectors. We find that, even when employing a higher convergence tolerance and fewer sample vectors, we lose almost no accuracy in parameter estimates and predictive distributions with an increased speedup. 

\section{Conclusion}
We propose to use iterative methods for doing inference with full-scale approximations. In particular, we introduce a novel and efficient preconditioner, called FITC preconditioner, which leverages the structure of the full-scale approximation. In our theoretical analysis, we demonstrate that the application of the proposed preconditioner considerably mitigates the sensitivity of the CG method's convergence rate with respect to the FSA parameters (number of inducing points $m$ and taper range $\gamma$), as well as the eigenvalue structure of the original covariance matrix. In a comprehensive simulation study, we demonstrate the superior efficiency of the FITC preconditioner compared to the state-of-the-art pivoted Cholesky preconditioner, when both using an FSA and a Vecchia approximation. Additionally, we highlight its utility as a control variate reducing variance in the stochastic approximation of the gradient of the negative log-likelihood. Furthermore, we introduce a novel simulation-based method for calculating predictive variances and illustrate its superiority in accuracy and speed over the state-of-the-art Lanczos method. In both simulated and real-world data experiments, we find that our proposed methodology achieves the same level of accuracy as Cholesky-based computations with a substantial reduction in computational time of approximately one order of magnitude. 

Future work could involve implementing the methods presented in this paper on GPUs. Additionally, an extension of the FSA framework to the multiresolution approximation \citep{katzfuss2017multi} could be valuable, offering a more flexible and efficient approach for hierarchical spatial modeling. Future research can also generalize our convergence theory in Theorems \ref{thm1} and \ref{th2} to other inducing points selection approaches such as the kmeans++ and cover tree algorithms. 


\section*{Acknowledgments}
This research was partially supported by the Swiss Innovation Agency - Innosuisse (grant number `57667.1 IP-ICT').

\clearpage
\begin{appendices}

\section{Preconditioned conjugate gradient algorithm}\label{appendix:CGalgo} Preconditioned conjugate gradient algorithm with Lanczos tridiagonal matrix:
\begin{algorithm}[H]
\caption{Preconditioned conjugate gradient algorithm}
    \begin{algorithmic}[1]
    \small
        \Require Matrix $\mathbf{A}$, preconditioner matrix $\mathbf{P}$, vector $\mathbf{b}$, L, tolerance
        \Ensure $\mathbf{u}_{l+1} \approx \mathbf{A}^{-1}\mathbf{b}$, tridiagonal matrix $\tilde{\mathbf{T}}$
        \State{early-stopping $\gets$ false}
        \State{$\alpha_0 \gets 1$}
        \State{$\beta_0 \gets 0$}
        \State{$\mathbf{u}_0 \gets \boldsymbol{0}$}
        \State{$\mathbf{r}_0 \gets \mathbf{b}$}
        \State{$\mathbf{z}_0 \gets \mathbf{P}^{-1}\mathbf{r}_0$}
        \State{$\mathbf{h}_0 \gets \mathbf{z}_0$}
        \For{$l \gets 0$ to $L$}
            \State{$\mathbf{v}_l \gets \mathbf{A}\mathbf{h}_l$}
            \State{$\alpha_{l+1} \gets \frac{\mathbf{r}_l^\mathrm{T}\mathbf{z}_l}{\mathbf{h}_l^\mathrm{T}\mathbf{v}_l}$}
            \State{$\mathbf{u}_{l+1} \gets \mathbf{u}_l + \alpha_{l+1} \mathbf{h}_l$}
            \State{$\mathbf{r}_{l+1} \gets \mathbf{r}_l - \alpha_{l+1} \mathbf{v}_l$}
            \If{$||\mathbf{r}_{l+1}||_2 <$ tolerance}
                \State{early-stopping $\gets$ true}
            \EndIf
            \State{$\mathbf{z}_{l+1} \gets \mathbf{P}^{-1}\mathbf{r}_{l+1}$}
            \State{$\beta_{l+1} \gets \frac{\mathbf{r}_{l+1}^\mathrm{T}\mathbf{z}_{l+1}}{\mathbf{r}_l^\mathrm{T}\mathbf{z}_l}$}
            \State{$\mathbf{h}_{l+1} \gets \mathbf{z}_{l+1} + \beta_{l+1} \mathbf{h}_l$}
            \State{$\tilde{\mathbf{T}}_{l+1,l+1} \gets \frac{1}{\alpha_{l+1}} + \frac{\beta_l}{\alpha_{l}}$}
            \If{$l > 0$}
                \State{$\tilde{\mathbf{T}}_{l,l+1}, \tilde{\mathbf{T}}_{l+1,l} \gets \frac{\sqrt{\beta_l}}{\alpha_{l}}$}
            \EndIf
            \If{early-stopping}
                \State{return $\mathbf{u}_{l+1}, \tilde{\mathbf{T}}$}
            \EndIf
        \EndFor
    \end{algorithmic}
\end{algorithm}

\section{Predictive (co-)variances within full-scale approximation}\label{AppVar}

The predictive covariance matrix is given by
\begin{align*}
\mathbf{\Sigma}^p_\dagger&={\mathbf{\Sigma}}_{{n_p}}^\dagger + \sigma^2 \boldsymbol{I}_{n_p}-({\mathbf{\Sigma}}_{n{n_p}}^\dagger)^{\mathrm{T}}\Tilde{\mathbf{\Sigma}}^{-1}_\dagger{\mathbf{\Sigma}}_{n{n_p}}^\dagger\\
&= {\mathbf{\Sigma}}_{{n_p}}^\dagger + \sigma^2 \boldsymbol{I}_{n_p}-({\mathbf{\Sigma}}_{n{n_p}}^\mathrm{s})^\mathrm{T}\Tilde{\mathbf{\Sigma}}_{\mathrm{s}}^{-1}{\mathbf{\Sigma}}_{n{n_p}}^\mathrm{s}\\
&\quad-({\mathbf{\Sigma}}_{mn}^\mathrm{T}{\mathbf{\Sigma}}_{m}^{-1}{\mathbf{\Sigma}}_{m{n_p}})^\mathrm{T}\Tilde{\mathbf{\Sigma}}_{\mathrm{s}}^{-1}{\mathbf{\Sigma}}_{mn}^\mathrm{T}{\mathbf{\Sigma}}_{m}^{-1}{\mathbf{\Sigma}}_{m{n_p}}\\
&\quad-({\mathbf{\Sigma}}_{n{n_p}}^\mathrm{s})^\mathrm{T}\Tilde{\mathbf{\Sigma}}_{\mathrm{s}}^{-1}{\mathbf{\Sigma}}_{mn}^\mathrm{T}{\mathbf{\Sigma}}_{m}^{-1}{\mathbf{\Sigma}}_{m{n_p}}\\
&\quad - \big(({\mathbf{\Sigma}}_{n{n_p}}^\mathrm{s})^\mathrm{T}\Tilde{\mathbf{\Sigma}}_{\mathrm{s}}^{-1}{\mathbf{\Sigma}}_{mn}^\mathrm{T}{\mathbf{\Sigma}}_{m}^{-1}{\mathbf{\Sigma}}_{m{n_p}}\big)^\mathrm{T}\\
&\quad+({\mathbf{\Sigma}}_{n{n_p}}^\mathrm{s})^\mathrm{T}\Tilde{\mathbf{\Sigma}}_{\mathrm{s}}^{-1} \mathbf{\Sigma}_{mn}^{\mathrm{T}}\boldsymbol{M}^{-1} \mathbf{\Sigma}_{mn} \Tilde{\mathbf{\Sigma}}_{\mathrm{s}}^{-1}{\mathbf{\Sigma}}_{n{n_p}}^\mathrm{s}\\
&\quad +({\mathbf{\Sigma}}_{mn}^\mathrm{T}{\mathbf{\Sigma}}_{m}^{-1}{\mathbf{\Sigma}}_{m{n_p}})^\mathrm{T}\Tilde{\mathbf{\Sigma}}_{\mathrm{s}}^{-1} \mathbf{\Sigma}_{mn}^{\mathrm{T}}\boldsymbol{M}^{-1}\mathbf{\Sigma}_{mn} \Tilde{\mathbf{\Sigma}}_{\mathrm{s}}^{-1}{\mathbf{\Sigma}}_{mn}^\mathrm{T}{\mathbf{\Sigma}}_{m}^{-1}{\mathbf{\Sigma}}_{m{n_p}}\\
&\quad + ({\mathbf{\Sigma}}_{n{n_p}}^\mathrm{s})^\mathrm{T}\Tilde{\mathbf{\Sigma}}_{\mathrm{s}}^{-1} \mathbf{\Sigma}_{mn}^{\mathrm{T}}\boldsymbol{M}^{-1} \mathbf{\Sigma}_{mn} \Tilde{\mathbf{\Sigma}}_{\mathrm{s}}^{-1}{\mathbf{\Sigma}}_{mn}^\mathrm{T}{\mathbf{\Sigma}}_{m}^{-1}{\mathbf{\Sigma}}_{m{n_p}}\\
&\quad + \big(({\mathbf{\Sigma}}_{n{n_p}}^\mathrm{s})^\mathrm{T}\Tilde{\mathbf{\Sigma}}_{\mathrm{s}}^{-1} \mathbf{\Sigma}_{mn}^{\mathrm{T}}\boldsymbol{M}^{-1} \mathbf{\Sigma}_{mn} \Tilde{\mathbf{\Sigma}}_{\mathrm{s}}^{-1}{\mathbf{\Sigma}}_{mn}^\mathrm{T}{\mathbf{\Sigma}}_{m}^{-1}{\mathbf{\Sigma}}_{m{n_p}}\big)^\mathrm{T}.
\end{align*}
Presuming the Cholesky factors of $\Tilde{\mathbf{\Sigma}}_{\mathrm{s}}$ and $\boldsymbol{M}=\mathbf{\Sigma}_{m}+\mathbf{\Sigma}_{mn} \Tilde{\mathbf{\Sigma}}_{\mathrm{s}}^{-1} \mathbf{\Sigma}_{mn}^{\mathrm{T}}$ are precomputed during training, the
costs for computing the single terms are
\begin{enumerate}
    \item $({\mathbf{\Sigma}}_{n{n_p}}^\mathrm{s})^\mathrm{T}\Tilde{\mathbf{\Sigma}}_{\mathrm{s}}^{-1}{\mathbf{\Sigma}}_{n{n_p}}^\mathrm{s}:\quad \mathcal{O}(n\cdot n_p \cdot n_\gamma + n\cdot n_p^2)$,
    \item $({\mathbf{\Sigma}}_{mn}^\mathrm{T}{\mathbf{\Sigma}}_{m}^{-1}{\mathbf{\Sigma}}_{m{n_p}})^\mathrm{T}\Tilde{\mathbf{\Sigma}}_{\mathrm{s}}^{-1}{\mathbf{\Sigma}}_{mn}^\mathrm{T}{\mathbf{\Sigma}}_{m}^{-1}{\mathbf{\Sigma}}_{m{n_p}}:\quad \mathcal{O}\big(m\cdot n_\gamma\cdot n + m^2\cdot n +m^2\cdot n_p + m\cdot n_p^2\big)$,
    \item $({\mathbf{\Sigma}}_{n{n_p}}^\mathrm{s})^\mathrm{T}\Tilde{\mathbf{\Sigma}}_{\mathrm{s}}^{-1}{\mathbf{\Sigma}}_{mn}^\mathrm{T}{\mathbf{\Sigma}}_{m}^{-1}{\mathbf{\Sigma}}_{m{n_p}}:\quad \mathcal{O}\big(m\cdot n_\gamma\cdot n + m^2\cdot n_p + m\cdot {n}_\gamma^p\cdot n_p + m\cdot n_p^2\big)$,
    \item $({\mathbf{\Sigma}}_{n{n_p}}^\mathrm{s})^\mathrm{T}\Tilde{\mathbf{\Sigma}}_{\mathrm{s}}^{-1} \mathbf{\Sigma}_{mn}^{\mathrm{T}}\boldsymbol{M}^{-1} \mathbf{\Sigma}_{mn} \Tilde{\mathbf{\Sigma}}_{\mathrm{s}}^{-1}{\mathbf{\Sigma}}_{n{n_p}}^\mathrm{s}:\\\mathcal{O}\big(m\cdot n_\gamma\cdot n + m^2\cdot n  +m\cdot {n}_\gamma^p\cdot n_p + m\cdot n_p^2\big)$,
    \item $({\mathbf{\Sigma}}_{mn}^\mathrm{T}{\mathbf{\Sigma}}_{m}^{-1}{\mathbf{\Sigma}}_{m{n_p}})^\mathrm{T}\Tilde{\mathbf{\Sigma}}_{\mathrm{s}}^{-1} \mathbf{\Sigma}_{mn}^{\mathrm{T}}\boldsymbol{M}^{-1} \mathbf{\Sigma}_{mn} \Tilde{\mathbf{\Sigma}}_{\mathrm{s}}^{-1}{\mathbf{\Sigma}}_{mn}^\mathrm{T}{\mathbf{\Sigma}}_{m}^{-1}{\mathbf{\Sigma}}_{m{n_p}}:\\
\mathcal{O}\big(m\cdot n_\gamma\cdot n + m^2\cdot n + m^2\cdot n_p + m\cdot n_p^2\big)$,
\item $({\mathbf{\Sigma}}_{n{n_p}}^\mathrm{s})^\mathrm{T}\Tilde{\mathbf{\Sigma}}_{\mathrm{s}}^{-1} \mathbf{\Sigma}_{mn}^{\mathrm{T}}\boldsymbol{M}^{-1} \mathbf{\Sigma}_{mn} \Tilde{\mathbf{\Sigma}}_{\mathrm{s}}^{-1}{\mathbf{\Sigma}}_{mn}^\mathrm{T}{\mathbf{\Sigma}}_{m}^{-1}{\mathbf{\Sigma}}_{m{n_p}}:\\
\mathcal{O}\big(m\cdot n_\gamma\cdot n + m\cdot {n}_\gamma^p\cdot n_p + m^2\cdot n + m^2\cdot n_p + m\cdot n_p^2\big)$.
\end{enumerate}
Therefore, the computational complexity for computing the predictive covariance is
\begin{align*}
    \mathcal{O}\big(m\cdot n_\gamma\cdot n + m\cdot {n}_\gamma^p\cdot n_p + m^2\cdot n + m^2\cdot n_p + n\cdot n_p \cdot n_\gamma + n\cdot n_p^2\big),
\end{align*}
and for the predictive variances is
\begin{align*}
    \mathcal{O}\big(m\cdot n_\gamma\cdot n + m\cdot {n}_\gamma^p\cdot n_p + m^2\cdot n + m^2\cdot n_p + n\cdot n_p \cdot n_\gamma\big).
\end{align*}
For the case when $n_p > n$, we obtain
\begin{align*}
    \mathcal{O}\big(n_p\cdot(m\cdot {n}_\gamma^p + m^2 + n\cdot n_p)\big) \qquad \text{and} \quad \mathcal{O}\big(n_p\cdot(m\cdot {n}_\gamma^p + m^2 + n \cdot n_\gamma)\big),
\end{align*}
respectively. And for the case when $n_p < n$, we obtain
\begin{align*}
    \mathcal{O}\big(n\cdot(m\cdot {n}_\gamma + m^2 + n_p \cdot n_\gamma + n_p^2)\big) \qquad \text{and} \quad \mathcal{O}\big(n\cdot(m\cdot {n}_\gamma + m^2 + n_p \cdot n_\gamma)\big),
\end{align*}
respectively.

\section{Fisher scoring}\label{Fishapp}

The Fisher information matrix $\boldsymbol{I} \in \mathbb{R}^{q \times q}$ for Fisher scoring is given by
$$
(\boldsymbol{I})_{k l}=\frac{1}{2} \operatorname{Tr}\left(\Tilde{\mathbf{\Sigma}}^{-1}_{\dagger}\frac{\partial \Tilde{\mathbf{\Sigma}}_{\dagger}}{\partial \boldsymbol{\theta}_k}\Tilde{\mathbf{\Sigma}}^{-1}_{\dagger}\frac{\partial \Tilde{\mathbf{\Sigma}}_{\dagger}}{\partial \boldsymbol{\theta}_l}\right), \quad 1 \leq k, l \leq q.
$$
These trace terms can be computed by using STE
\begin{align*}
    \operatorname{Tr}\left(\Tilde{\mathbf{\Sigma}}^{-1}_{\dagger}\frac{\partial \Tilde{\mathbf{\Sigma}}_{\dagger}}{\partial \boldsymbol{\theta}_k}\Tilde{\mathbf{\Sigma}}^{-1}_{\dagger}\frac{\partial \Tilde{\mathbf{\Sigma}}_{\dagger}}{\partial \boldsymbol{\theta}_l}\right)\approx\frac{1}{\ell}\sum_{i=1}^{\ell} \big(\boldsymbol{z}_i^\mathrm{T} \Tilde{\mathbf{\Sigma}}^{-1}_{\dagger}\frac{\partial \Tilde{\mathbf{\Sigma}}_{\dagger}}{\partial \boldsymbol{\theta}_k}\big)\big(\Tilde{\mathbf{\Sigma}}^{-1}_{\dagger}\frac{\partial \Tilde{\mathbf{\Sigma}}_{\dagger}}{\partial \boldsymbol{\theta}_l} \boldsymbol{z}_i\big), \quad 1 \leq k, l \leq q,
\end{align*}
where the linear solves $\Tilde{\mathbf{\Sigma}}^{-1}_{\dagger}\boldsymbol{z}_i$ and $\Tilde{\mathbf{\Sigma}}^{-1}_{\dagger}\frac{\partial \Tilde{\mathbf{\Sigma}}_{\dagger}}{\partial \boldsymbol{\theta}_l} \boldsymbol{z}_i$ can be computed by the CG method or in the Cholesky-based variant using the Sherman-Woodbury-Morrison formula and the Cholesky factor of $\Tilde{\mathbf{\Sigma}}_\mathrm{s}$.

\section{Approximate predictive variances using simulation}\label{radapp}

Note that we use Rademacher random vectors with entries $\pm 1$ because the stochastic approximation of the diagonal of a matrix $\mathbf{A}\in\mathbb{R}^{n\times n}$ with Gaussian random vectors $\boldsymbol{z}_i$ is given by
\begin{align*}
    \text{diag}(\mathbf{A})\approx \Big[\sum_{i=1}^\ell \boldsymbol{z}_i \circ \mathbf{A} \boldsymbol{z}_i\Big] \oslash \Big[\sum_{i=1}^\ell \boldsymbol{z}_i \circ \boldsymbol{z}_i\Big],
\end{align*}
where the operator $\oslash$ refers to the elementwise division (Hadamard division).

Therefore, the normalization of a single sample $\boldsymbol{z}_i \circ \mathbf{A} \boldsymbol{z}_i$ by all sample vectors $\boldsymbol{z}_i$ renders it impractical to compute the optimal scaling factor $\hat{\boldsymbol{c}}_\text{opt}$ for applying variance reduction as introduced in the following subsection.

\begin{proof}[Proof of Proposition \ref{PropPredVar}]\label{proofalgo}
By standard results, the only non-deterministic term \newline $1/\ell\sum_{i=1}^\ell \boldsymbol{z}_i^{(1)} \circ ({\mathbf{\Sigma}}_{n{n_p}}^\mathrm{s})^\mathrm{T}\Tilde{\mathbf{\Sigma}}_{\mathrm{s}}^{-1}{\mathbf{\Sigma}}_{n{n_p}}^\mathrm{s}\boldsymbol{z}_i^{(1)}$ in Algorithm \ref{alg:pred_var} is an unbiased and consistent estimate of $\text{diag}\big(({\mathbf{\Sigma}}_{n{n_p}}^\mathrm{s})^\mathrm{T}\Tilde{\mathbf{\Sigma}}_{\mathrm{s}}^{-1}{\mathbf{\Sigma}}_{n{n_p}}^\mathrm{s}\big)$, and the claim in Proposition \ref{PropPredVar} thus follows.
\end{proof}

\section{Preconditioner as control variate}\label{contvarapp}
\begin{align}\label{Troptc}
\begin{split}
    \Tr \Big(\Tilde{\mathbf{\Sigma}}^{-1}_{\dagger}\frac{\partial \Tilde{\mathbf{\Sigma}}_{\dagger}}{\partial \boldsymbol{\theta}}\Big) &\approx \frac{1}{\ell} \sum_{i=1}^{\ell} \Big(\big(\boldsymbol{z}_i^\mathrm{T}\Tilde{\mathbf{\Sigma}}^{-1}_{\dagger}\big)\big(\frac{\partial\Tilde{\mathbf{\Sigma}}_{\dagger}}{\partial \boldsymbol{\theta}}\boldsymbol{P}^{-1} \boldsymbol{z}_i\big)-\hat{c}_{opt}\cdot\big(\boldsymbol{z}_i^\mathrm{T}\boldsymbol{P}^{-1}\big)\big(\frac{\partial\boldsymbol{P}}{\partial \boldsymbol{\theta}}\boldsymbol{P}^{-1} \boldsymbol{z}_i\big)\Big)\\
    &\quad + \hat{c}_{opt}\cdot\Tr \Big(\boldsymbol{P}^{-1}\frac{\partial \boldsymbol{P}}{\partial \boldsymbol{\theta}}\Big),
\end{split}
\end{align}
where 
\begin{align*}
    \hat{c}_{opt} = \frac{\sum\limits_{i=1}^\ell\big((\boldsymbol{z}_i^\mathrm{T}\Tilde{\mathbf{\Sigma}}^{-1}_{\dagger})(\frac{\partial\Tilde{\mathbf{\Sigma}}_{\dagger}}{\partial \boldsymbol{\theta}}\boldsymbol{P}^{-1} \boldsymbol{z}_i)-\Tilde{T}_\ell\big) (\boldsymbol{z}_i^\mathrm{T}\boldsymbol{P}^{-1})(\frac{\partial\boldsymbol{P}}{\partial \boldsymbol{\theta}}\boldsymbol{P}^{-1} \boldsymbol{z}_i)}{\sum\limits_{i=1}^\ell \big((\boldsymbol{z}_i^\mathrm{T}\boldsymbol{P}^{-1})(\frac{\partial\boldsymbol{P}}{\partial \boldsymbol{\theta}}\boldsymbol{P}^{-1} \boldsymbol{z}_i)\big)^2}.
\end{align*}

\begin{align*}
\text{diag}(({\mathbf{\Sigma}}_{n{n_p}}^\mathrm{s})^\mathrm{T}\Tilde{\mathbf{\Sigma}}_{\mathrm{s}}^{-1}{\mathbf{\Sigma}}_{n{n_p}}^\mathrm{s}) \approx &\frac{1}{\ell}\sum_{i=1}^\ell \Big(\boldsymbol{z}_i \circ \big(({\mathbf{\Sigma}}_{n{n_p}}^\mathrm{s})^\mathrm{T}\Tilde{\mathbf{\Sigma}}_{\mathrm{s}}^{-1}{\mathbf{\Sigma}}_{n{n_p}}^\mathrm{s}\boldsymbol{z}_i\big)\\
&-\hat{\boldsymbol{c}}_\text{opt}\circ\boldsymbol{z}_i \circ\big(({\mathbf{\Sigma}}_{n{n_p}}^\mathrm{s})^\mathrm{T}\boldsymbol{P}_{\mathrm{s}}^{-1}{\mathbf{\Sigma}}_{n{n_p}}^\mathrm{s}\boldsymbol{z}_i\big)\Big)\\
&+\hat{\boldsymbol{c}}_\text{opt}\circ\text{diag}\big(({\mathbf{\Sigma}}_{n{n_p}}^\mathrm{s})^\mathrm{T}\boldsymbol{P}_{\mathrm{s}}^{-1}{\mathbf{\Sigma}}_{n{n_p}}^\mathrm{s}\big),
\end{align*}
where 
\begin{align*}
 \hat{\boldsymbol{c}}_\text{opt} = &\Big(\sum\limits_{i=1}^\ell\big(\boldsymbol{z}_i \circ ({\mathbf{\Sigma}}_{n{n_p}}^\mathrm{s})^\mathrm{T}\Tilde{\mathbf{\Sigma}}_{\mathrm{s}}^{-1}{\mathbf{\Sigma}}_{n{n_p}}^\mathrm{s}\boldsymbol{z}_i-{\boldsymbol{D}}_\ell\big)\big(\boldsymbol{z}_i \circ ({\mathbf{\Sigma}}_{n{n_p}}^\mathrm{s})^\mathrm{T}\boldsymbol{P}_{\mathrm{s}}^{-1}{\mathbf{\Sigma}}_{n{n_p}}^\mathrm{s}\boldsymbol{z}_i\big)\Big)\\
 &\oslash\Big(\sum\limits_{i=1}^\ell\big(\boldsymbol{z}_i \circ ({\mathbf{\Sigma}}_{n{n_p}}^\mathrm{s})^\mathrm{T}\boldsymbol{P}_{\mathrm{s}}^{-1}{\mathbf{\Sigma}}_{n{n_p}}^\mathrm{s}\boldsymbol{z}_i\big)^2\Big).   
\end{align*}

\section{Runtime of the FITC preconditioner}\label{AppRT}

\textbf{Time complexity of computing $\log \det\big(\hat{\boldsymbol{P}}\big)$, $\Tr\big(\hat{\boldsymbol{P}}^{-1}\frac{\partial\hat{\boldsymbol{P}}}{\partial \boldsymbol{\theta}}\big)$ and $\hat{\boldsymbol{P}}^{-1}\boldsymbol{y}$:}\\
To compute the log determinant of the preconditioner, we make use of Sylvester's determinant theorem \citep{sylvester1851xxxvii}
\begin{align*}
    \det\big(\hat{\boldsymbol{P}}\big) &=\det\big(\boldsymbol{D}_{\mathrm{s}} + \mathbf{\Sigma}_{mn}^{\mathrm{T}}\mathbf{\Sigma}_{m}^{-1}\mathbf{\Sigma}_{mn}\big)\\
    &=\det\big(\mathbf{\Sigma}_{m}+\mathbf{\Sigma}_{mn}\boldsymbol{D}_{\mathrm{s}}^{-1}\mathbf{\Sigma}_{mn}^{\mathrm{T}}\big)\cdot\det\big(\mathbf{\Sigma}_{m}\big)^{-1}\cdot\det\big(\boldsymbol{D}_{\mathrm{s}}\big).
\end{align*}
After calculating $\mathbf{\Sigma}_{m}+\mathbf{\Sigma}_{mn} \boldsymbol{D}_{\mathrm{s}}^{-1}\mathbf{\Sigma}_{mn}^{\mathrm{T}}\in\mathbb{R}^{m\times m}$ in $\mathcal{O}(n\cdot m^2)$, we obtain its Cholesky factor in $\mathcal{O}(m^3)$. As a result, the computation of $\log\det\big(\hat{\boldsymbol{P}}\big)$ requires $\mathcal{O}(n\cdot m^2)$ time. For computing the trace $\Tr\big(\hat{\boldsymbol{P}}^{-1}\frac{\partial\hat{\boldsymbol{P}}}{\partial \boldsymbol{\theta}}\big)$, we can proceed similar.

To compute linear solves with the preconditioner, we make use of the Sherman-Woodbury-Morrison formula \citep{woodbury1950inverting}
\begin{align*}
    \hat{\boldsymbol{P}}^{-1} \boldsymbol{y}&=\big(\boldsymbol{D}_{\mathrm{s}} + \mathbf{\Sigma}_{mn}^{\mathrm{T}}\mathbf{\Sigma}_{m}^{-1}\mathbf{\Sigma}_{mn}\big)^{-1} \boldsymbol{y}\\
    &=\boldsymbol{D}_{\mathrm{s}}^{-1}\boldsymbol{y}-\boldsymbol{D}_{\mathrm{s}}^{-1} \mathbf{\Sigma}_{mn}^{\mathrm{T}}\left(\mathbf{\Sigma}_{m}+\mathbf{\Sigma}_{mn} \boldsymbol{D}_{\mathrm{s}}^{-1}\mathbf{\Sigma}_{mn}^{\mathrm{T}}\right)^{-1} \mathbf{\Sigma}_{mn}\boldsymbol{D}_{\mathrm{s}}^{-1} \boldsymbol{y}.
\end{align*}
After computing $\boldsymbol{D}_{\mathrm{s}}^{-1}\boldsymbol{y}$ in linear time $\mathcal{O}(n)$, the operation $\mathbf{\Sigma}_{mn}\boldsymbol{D}_{\mathrm{s}}^{-1} \boldsymbol{y}$ necessitates a time complexity of $\mathcal{O}(n\cdot m)$. Subsequently, when we use $\mathbf{\Sigma}_{m}+\mathbf{\Sigma}_{mn} \boldsymbol{D}_{\mathrm{s}}^{-1}\mathbf{\Sigma}_{mn}^{\mathrm{T}}\in\mathbb{R}^{m\times m}$ and its Cholesky factor computed already for the log determinant, performing a linear solve using this factor requires $\mathcal{O}(m^2)$ time. As a result, following these precomputations, each solving operation with the preconditioner $\hat{\boldsymbol{P}}$ incurs a total time complexity of $\mathcal{O}(n \cdot m)$. \\
\\
\textbf{Time complexity of drawing samples from $\mathcal{N}(\boldsymbol{0}, \hat{\boldsymbol{P}})$:} \\
We draw samples from $\mathcal{N}(\boldsymbol{0}, \hat{\boldsymbol{P}})$ via the reparameterization trick used by \cite{gardner2018gpytorch} in Appendix C.1. If $\boldsymbol{\epsilon}_1 \in \mathbb{R}^m$ and $\boldsymbol{\epsilon}_2 \in \mathbb{R}^n$ are standard normal vectors, then $\big(\mathbf{\Sigma}_{mn}^{\mathrm{T}}\mathbf{\Sigma}_{m}^{-\frac{1}{2}} \boldsymbol{\epsilon}_1+ \boldsymbol{D}_{\mathrm{s}}^{\frac{1}{2}}\boldsymbol{\epsilon}_2\big)$, is a sample from $\mathcal{N}(\boldsymbol{0},\boldsymbol{D}_{\mathrm{s}} + \mathbf{\Sigma}_{mn}^{\mathrm{T}}\mathbf{\Sigma}_{m}^{-1}\mathbf{\Sigma}_{mn})$, where $\mathbf{\Sigma}_{m}^{\frac{1}{2}}$ is the Cholesky factor of $\mathbf{\Sigma}_{m}$ and $\boldsymbol{D}_{\mathrm{s}}^{\frac{1}{2}}$ is the elementwise square-root of $\boldsymbol{D}_{\mathrm{s}}$. Assuming the Cholesky factor is already computed, sampling from $\mathcal{N}(\boldsymbol{0}, \hat{\boldsymbol{P}})$ requires matrix-vector multiplications for a total of $\mathcal{O}(n \cdot m)$ time.

\section{Convergence analysis of the CG method with and without the FITC preconditioner}\label{AppConv}
\begin{definition}\label{defO}
    \textbf{Big $\mathcal{O}$ in probability notation}
    
    For a set of random variables $X_i$ and a corresponding set of constants $a_i$ both indexed by $i$, the notation
$$
X_i=\mathcal{O}_P\left(a_i\right) 
$$
means that the set of values $X_i / a_i$ is stochastically bounded. That is, for any $\varepsilon>0$, there exists a finite $M>0$ and a finite $N>0$ such that
$$
\mathbb{P}\Bigg(\left|\frac{X_i}{a_i}\right|>M\Bigg)<\varepsilon, \forall i\geq N.
$$
\end{definition}
\begin{lemma}\label{lemnyst}
  \textbf{Error bound for Nyström low-rank approximation}
  
  Under Assumptions~\ref{assumpt1}--\ref{assumpt3}, let $\boldsymbol{\Sigma}_{mn}^\mathrm{T}\boldsymbol{\Sigma}^{-1}_m\boldsymbol{\Sigma}_{mn}$ be the rank-$m$ approximation of a covariance matrix $\boldsymbol{\Sigma}$ with eigenvalues $\lambda_1 \geq ... \geq \lambda_n> 0$ as described in Section \ref{sectFSA}. Then, the following inequality holds
\begin{align*}
    ||{\boldsymbol{\Sigma}} - {\boldsymbol{\Sigma}}_{mn}^\mathrm{T}{\boldsymbol{\Sigma}}^{-1}_m{\boldsymbol{\Sigma}}_{mn}||_2 \leq \mathcal{O}_P\Big(\lambda_{m+1} + \frac{n}{\sqrt{m}} \cdot\sigma^2_1\Big),
\end{align*}
where $\mathcal{O}_P$ is the $\mathcal{O}$-notation in probability.
\end{lemma}
\begin{proof}
    We start using Theorem 2 in \citet{kumar2012sampling}, which states that with probability at least $1 - \epsilon$, the following inequality holds for any $m$
    \begin{align*}
        ||{\boldsymbol{\Sigma}} - {\boldsymbol{\Sigma}}_{mn}^\mathrm{T}{\boldsymbol{\Sigma}}^{-1}_m{\boldsymbol{\Sigma}}_{mn}||_2&\leq ||{\boldsymbol{\Sigma}} - \hat{\boldsymbol{\Sigma}}_m||_2 \\
        &\quad+ \frac{2 \cdot n}{\sqrt{m}} \max\limits_{i}({{\Sigma}}_{ii})\Big(1+\sqrt{\frac{n-m}{n-\frac{1}{2}} \frac{\log\big(\frac{1}{\epsilon}\big)}{\beta(m, n)}}  \frac{d_{\max }^{\mathbf{\Sigma}}}{\max\limits_{i}({{L}}_{ii})}\Big),
    \end{align*}
    where $\hat{\boldsymbol{\Sigma}}_m$ is the best rank-$m$ approximation of $\boldsymbol{\Sigma}$ with respect to the spectral norm $||\cdot||_2$, $\beta(m,n) = 1-\frac{1}{2\cdot\max(m,n-m)}$, ${{L}}_{ii}$ is the $i$-th entry of the diagonal of the Cholesky factor $\boldsymbol{L}$ of $\boldsymbol{\Sigma}$ and $d_{\max }^{\mathbf{\Sigma}} = \max\limits_{ij}\sqrt{\Sigma_{ii}+\Sigma_{jj}-2\cdot \Sigma_{ij}}$.
    
    By the Eckart-Young-Mirsky theorem for the spectral norm \citep{eckart1936approximation}, we know that the best rank-$m$ approximation of $\boldsymbol{\Sigma}$ in the spectral norm is its $m$-truncated eigenvalue decomposition considering only the $m$ largest eigenvalues. Therefore, we obtain
    \begin{align*}
        ||{\boldsymbol{\Sigma}} - \hat{\boldsymbol{\Sigma}}_m||_2 \leq \lambda_{m+1},
    \end{align*}
    where ${\lambda}_{m+1}$ is the $(m+1)$-th largest eigenvalue of ${\boldsymbol{\Sigma}}$.
    
    Since ${{\Sigma}}_{ij} \geq 0$ for all $i,j = 1,...,n$ and $\boldsymbol{\Sigma}$ has by Assumption~\ref{assumpt3} a constant diagonal with entries equal to $\sigma^2_1$, we have $d_{\max }^{\mathbf{\Sigma}}\leq \sqrt{2\cdot{{\Sigma}}_{ii}} = \sqrt{2}\cdot\sqrt{\sigma_1^2}$. Moreover, by the Cholesky decomposition algorithm, we have $\max\limits_{i}({{L}}_{ii}) =  \sqrt{\sigma_1^2}$. Hence, we obtain
    \begin{align*}
        ||{\boldsymbol{\Sigma}} - {\boldsymbol{\Sigma}}_{mn}^\mathrm{T}{\boldsymbol{\Sigma}}^{-1}_m{\boldsymbol{\Sigma}}_{mn}||_2
        &\leq \lambda_{m+1}+ \frac{2 \cdot n}{\sqrt{m}} \cdot\sigma^2_1\cdot\Big(1+\sqrt{\frac{n-m}{n-\frac{1}{2}} \frac{\log\big(\frac{1}{\epsilon}\big)}{\beta(m, n)}}  \frac{\sqrt{2}\cdot\sqrt{\sigma_1^2}}{\sqrt{\sigma_1^2}}\Big)\\
        &= \lambda_{m+1}+ \frac{2 \cdot n}{\sqrt{m}} \cdot\sigma^2_1\cdot\Big(1+\sqrt{2\cdot\frac{n-m}{n-\frac{1}{2}} \frac{\log\big(\frac{1}{\epsilon}\big)}{\beta(m, n)}}\Big).
    \end{align*}
    Since $\max(m,n-m)\geq \frac{n}{2}$, we have $\frac{1}{\beta(m,n)} = \frac{1}{1-\frac{1}{2\cdot \max(m,n-m)}}\leq\frac{1}{1-\frac{1}{n}}=\frac{n}{n-1}\leq 2$. 
    Hence, we obtain
    \begin{align*}
        ||{\boldsymbol{\Sigma}} - {\boldsymbol{\Sigma}}_{mn}^\mathrm{T}{\boldsymbol{\Sigma}}^{-1}_m{\boldsymbol{\Sigma}}_{mn}||_2
        &\leq \lambda_{m+1}+ \frac{2 \cdot n}{\sqrt{m}} \cdot\sigma^2_1\cdot\Big(1+\sqrt{4\cdot\frac{n-m}{n-\frac{1}{2}} \log\big(\frac{1}{\epsilon}\big)}\Big)\\
        &\leq \lambda_{m+1}+ \frac{2 \cdot n}{\sqrt{m}} \cdot\sigma^2_1\cdot\Big(1+\sqrt{\frac{8}{3}\cdot (n-m)\cdot \log\big(\frac{1}{\epsilon}\big)}\Big).
    \end{align*}
    
    To use the $\mathcal{O}$-notation in probability in Definition \ref{defO}, we have to show that there exists an $N$ such that for all $m\geq N$ it holds that $(n-m)\cdot \log\big(\frac{1}{\epsilon}\big)$ is finite.
    
    We obtain for a constant $C>0$
    \begin{align*}
        (n-m)\cdot \log\big(\frac{1}{\epsilon}\big) \leq C \Longleftrightarrow m \geq n-\frac{C}{\log\big(\frac{1}{\epsilon}\big)},
    \end{align*}
    and therefore, $N= \max\Big(n-\frac{C}{\log\big(\frac{1}{\epsilon}\big)},1\Big)$.
    
    Using the following equivalence
    \begin{equation*}
        \begin{gathered}
    \mathbb{P}\Bigg(\Bigg|\frac{||{\boldsymbol{\Sigma}} - {\boldsymbol{\Sigma}}_{mn}^\mathrm{T}{\boldsymbol{\Sigma}}^{-1}_m{\boldsymbol{\Sigma}}_{mn}||_2}{\lambda_{m+1}+ \frac{2 \cdot n}{\sqrt{m}} \cdot\sigma^2_1\cdot\Big(1+\sqrt{\frac{8}{3}\cdot (n-m)\cdot \log\big(\frac{1}{\epsilon}\big)}\Big)}\Bigg|\leq 1\Bigg) > 1-\epsilon\\ \Longleftrightarrow\\
     \mathbb{P}\Bigg(\Bigg|\frac{||{\boldsymbol{\Sigma}} - {\boldsymbol{\Sigma}}_{mn}^\mathrm{T}{\boldsymbol{\Sigma}}^{-1}_m{\boldsymbol{\Sigma}}_{mn}||_2}{\lambda_{m+1}+ \frac{2 \cdot n}{\sqrt{m}} \cdot\sigma^2_1\cdot\Big(1+\sqrt{\frac{8}{3}\cdot (n-m)\cdot \log\big(\frac{1}{\epsilon}\big)}\Big)}\Bigg|> 1\Bigg) < \epsilon,
     \end{gathered}
    \end{equation*}
     we have for all $m \geq \max\Big(n-\frac{C}{\log\big(\frac{1}{\epsilon}\big)},1\Big)$ that
     \begin{align*}
         \epsilon &> \mathbb{P}\Bigg(\Bigg|\frac{||{\boldsymbol{\Sigma}} - {\boldsymbol{\Sigma}}_{mn}^\mathrm{T}{\boldsymbol{\Sigma}}^{-1}_m{\boldsymbol{\Sigma}}_{mn}||_2}{\lambda_{m+1}+ \frac{2 \cdot n}{\sqrt{m}} \cdot\sigma^2_1\cdot\Big(1+\sqrt{\frac{8}{3}\cdot (n-m)\cdot \log\big(\frac{1}{\epsilon}\big)}\Big)}\Bigg|> 1\Bigg)\\ & \geq \mathbb{P}\Bigg(\Bigg|\frac{||{\boldsymbol{\Sigma}} - {\boldsymbol{\Sigma}}_{mn}^\mathrm{T}{\boldsymbol{\Sigma}}^{-1}_m{\boldsymbol{\Sigma}}_{mn}||_2}{\lambda_{m+1}+ \frac{2 \cdot n}{\sqrt{m}} \cdot\sigma^2_1\cdot\Big(1+\sqrt{\frac{8}{3}\cdot C}\Big)}\Bigg|> 1\Bigg)\\ &  \geq \mathbb{P}\Bigg(\Bigg|\frac{||{\boldsymbol{\Sigma}} - {\boldsymbol{\Sigma}}_{mn}^\mathrm{T}{\boldsymbol{\Sigma}}^{-1}_m{\boldsymbol{\Sigma}}_{mn}||_2}{\lambda_{m+1}+ \frac{n}{\sqrt{m}} \cdot\sigma^2_1}\Bigg|> M\Bigg),
     \end{align*}
     with $M = 2\cdot\big(1+\sqrt{\frac{8}{3}\cdot C}\big)$. Hence, we obtain 
      \begin{align*}
        ||{\boldsymbol{\Sigma}} - {\boldsymbol{\Sigma}}_{mn}^\mathrm{T}{\boldsymbol{\Sigma}}^{-1}_m{\boldsymbol{\Sigma}}_{mn}||_2=\mathcal{O}_P\Big(\lambda_{m+1} + \frac{n}{\sqrt{m}} \cdot\sigma^2_1\Big).
    \end{align*}
\end{proof}
\begin{lemma}\label{lem1}
    \textbf{Bound for the Condition Number of $\Tilde{\boldsymbol{\Sigma}}_{\dagger}$}
    
    Let $\Tilde{\boldsymbol{\Sigma}}_\dagger \in \mathbb{R}^{n \times n}$ be the FSA of a covariance matrix $\Tilde{\boldsymbol{\Sigma}} =\boldsymbol{\Sigma} + \sigma^2\boldsymbol{I}_n  \in \mathbb{R}^{n \times n}$ with $m$ inducing points and taper range $\gamma$, where ${\boldsymbol{\Sigma}}$ has eigenvalues $\lambda_1 \geq ... \geq \lambda_n> 0$. Under Assumptions~\ref{assumpt1}--\ref{assumpt3}, the condition number $\kappa( \Tilde{\boldsymbol{\Sigma}}_\dagger)$ is bounded by
\begin{align*}
\kappa(\Tilde{\boldsymbol{\Sigma}}_\dagger)\leq  \mathcal{O}_P\bigg(\frac{1}{\sigma^2}\cdot\Big(\big(\lambda_{m+1} + \frac{n}{\sqrt{m}}\cdot\sigma_1^2\big)\cdot\sqrt{n\cdot n_\gamma}+\lambda_1\Big)\bigg) + 1,
\end{align*}
where $n_\gamma$ is the average number of nonzero entries per row in $\mathbf{T}(\gamma)$.
\end{lemma}
\begin{proof}
We use the identity
\begin{align*}
\kappa(\Tilde{\boldsymbol{\Sigma}}_\dagger) & = ||\Tilde{\boldsymbol{\Sigma}}_\dagger||_2\cdot||\Tilde{\boldsymbol{\Sigma}}_\dagger^{-1}||_2,
\end{align*}
and bound $||\Tilde{\boldsymbol{\Sigma}}_\dagger^{-1}||_2$ by
\begin{align*}
    ||\Tilde{\boldsymbol{\Sigma}}_\dagger^{-1}||_2 = \frac{1}{\lambda^\dagger_n + \sigma^2} \leq \frac{1}{\sigma^2},
\end{align*}
where $\lambda^\dagger_n$ is the smallest eigenvalue of $(\boldsymbol{\Sigma} - \boldsymbol{\Sigma}_{mn}^\mathrm{T}\boldsymbol{\Sigma}^{-1}_m\boldsymbol{\Sigma}_{mn})\circ \boldsymbol{T}(\gamma)  + \boldsymbol{\Sigma}_{mn}^\mathrm{T}\boldsymbol{\Sigma}^{-1}_m\boldsymbol{\Sigma}_{mn}>0$.

Therefore, we obtain
\begin{align*}
    \kappa(\Tilde{\boldsymbol{\Sigma}}_\dagger) & = ||\Tilde{\boldsymbol{\Sigma}}_\dagger||_2\cdot||\Tilde{\boldsymbol{\Sigma}}_\dagger^{-1}||_2\leq \frac{1}{\sigma^2}\cdot||\Tilde{\boldsymbol{\Sigma}}_\dagger||_2\\
    &= \frac{1}{\sigma^2}\cdot||(\boldsymbol{\Sigma} - \boldsymbol{\Sigma}_{mn}^\mathrm{T}\boldsymbol{\Sigma}^{-1}_m\boldsymbol{\Sigma}_{mn})\circ \boldsymbol{T}(\gamma)  + \boldsymbol{\Sigma}_{mn}^\mathrm{T}\boldsymbol{\Sigma}^{-1}_m\boldsymbol{\Sigma}_{mn} + \sigma^2\boldsymbol{I}_n||_2\\
    &\leq \frac{1}{\sigma^2}\cdot||(\boldsymbol{\Sigma} - \boldsymbol{\Sigma}_{mn}^\mathrm{T}\boldsymbol{\Sigma}^{-1}_m\boldsymbol{\Sigma}_{mn})\circ \boldsymbol{T}(\gamma)||_2  + ||\boldsymbol{\Sigma}_{mn}^\mathrm{T}\boldsymbol{\Sigma}^{-1}_m\boldsymbol{\Sigma}_{mn}||_2 + ||\sigma^2\boldsymbol{I}_n||_2\\
    &= \frac{1}{\sigma^2}\cdot\Big(||(\boldsymbol{\Sigma} - \boldsymbol{\Sigma}_{mn}^\mathrm{T}\boldsymbol{\Sigma}^{-1}_m\boldsymbol{\Sigma}_{mn})\circ \boldsymbol{T}(\gamma)||_2  + ||\boldsymbol{\Sigma}_{mn}^\mathrm{T}\boldsymbol{\Sigma}^{-1}_m\boldsymbol{\Sigma}_{mn}||_2 + \sigma^2\Big)\\
    &= \frac{1}{\sigma^2}\cdot\Big(||({\boldsymbol{\Sigma}} - {\boldsymbol{\Sigma}}_{mn}^\mathrm{T}{\boldsymbol{\Sigma}}^{-1}_m{\boldsymbol{\Sigma}}_{mn})\circ \boldsymbol{T}(\gamma)||_2  + ||{\boldsymbol{\Sigma}}_{mn}^\mathrm{T}{\boldsymbol{\Sigma}}^{-1}_m{\boldsymbol{\Sigma}}_{mn}||_2 \Big)+ 1,
\end{align*}
where we used the triangle inequality.

Since $({\boldsymbol{\Sigma}} - {\boldsymbol{\Sigma}}_{mn}^\mathrm{T}{\boldsymbol{\Sigma}}^{-1}_m{\boldsymbol{\Sigma}}_{mn})$ and $ \boldsymbol{T}(\gamma)$ are symmetric and positive semidefinite matrices we can apply Theorem 5.3.4 in \citet{horn1991topics} to obtain
\begin{align*}
    ||({\boldsymbol{\Sigma}} - {\boldsymbol{\Sigma}}_{mn}^\mathrm{T}{\boldsymbol{\Sigma}}^{-1}_m{\boldsymbol{\Sigma}}_{mn})\circ \boldsymbol{T}(\gamma)||_2 &\leq ||{\boldsymbol{\Sigma}} - {\boldsymbol{\Sigma}}_{mn}^\mathrm{T}{\boldsymbol{\Sigma}}^{-1}_m{\boldsymbol{\Sigma}}_{mn}||_2\cdot ||\boldsymbol{T}(\gamma)||_2.
\end{align*}
By Lemma \ref{lemnyst}, we obtain
\begin{align*}
    ||{\boldsymbol{\Sigma}} - {\boldsymbol{\Sigma}}_{mn}^\mathrm{T}{\boldsymbol{\Sigma}}^{-1}_m{\boldsymbol{\Sigma}}_{mn}||_2 &\leq \mathcal{O}_P\Big(\lambda_{m+1} + \frac{n}{\sqrt{m}}\cdot\sigma_1^2\Big). 
\end{align*}

By using $||\cdot||_2\leq||\cdot||_F$ and $\frac{1}{n}\cdot\sum\limits_{i=1}^n\sum\limits_{j\text{ : }\mathbf{T}(\gamma)_{i,j}\neq 0}^n1 = n_\gamma$, we obtain
\begin{align*}
    || \boldsymbol{T}(\gamma)||_2&\leq|| \mathbf{T}(\gamma)||_F = \sqrt{\sum_{i=1}^n\sum_{j=1}^n\big|[\mathbf{T}(\gamma)]_{i,j}\big|^2}=\sqrt{\sum_{i=1}^n\sum_{j\text{ : }\mathbf{T}(\gamma)_{i,j}\neq 0}^n\big|[\mathbf{T}(\gamma)]_{i,j}\big|^2}\\
    &\leq \sqrt{\sum_{i=1}^n\sum_{j\text{ : }\mathbf{T}(\gamma)_{i,j}\neq 0}^n1} = \sqrt{n\cdot n_\gamma}.
\end{align*}
Moreover, using the Loewner order $\preceq$, see, e.g., Chapter 7 in \citet{jahn2020introduction}, we know that $ 0 \preceq({\boldsymbol{\Sigma}} - {\boldsymbol{\Sigma}}_{mn}^\mathrm{T}{\boldsymbol{\Sigma}}^{-1}_m{\boldsymbol{\Sigma}}_{mn})$ since $({\boldsymbol{\Sigma}} - {\boldsymbol{\Sigma}}_{mn}^\mathrm{T}{\boldsymbol{\Sigma}}^{-1}_m{\boldsymbol{\Sigma}}_{mn})$ is positive semidefinite and therefore, we have ${\boldsymbol{\Sigma}}_{mn}^\mathrm{T}{\boldsymbol{\Sigma}}^{-1}_m{\boldsymbol{\Sigma}}_{mn}\preceq{\boldsymbol{\Sigma}} $ and obtain
\begin{align*}
    ||{\boldsymbol{\Sigma}}_{mn}^\mathrm{T}{\boldsymbol{\Sigma}}^{-1}_m{\boldsymbol{\Sigma}}_{mn}||_2 \leq ||{\boldsymbol{\Sigma}}||_2 = {\lambda}_1.
\end{align*}

In conclusion, we obtain
\begin{align*}
    \kappa\big(\Tilde{\boldsymbol{\Sigma}}_\dagger\big) &\leq  \frac{1}{\sigma^2}\cdot\Big(||({\boldsymbol{\Sigma}} - {\boldsymbol{\Sigma}}_{mn}^\mathrm{T}{\boldsymbol{\Sigma}}^{-1}_m{\boldsymbol{\Sigma}}_{mn})\circ \boldsymbol{T}(\gamma)||_2  + ||{\boldsymbol{\Sigma}}_{mn}^\mathrm{T}{\boldsymbol{\Sigma}}^{-1}_m{\boldsymbol{\Sigma}}_{mn}||_2 \Big)+ 1\\
    &\leq \mathcal{O}_P\bigg(\frac{1}{\sigma^2}\cdot\Big(\big(\lambda_{m+1} + \frac{n}{\sqrt{m}}\cdot\sigma_1^2\big)\cdot\sqrt{n\cdot n_\gamma}+\lambda_1\Big)\bigg) + 1.
\end{align*}
\end{proof}

\begin{proof}[Proof of Theorem \ref{thm1}]\label{proof1}
    We begin by stating a well-known CG convergence result \citep{trefethen2022numerical}, which bounds the error in terms of the
conditioning number
\begin{align*}
        \frac{\left\|\mathbf{u}^*-\mathbf{u}_k\right\|_{\Tilde{\boldsymbol{\Sigma}}_\dagger}}{\left\|\mathbf{u}^*-\mathbf{u}_0\right\|_{\Tilde{\boldsymbol{\Sigma}}_\dagger}} &\leq 2\cdot\bigg(\frac{\sqrt{\kappa(\Tilde{\boldsymbol{\Sigma}}_\dagger)}-1}{\sqrt{\kappa(\Tilde{\boldsymbol{\Sigma}}_\dagger)}+1}\bigg)^k.
\end{align*}
By using Lemma \ref{lem1}, we obtain
\begin{align*}
        \frac{\left\|\mathbf{u}^*-\mathbf{u}_k\right\|_{\Tilde{\boldsymbol{\Sigma}}_\dagger}}{\left\|\mathbf{u}^*-\mathbf{u}_0\right\|_{\Tilde{\boldsymbol{\Sigma}}_\dagger}} &\leq 2\cdot \left(\frac{\sqrt{\kappa(\Tilde{\boldsymbol{\Sigma}}_\dagger)}-1}{\sqrt{\kappa(\Tilde{\boldsymbol{\Sigma}}_\dagger)}+1}\right)^k\\
        &\leq 2\cdot \left(\frac{\sqrt{\mathcal{O}_P\bigg(\frac{1}{\sigma^2}\cdot\Big(\big(\lambda_{m+1} + \frac{n}{\sqrt{m}}\cdot\sigma_1^2\big)\cdot\sqrt{n\cdot n_\gamma}+\lambda_1\Big)\bigg)+1}-1}{\sqrt{\mathcal{O}_P\bigg(\frac{1}{\sigma^2}\cdot\Big(\big(\lambda_{m+1} + \frac{n}{\sqrt{m}}\cdot\sigma_1^2\big)\cdot\sqrt{n\cdot n_\gamma}+\lambda_1\Big)\bigg)+1}+1}\right)^k\\
        &\leq 2\cdot \left(\frac{\sqrt{\mathcal{O}_P\bigg(\frac{1}{\sigma^2}\cdot\Big(\big(\lambda_{m+1} + \frac{n}{\sqrt{m}}\cdot\sigma_1^2\big)\cdot\sqrt{n\cdot n_\gamma}+\lambda_1\Big)\bigg)}+1-1}{\sqrt{\mathcal{O}_P\bigg(\frac{1}{\sigma^2}\cdot\Big(\big(\lambda_{m+1} + \frac{n}{\sqrt{m}}\cdot\sigma_1^2\big)\cdot\sqrt{n\cdot n_\gamma}+\lambda_1\Big)\bigg)}+1}\right)^k\\
        &\leq 2\cdot\left(\frac{1}{1 + \mathcal{O}_P\bigg(\sigma\cdot\Big(\big(\lambda_{m+1} + \frac{n}{\sqrt{m}}\cdot\sigma_1^2\big)\cdot\sqrt{n\cdot n_\gamma}+\lambda_1\Big)^{-\frac{1}{2}}\bigg)}\right)^k.
\end{align*}
\end{proof}

\begin{lemma}\label{lem2}
    \textbf{Bound for the Condition Number of $\widehat{\boldsymbol{P}}^{-1} \Tilde{\boldsymbol{\Sigma}}_\dagger$} 
    
    Let $\Tilde{\boldsymbol{\Sigma}}_\dagger \in \mathbb{R}^{n \times n}$ be the FSA of a covariance matrix $\Tilde{\boldsymbol{\Sigma}} =\boldsymbol{\Sigma} + \sigma^2\boldsymbol{I}_n  \in \mathbb{R}^{n \times n}$ depending on the number of inducing points $m$ and taper range $\gamma$, where ${\boldsymbol{\Sigma}}$ has eigenvalues $\lambda_1 \geq ... \geq \lambda_n> 0$. Furthermore, let $\widehat{\boldsymbol{P}}\in \mathbb{R}^{n \times n}$ be the FITC preconditioner for $\Tilde{\boldsymbol{\Sigma}}_\dagger$. Under Assumptions~\ref{assumpt1}--\ref{assumpt3}, the condition number $\kappa(\widehat{\boldsymbol{P}}^{-1} \Tilde{\boldsymbol{\Sigma}}_\dagger)$ is bounded by
\begin{align*}
\kappa(\widehat{\boldsymbol{P}}^{-1} \Tilde{\boldsymbol{\Sigma}}_\dagger) \leq \bigg(1+\mathcal{O}_P\Big(\frac{1}{\sigma^2}\cdot\big(\lambda_{m+1} + \frac{n}{\sqrt{m}}\cdot\sigma_1^2\big)\cdot \sqrt{n\cdot (n_\gamma-1)}\Big)\bigg)^2,
\end{align*}
where $n_\gamma$ is the average number of nonzero entries per row in $\mathbf{T}(\gamma)$.
\end{lemma}

\begin{proof}
For $\boldsymbol{E} = \Tilde{\boldsymbol{\Sigma}}_\dagger - \widehat{\boldsymbol{P}}$, we obtain
\begin{align*}
\kappa(\widehat{\boldsymbol{P}}^{-1} \Tilde{\boldsymbol{\Sigma}}_\dagger) &= ||\widehat{\boldsymbol{P}}^{-1} \Tilde{\boldsymbol{\Sigma}}_\dagger||_2 \cdot || \Tilde{\boldsymbol{\Sigma}}_\dagger^{-1}\widehat{\boldsymbol{P}}||_2 \\
&= ||\widehat{\boldsymbol{P}}^{-1} (\widehat{\boldsymbol{P}} + \boldsymbol{E})||_2 \cdot || \Tilde{\boldsymbol{\Sigma}}_\dagger^{-1}(\Tilde{\boldsymbol{\Sigma}}_\dagger - \boldsymbol{E})||_2 \\
&= || \boldsymbol{I}_n + \widehat{\boldsymbol{P}}^{-1} \boldsymbol{E}||_2 \cdot || \boldsymbol{I}_n - \Tilde{\boldsymbol{\Sigma}}_\dagger^{-1}\boldsymbol{E}||_2.
\end{align*}
Applying Cauchy-Schwarz and the triangle inequality, we have
\begin{align*}
\kappa(\widehat{\boldsymbol{P}}^{-1} \Tilde{\boldsymbol{\Sigma}}_\dagger) &\leq \big(1+||\widehat{\boldsymbol{P}}^{-1} ||_2\cdot||\boldsymbol{E}||_2\big)\cdot\big(1+||\Tilde{\boldsymbol{\Sigma}}_\dagger^{-1}||_2\cdot|| \boldsymbol{E}||_2\big). 
\end{align*}
Furthermore, similar to the first step in the proof of Lemma \ref{lem1}, we have $||\widehat{\boldsymbol{P}}^{-1} ||_2\leq\frac{1}{\sigma^2}$ and $||\Tilde{\boldsymbol{\Sigma}}_\dagger^{-1}||_2\leq\frac{1}{\sigma^2}$. Therefore, we obtain
\begin{align*}
\kappa(\widehat{\boldsymbol{P}}^{-1} \Tilde{\boldsymbol{\Sigma}}_\dagger) &\leq \big(1+\frac{1}{\sigma^2}\cdot||\boldsymbol{E}||_2\big)^2 = \big(1+\frac{1}{\sigma^2}\cdot||\Tilde{\boldsymbol{\Sigma}}_{\mathrm{s}}-\text{diag}(\Tilde{\boldsymbol{\Sigma}}_{\mathrm{s}})||_2\big)^2\\
& = \big(1+\frac{1}{\sigma^2}\cdot||(\boldsymbol{\Sigma} - \boldsymbol{\Sigma}_{mn}^\mathrm{T}\boldsymbol{\Sigma}^{-1}_m\boldsymbol{\Sigma}_{mn})\circ \mathbf{T}(\gamma)\\
&\quad-(\boldsymbol{\Sigma} - \boldsymbol{\Sigma}_{mn}^\mathrm{T}\boldsymbol{\Sigma}^{-1}_m\boldsymbol{\Sigma}_{mn})\circ\text{diag}\big(\mathbf{T}(\gamma)\big)||_2\big)^2\\
& = \Big(1+\frac{1}{\sigma^2}\cdot||(\boldsymbol{\Sigma} - \boldsymbol{\Sigma}_{mn}^\mathrm{T}\boldsymbol{\Sigma}^{-1}_m\boldsymbol{\Sigma}_{mn})\circ \Big(\mathbf{T}(\gamma)-\text{diag}\big(\mathbf{T}(\gamma)\big)\Big)||_2\Big)^2.
\end{align*}
Since $(\boldsymbol{\Sigma} - \boldsymbol{\Sigma}_{mn}^\mathrm{T}\boldsymbol{\Sigma}^{-1}_m\boldsymbol{\Sigma}_{mn})$ is positive semidefinite and symmetric and additionally $\mathbf{T}(\gamma)-\text{diag}\big(\mathbf{T}(\gamma)\big)$ is symmetric we can apply Theorem 5.3.4 in \citet{horn1991topics} to obtain
\begin{align*}
    \kappa(\widehat{\boldsymbol{P}}^{-1} \Tilde{\boldsymbol{\Sigma}}_\dagger) &\leq \Big(1+\frac{1}{\sigma^2}\cdot||(\boldsymbol{\Sigma} - \boldsymbol{\Sigma}_{mn}^\mathrm{T}\boldsymbol{\Sigma}^{-1}_m\boldsymbol{\Sigma}_{mn})\circ \Big(\mathbf{T}(\gamma)-\text{diag}\big(\mathbf{T}(\gamma)\big)\Big)||_2\Big)^2\\
    &\leq \Big(1+\frac{1}{\sigma^2}\cdot||\boldsymbol{\Sigma} - \boldsymbol{\Sigma}_{mn}^\mathrm{T}\boldsymbol{\Sigma}^{-1}_m\boldsymbol{\Sigma}_{mn}||_2\cdot ||\Big(\mathbf{T}(\gamma)-\text{diag}\big(\mathbf{T}(\gamma)\big)\Big)||_2\Big)^2.
\end{align*}
Moreover, using the same arguments as in the proof of Lemma \ref{lem1}, we obtain
\begin{align*}
\kappa(\widehat{\boldsymbol{P}}^{-1} \Tilde{\boldsymbol{\Sigma}}_\dagger) &\leq\bigg(1+\mathcal{O}_P\Big(\frac{1}{\sigma^2}\cdot\big(\lambda_{m+1} + \frac{n}{\sqrt{m}}\cdot\sigma_1^2\big)\cdot \sqrt{n\cdot (n_\gamma-1)}\Big)\bigg)^2.
\end{align*}
\end{proof}

\begin{proof}[Proof of Theorem \ref{th2}]\label{proof2}
By proceeding analogue to the proof of Theorem \ref{thm1} and by using Lemma \ref{lem2}, we obtain
\begin{align*}
        \frac{\left\|\mathbf{u}^*-\mathbf{u}_k\right\|_{\Tilde{\boldsymbol{\Sigma}}_\dagger}}{\left\|\mathbf{u}^*-\mathbf{u}_0\right\|_{\Tilde{\boldsymbol{\Sigma}}_\dagger}} &
        \leq 2\cdot \left(\frac{\sqrt{\kappa(\widehat{\boldsymbol{P}}^{-1} \Tilde{\boldsymbol{\Sigma}}_\dagger)}-1}{\sqrt{\kappa(\widehat{\boldsymbol{P}}^{-1} \Tilde{\boldsymbol{\Sigma}}_\dagger)}+1}\right)^k\\
        &\leq 2\cdot \left(\frac{1+\mathcal{O}_P\Big(\frac{1}{\sigma^2}\cdot\big(\lambda_{m+1} + \frac{n}{\sqrt{m}}\cdot\sigma_1^2\big)\cdot \sqrt{n\cdot (n_\gamma-1)}\Big)-1}{1+\mathcal{O}_P\Big(\frac{1}{\sigma^2}\cdot\big(\lambda_{m+1} + \frac{n}{\sqrt{m}}\cdot\sigma_1^2\big)\cdot \sqrt{n\cdot (n_\gamma-1)}\Big)+1}\right)^k\\
        &= 2\cdot\left(\frac{\mathcal{O}_P\Big(\frac{1}{\sigma^2}\cdot\big(\lambda_{m+1} + \frac{n}{\sqrt{m}}\cdot\sigma_1^2\big)\cdot \sqrt{n\cdot (n_\gamma-1)}\Big)}{\mathcal{O}_P\Big(\frac{1}{\sigma^2}\cdot\big(\lambda_{m+1} + \frac{n}{\sqrt{m}}\cdot\sigma_1^2\big)\cdot \sqrt{n\cdot (n_\gamma-1)}\Big)+2}\right)^k\\
        &\leq2\cdot\left(\frac{1}{1 + \mathcal{O}_P\Big(\sigma^2\cdot\big((\lambda_{m+1} + \frac{n}{\sqrt{m}}\cdot\sigma_1^2)\cdot \sqrt{n\cdot (n_\gamma-1)}\big)^{-1}\Big)}\right)^k.
\end{align*}
\end{proof}

\section{Choice of full-scale approximation parameters}

In Figure \ref{fig:APP}, we report the negative log-likelihood evaluated at the data-generating parameters for different numbers of inducing points $m\in \{100,200,\dots,1'000\}$ and taper ranges $n_\gamma\in \{10,20,\dots,120\}$ when using the mentioned range parameters. As expected, in scenarios with a low effective range, characterized by minimal large-scale spatial dependence, the number of inducing points $m$ has a small effect on the negative log-likelihood, and the residual process captures most of the dependence already with a relatively small range $\gamma$. Conversely, in cases with a large effective range, the FSA needs a large number of inducing points $m$, and the taper range $\gamma$ is less important to achieve a good approximation. For our subsequent simulation studies, we use $m = 500$ and $\gamma = 0.016$ corresponding to $n_\gamma = 80$ as a reasonable balance between accuracy and computational efficiency.
\begin{figure}[H]
    \centering
    \includegraphics[width=\linewidth]{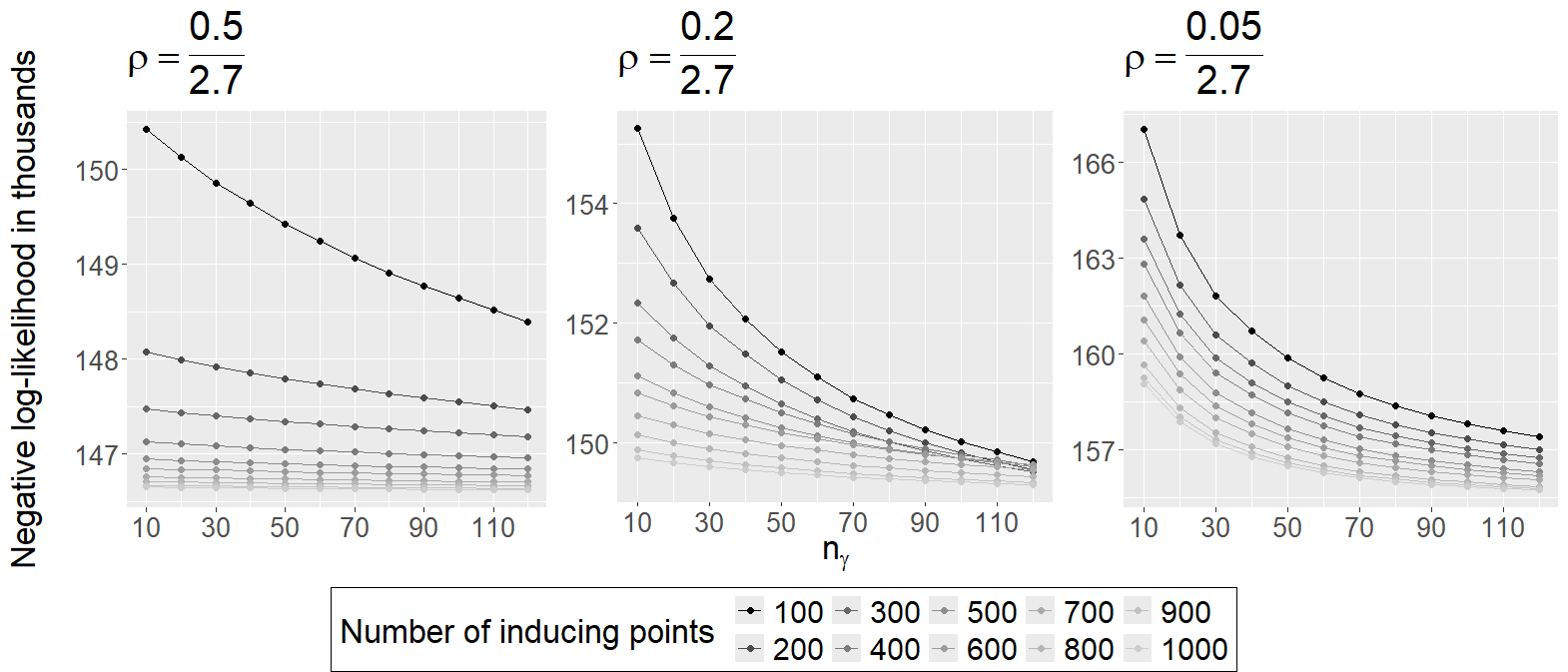}
    \caption{The negative log-likelihood for different combinations of the number of inducing points $m$ and the taper range $\gamma$ ($n_\gamma$) for different effective ranges (0.5, 0.2, 0.05 from left to right).}
    \label{fig:APP}
\end{figure}

\section{Methods for choosing inducing points}\label{AppIP} We compare the negative log-likelihood across different inducing point methods, evaluated at the true population parameters over 25 simulation iterations. All calculations are performed using the Cholesky decomposition, ensuring that the only source of randomness stems from the selection of inducing points. This setup allows us to isolate and analyze the variability introduced by the inducing point selection methods. As shown in Figure \ref{fig:FSA}, the results within the FSA framework closely mirror those observed for the FITC approximation.
\begin{figure}[!htbp]
    \centering
    \includegraphics[width=\linewidth]{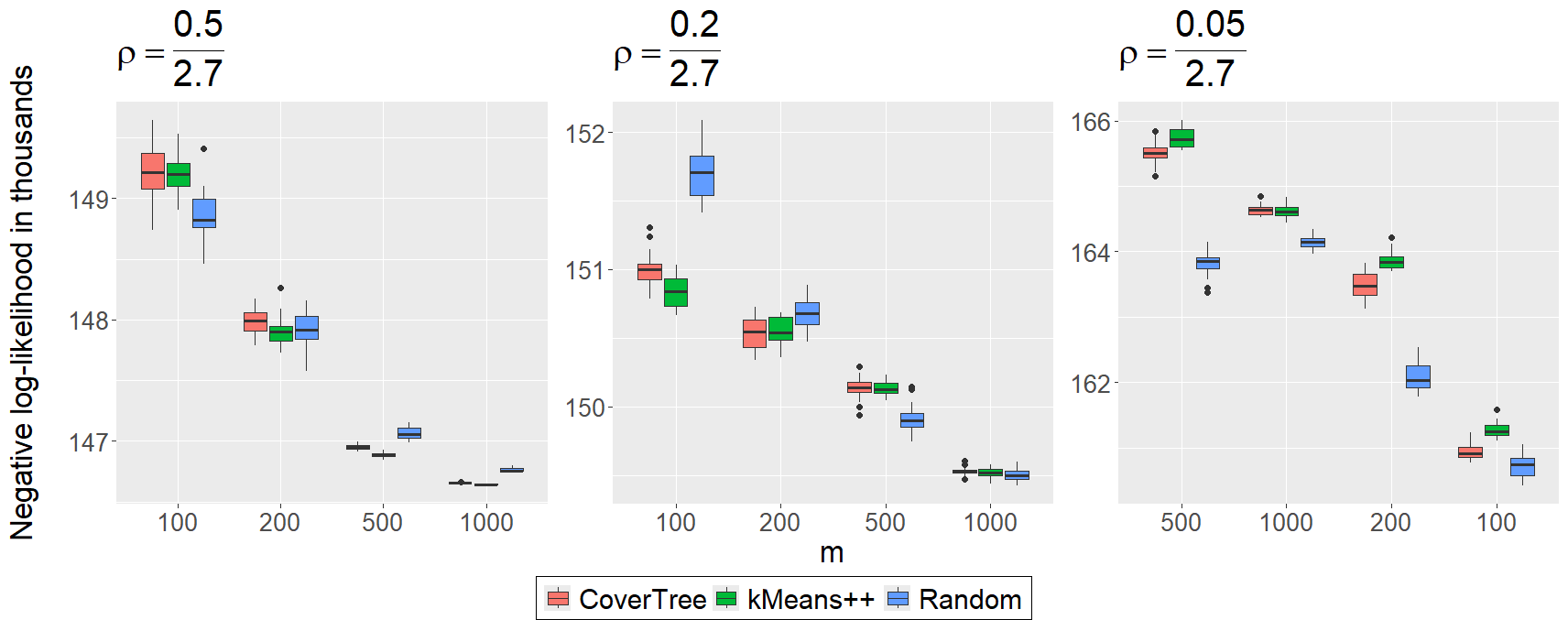}
    \caption{Box-plots of the negative log-likelihood for the FSA with $n_\gamma = 80$ for different effective ranges (0.5, 0.2, 0.05 from left to right) and numbers of inducing points $m$ ($n = 100'000$, $n_\gamma = 80$).}
    \label{fig:FSA}
\end{figure}

\section{Comparison of preconditioners}\label{simstudapp}

In Table \ref{Table1Precapp}, we observe that the CG method needs fewer iterations when the effective range is smaller. This observation aligns with the outcomes detailed in Theorem \ref{thm1}, as a diminished range parameter corresponds to a reduced largest eigenvalue of $\boldsymbol{\Sigma}$. For larger effective ranges, all preconditioners accelerate the convergence of the CG method.

\begin{table}[!htbp]
\centering
\begin{tabular}{ |p{9cm}||p{1.2cm}||p{1.2cm}||p{1.2cm}|  }
 \hline
\textbf{Effective Range}& 0.5 & 0.2 & 0.05\\
 \hhline{|=||=||=||=|}
 \textbf{No preconditioner} &     &&  \\
 CG-Iterations for $\Tilde{\boldsymbol{\Sigma}}_{\dagger}^{-1}\boldsymbol{y}$ & 342&279  & 74   \\
 Time for computing negative log-likelihood (s)& 94&54  & 14   \\ 
  
 \hline
 \textbf{FITC preconditioner}&   &  &    \\
 CG-Iterations for $\Tilde{\boldsymbol{\Sigma}}_{\dagger}^{-1}\boldsymbol{y}$ &  7 &9 & 21   \\
 Time for computing negative log-likelihood (s)& 7&10  & 15   \\ 
 \hline
 \textbf{Pivoted Cholesky} ($k=200$) & & &  \\
 CG-Iterations for $\Tilde{\boldsymbol{\Sigma}}_{\dagger}^{-1}\boldsymbol{y}$ & 47&91  & 44   \\
 Time for computing negative log-likelihood (s)& 35&52  & 33   \\ 
  
 \hline
 \textbf{Pivoted Cholesky} ($k=500$)  &  & &   \\
 CG-Iterations for $\Tilde{\boldsymbol{\Sigma}}_{\dagger}^{-1}\boldsymbol{y}$ & 26 &52  & 43  \\
 Time for computing negative log-likelihood (s)& 70&83  & 80   \\ 
 
 \hline
 \textbf{Pivoted Cholesky} ($k=1'000$)&  &   & \\
 CG-Iterations for $\Tilde{\boldsymbol{\Sigma}}_{\dagger}^{-1}\boldsymbol{y}$ & 13&32  & 42   \\
 Time for computing negative log-likelihood (s)&243 &266  & 273  \\ 
 
 \hline
 
\end{tabular}
\caption{Number of (preconditioned) CG-iterations for the linear solve $\Tilde{\boldsymbol{\Sigma}}_{\dagger}^{-1}\boldsymbol{y}$ and the time in seconds (s) for computing the negative log-likelihood for the true population parameters ($n = 100'000$, $m = 500$, $n_\gamma = 80$).}\label{Table1Precapp} 
\end{table}

\newpage

\section{Comparison of methods for calculating predictive variances}\label{AppPD} We present the results for all three range parameters.
\begin{table}[ht!]
\centering
\scriptsize	
\begin{tabular}{ |p{1.5cm}||p{1.cm}|p{1.3cm}||p{1.cm}|p{1.3cm}||p{1.cm}|p{1.3cm}|  }
 \hline
 &\multicolumn{2}{c||}{Effective range 0.5}& \multicolumn{2}{c||}{Effective range 0.2}& \multicolumn{2}{c|}{Effective range 0.05}\\
 \hline
 & Lanczos &Stochastic& Lanczos &Stochastic& Lanczos &Stochastic\\
 \hline
 $\boldsymbol{k,\ell = 50}$ &     &   & & & &\\
 Log-Score & 1.4588 & 1.4587 (1.7e-07) &1.4810 &1.4804 (1.0e-05) & 1.5541   &  1.5300 (4.8e-05) \\
 RMSE & 0.00018 & 4.1e-06 (2.4e-06)&0.0066 & 0.00050 (2.0e-05) &0.56& 0.042 (2.2e-05)\\
 Time (s) & 50& 51&53 & 56&57 &60\\
 \hline
 $\boldsymbol{k,\ell = 200}$ &     &   & & & & \\
 Log-Score &  1.4588 &1.4587 (7.4e-08) &1.4810 & 1.4804 (5.0e-06) &1.5535   &  1.5300 (5.9e-05)\\
 RMSE &  0.00015&8.9e-07 (5.0e-07) & 0.0065& 0.00020 (9.6e-06)& 0.55& 0.020 (7.2e-05)\\
 Time (s) & 53&57 & 57&60 &60 &64\\
 \hline
 $\boldsymbol{k,\ell = 500}$ &     &    & & & &\\
 Log-Score &  1.4588 & 1.4587 (4.1e-08)&1.4809 & 1.4804 (3.5e-06) &1.5520   &  1.5299 (8.7e-06) \\
 RMSE & 0.00013 & 6.3e-07 (3.4e-07)&0.0061 &0.00012 (7.4e-06) &0.54& 0.013 (1.3e-05)\\
 Time (s) & 62& 63& 67& 67& 72&72\\
 \hline
 $\boldsymbol{k,\ell = 1000}$ &     &    & & & &\\
 Log-Score &  1.4588 &1.4587 (2.8e-08) &1.4808 & 1.4804 (2.0e-06) &1.5512   &  1.5299 (6.6e-06) \\
 RMSE &  0.00011&5.8e-07 (1.8e-07) &0.0056 &  8.7e-05 (3.8e-06)&0.53& 0.0089 (1.2e-10)\\
 Time (s) & 120& 77& 124& 80& 128&87\\
 \hline
 $\boldsymbol{k,\ell = 2000}$ &     &  & & & &  \\
 Log-Score &   1.4587& 1.4587 (2.4e-08)& 1.4806& 1.4804 (1.5e-06) &1.5497   &  1.5299 (7.9e-06) \\
 RMSE &  8.6e-05&5.4e-07 (2.1e-07) &0.0047 & 6.2e-05 (2.4e-06)&0.50& 0.0063 (1.4e-05)\\
 Time (s) & 405& 122& 409& 126& 413&131\\
 \hline
 $\boldsymbol{k,\ell = 5000}$ &     &   & & & & \\
 Log-Score & 1.4587 &1.4587 (9.9e-09) & 1.4804& 1.4804 (8.7e-07)& 1.5459   &  1.5299 (5.8e-06) \\
 RMSE & 4.4e-05 &5.3e-07 (1.2e-07) & 0.0028& 3.7e-05 (1.9e-06) &0.42& 0.0040 (2.6e-06)\\
 Time (s) & 1452& 198& 1453& 203& 1468&210\\
 \hline
  \hline
 $\textbf{Cholesky:}$ &\multicolumn{2}{c||}{}& \multicolumn{2}{c||}{}& \multicolumn{2}{c|}{}\\
 Log-Score &\multicolumn{2}{c||}{1.4587}& \multicolumn{2}{c||}{1.4804}& \multicolumn{2}{c|}{1.5299}\\
 Time (s) &\multicolumn{2}{c||}{1651}& \multicolumn{2}{c||}{1652}& \multicolumn{2}{c|}{1658}\\
 \hline
\end{tabular}
\caption{Log-Score and RMSE of the predictive variance approximation computed by using the Cholesky factor, Lanczos methods, and stochastic diagonal approximations (mean and standard deviation) with different ranks $k$ and number of sample vectors $\ell$ ($n = 100'000$, $n_p = 100'000$, $m = 500$, $n_\gamma = 80$).}\label{table2} 
\end{table}
\clearpage

\section{Accuracy and computational time of parameter estimates and predictive distributions}\label{Appfull} We report the bias and RMSE of the estimated covariance parameters for each of the effective ranges 0.5, 0.2, and 0.05, as well as the averages and standard errors of the RMSE of the predictive mean, the log-score, and the CRPS.

 \begin{table}[ht!]
 \small
 \centering
 \begin{tabular}{ |p{.3cm}|p{1.7cm}||p{1.5cm}|p{1.5cm}||p{1.5cm}|p{1.5cm}||p{1.5cm}|p{1.5cm}|  }
  \hline
  & &\multicolumn{2}{l||}{Effective range 0.5}&\multicolumn{2}{l||}{Effective range 0.2}&\multicolumn{2}{l|}{Effective range 0.05}\\
 \hhline{|=|=||=|=||=|=||=|=|}
  & & \small{Cholesky} & \small{Iterative}
  & \small{Cholesky} & \small{Iterative} & \small{Cholesky} & \small{Iterative}\\
 \hhline{|=|=||=|=||=|=||=|=|}
   \multirow[c]{7}{*}[0in]{\rotatebox{90}{\parbox{3cm}{\textbf{ Parameter Est.}}}}& Bias of $\sigma^2$ & 0.0211&0.0215 & 0.0097  & 0.0097   &0.0135 &0.0147\\ 
  \cline{2-8} 
  & RMSE & 0.0218& 0.0222& 0.0125   & 0.0130    &0.0146 &0.0157\\
   \cline{2-8} 
  & Bias of $\sigma^2_1$  &0.0905 & 0.0519&0.0382   & 0.0408  &0.0356 & 0.0344 \\ 
   \cline{2-8} 
  & RMSE &0.161 &0.222 & 0.061 & 0.062 & 0.0394& 0.0384    \\
   \cline{2-8} 
   & Bias of $\rho$ &0.0629 & 0.0624& 0.0106 &0.0107 & 0.00946&0.00941 \\ 
    \cline{2-8}
   & RMSE &0.0630 & 0.0625& 0.0109&0.0109 & 0.00348 & 0.00343  \\
    \cline{2-8} 
   & Speedup& \multicolumn{2}{c||}{8.2}& \multicolumn{2}{c||}{7.5}& \multicolumn{2}{c|}{6.3} \\ 
  \hhline{|=|=||=|=||=|=||=|=|}
  \multirow[c]{7}{*}[0in]{\rotatebox{90}{\textbf{Prediction}}}
  &RMSE  & 1.0255&1.0255 &1.0491 & 1.0491& 1.1145& 1.1145\\
  \cline{2-8}
  &SE  &0.00047 &0.00047 & 0.00058&0.00059 &0.0010 & 0.0010  \\
  \cline{2-8}
  &Log-Score &1.4443 & 1.4443&1.4671 & 1.4671&1.5277 &1.5276\\
  \cline{2-8}
  &SE &0.00044 &0.00044 & 0.00063&0.00063 & 0.00089&0.00089 \\
  \cline{2-8}
  &CRPS & 0.5785& 0.5785&0.5916 & 0.5916& 0.6285&0.6285 \\
  \cline{2-8}
  &SE & 0.00024& 0.00024& 0.00039& 0.00039&0.00059 & 0.00059 \\
  \cline{2-8}
  &Speedup & \multicolumn{2}{c||}{32.4} & \multicolumn{2}{c||}{27.2} & \multicolumn{2}{c|}{22.6} \\ 
  \hhline{|=|=||=|=||=|=||=|=|}
  \multicolumn{2}{|l||}{\textbf{Speedup}} &\multicolumn{2}{c||}{11.6} & \multicolumn{2}{c||}{11.2} & \multicolumn{2}{c|}{8.9} \\ 
  \hline
 
 \end{tabular}
 \caption{Bias and RMSE of the estimated covariance parameters over 10 simulation runs and the averages and standard errors (SE) of the predictive measures (RMSE, log-score, CRPS) and the average speedup between the iterative method with the FITC preconditioner and the Cholesky-based computations ($n = 100'000$, $n_p = 100'000$, $m = 500$, $n_\gamma = 80$).}\label{Table5} 
\end{table}
 \clearpage

\clearpage

\section{Additional results for the FITC preconditioner for Vecchia approximations}\label{vecchia_app}
$ $

\begin{figure}[!htbp]
    \centering
    \includegraphics[width=\linewidth]{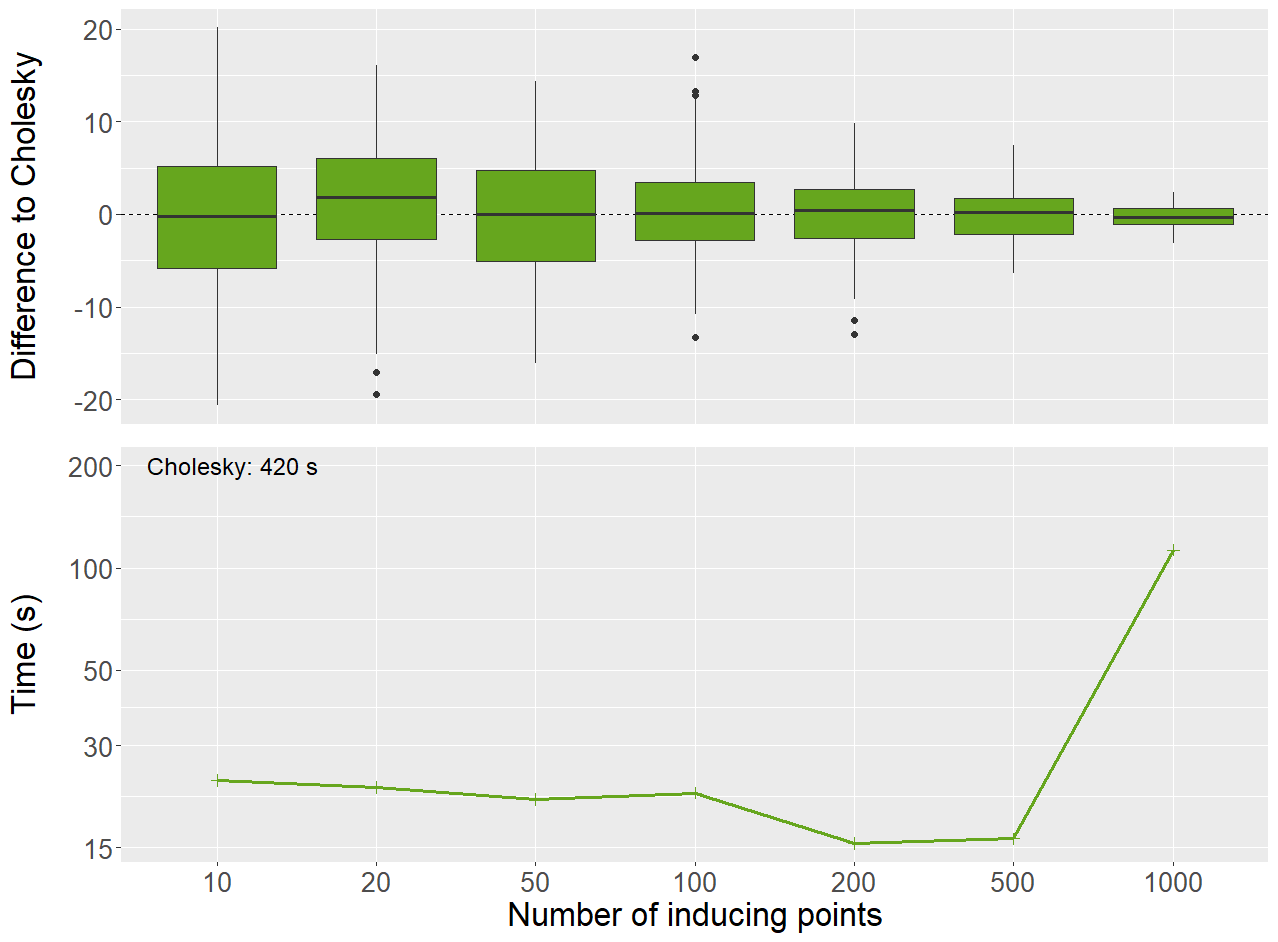}
    \caption{Differences of the log-marginal likelihood compared to Cholesky-based computations and runtime for the FITC preconditioner for different numbers of inducing points on simulated data with $n = 100'000$ and a range of $\rho = 0.05$.}
    \label{fig:FITC_vs_PC_005_rank}
\end{figure}
\clearpage

\section{Real-world application}\label{realwapp} 
$ $


\begin{table}[ht!]
\centering
\begin{tabular}{ |p{1.9cm}|p{1.5cm}|p{1.5cm}|p{1.5cm}|p{1.5cm}|p{1.5cm}|p{1.5cm}|p{1.5cm}|p{1.5cm}|  }
\hline
 $\delta$&\multicolumn{2}{c|}{0.001}& \multicolumn{2}{c|}{1}& \multicolumn{2}{c|}{10}\\
 \hline
 $\ell$&50& 20 & 50 & 20&  50 & 20\\
 \hhline{|=|=|=|=|=|=|=|}
 \textbf{Estim.}& & & & & &\\
 $\beta_{\text{intercept}}$ & -26.16 & -26.15&   -26.30& -26.42& -26.57&-27.87\\
 $\beta_{\text{east}}$ & -0.00659 & -0.00659 &  -0.00660& -0.00661& -0.00665&-0.00671\\
 $\beta_{\text{north}}$ & 0.000901 & 0.000897 &  0.000904 & 0.000912& 0.000882&0.000110\\
 $\sigma^2$ & 0.1428 & 0.1428 & 0.1428& 0.1427& 0.1424&0.1434\\
 $\sigma^2_1$& 4.103 &4.105 & 4.093& 4.091& 4.079&4.0328\\ 
  $\rho$& 24.16 & 24.17 & 24.13& 24.10& 24.03&23.96\\
  \hline
  Time (s)& 7678 & 5753 &  5054 & 3782& 4668& 3513\\ 
  Speedup &  4.3 & 5.8 & 6.6 & 8.8 & 7.1 & 9.5\\
 \hhline{|=|=|=|=|=|=|=|}
 \textbf{Predict.}&  & & & & &\\
 RMSE &  1.4175 &1.4174 & 1.4177& 1.4173& 1.4169&1.4250\\
 Log-Score &  1.7246 &1.7246 & 1.7244& 1.7242& 1.7236&1.7252\\
 CRPS &  0.75450 &0.75450 & 0.75450& 0.75423& 0.75400&0.75734\\
 \hline
 Time (s) &  1628 & 1630 &  587& 592& 392 &390\\
 Speedup & 3.4  & 3.4 & 9.4& 9.3 & 14.1 & 14.2\\
 \hhline{|=|=|=|=|=|=|=|}
 \textbf{Total} & & & & & &\\
 Time (s) &  9306 &  7383  & 5641&  4374 & 5060 & 3903\\
 Speedup &  4.2&  5.3 & 6.9 & 8.9 & 7.7& 9.9\\
 \hline
 
\end{tabular}
\caption{Results for GP inference within the FSA framework using the iterative method with the FITC preconditioner with different CG convergence tolerances $\delta\in\{0.001,1,10\}$ and numbers of sample vectors $\ell\in\{50,20\}$.}\label{TableCG} 
\end{table}
\end{appendices}
\clearpage

\bibliographystyle{abbrvnat}

\bibliography{references}

\end{document}